\documentclass[a4paper,USenglish,cleveref, autoref, thm-restate,numberwithinsect]{lipics-v2021}

\pdfoutput=1 

\author{Noam {Zilberstein}}{Cornell University}{noamz@cs.cornell.edu}{https://orcid.org/0000-0001-6388-063X}{}
\author{Daniele {Gorla}}{``Sapienza'' Universit\`a di Roma}{gorla@di.uniroma1.it}{https://orcid.org/0000-0001-8859-9844}{}
\author{Alexandra {Silva}}{Cornell University and University College London}{alexandra.silva@cornell.edu}{https://orcid.org/0000-0001-5014-9784}{}

\authorrunning{N. Zilberstein, D. Gorla, and A. Silva} 

\Copyright{Noam Zilberstein, Daniele Gorla, and Alexandra Silva} 

\ccsdesc[500]{Theory of computation~Parallel computing models}
\ccsdesc[500]{Theory of computation~Denotational semantics}
\ccsdesc[500]{Theory of computation~Probabilistic computation}

\keywords{Denotational Semantics, Pomsets, Concurrency, Convex Powerset} 

\category{} 

\relatedversion{}
\relatedversiondetails{Full Version}{https://doi.org/10.48550/arXiv.2503.02768}

\funding{This work was partially supported ARIA's Safeguarded AI programme.}

\nolinenumbers 

\usepackage[all,cmtip]{xy}
\usepackage{cite}
\usepackage[inline]{enumitem}
\usepackage{amsmath,amssymb,amsfonts,amsthm}
\usepackage{algorithmic}
\usepackage{graphicx}
\usepackage{textcomp}
\usepackage[dvipsnames]{xcolor}
\usepackage{mathpartir}
\usepackage{mathtools}
\usepackage{stmaryrd}
\usepackage{scalerel}
\usepackage{enumitem}
\usepackage{xr-hyper}
\usepackage{hyperref}
\hypersetup{
    colorlinks,
    linkcolor={red!50!black},
    citecolor={blue!50!black},
    urlcolor={blue!50!black}
}
\usepackage{multicol}
\usepackage{xr}

\def\BibTeX{{\rm B\kern-.05em{\sc i\kern-.025em b}\kern-.08em
    T\kern-.1667em\lower.7ex\hbox{E}\kern-.125emX}}

\def\extended{0}
\usepackage{olmacros}

\ifx\extended\undefined
\fi

\newcommand{\lpof}{\textsc{lpof}}
\newcommand{\dcpo}{\textsc{dcpo}}

\allowdisplaybreaks

\newcommand{\guardarrows}{
  \ar@<.05mm>@[NavyBlue]^{\Tarr}[ul]
  \ar@<-.05mm>@[NavyBlue][ul]
  \ar@<-.05mm>@[Maroon]_{\Farr}[ur]  
  \ar@<.05mm>@[Maroon][ur]
}

\EventEditors{Patricia Bouyer and Jaco van de Pol}
\EventNoEds{2}
\EventLongTitle{36th International Conference on Concurrency Theory (CONCUR 2025)}
\EventShortTitle{CONCUR 2025}
\EventAcronym{CONCUR}
\EventYear{2025}
\EventDate{August 26--29, 2025}
\EventLocation{Aarhus, Denmark}
\EventLogo{}
\SeriesVolume{348}
\ArticleNo{4}

\begin{document}

\title{Denotational Semantics for Probabilistic and Concurrent Programs}

\maketitle

\begin{abstract}
We develop a denotational model for {\em probabilistic and concurrent imperative programs}, a class of programs with standard control flow via conditionals and while-loops, as well as probabilistic actions and parallel composition. Whereas semantics for concurrent or randomized programs in isolation is well studied, their combination has not been thoroughly explored and presents unique challenges. The crux of the problem is that interactions between control flow, probabilistic actions, and concurrent execution cannot be captured by straightforward generalizations of prior work on pomsets and convex languages, prominent models for those effects, individually.  Our model has good domain theoretic properties, important for semantics of unbounded loops. We also prove two adequacy theorems, showing that the model subsumes typical powerdomain semantics for concurrency and convex powerdomain semantics for probabilistic nondeterminism.
\end{abstract}

\section{Introduction}
\label{sec:intro}


From simple imperative languages, to functional ones, to probabilistic and quantum paradigms, the development of programming languages is often accompanied by a study of mathematical objects capturing the semantics of syntactic constructs. Program semantics underpins many static analysis and verification techniques, and also how a language can be extended.


Semantics can be presented in different styles---denotational, operational, and axiomatic---and each approach offers different insights.
 Denotational semantics is often chosen for expressivity and extensibility. For example, adding recursion raises questions about the underlying \emph{domain}, motivating the development of domain theory to provide mathematical representations of iterated computations \cite{scott1970outline,scott1971mathematical}.
Other constructs necessitate more involved domains; two prime examples are probabilistic and concurrent programs. 

Probabilistic programs have been widely studied---going back to the seminal works of Kozen \cite{psem} and McIver and Morgan \cite{mciver2005abstraction}---and their semantic domains require extra algebraic structure to capture the \emph{distribution} of outcomes generated by operations such as sampling and coin flips. This added convex structure brings additional complexity to questions of language extensions and program analysis, but the last decade has brought success stories on taming the complexity through program logics and predicate transformer calculi \cite{mciver2005abstraction,wpe,ellora,den_hartog1999verifying,zilberstein2025demonic}.

Concurrent programs have prominent applications in distributed systems and have long been a target of program analysis, as concurrency bugs are hard to detect and prevent. There have been many approaches to concurrent semantics, some arising from process algebra research and some from more practical hardware considerations (\eg memory models). In terms of denotational semantics, {\em pomsets} \cite{gischer1988equational,pratt1976semantical,brookes2004semantics,grabowski1981partial} provide a prominent model, expressing causality between actions and parallel branching as a partial order.



Semantics of concurrent and probabilistic programs, separately, are subtle and complex. Unsurprisingly, their combination introduces additional challenges. In this paper, we develop a denotational model for programs that mix probabilistic and concurrent constructs.
We work with a simple concurrent imperative programming language, shown below.
\begin{equation}\label{fig:syntax}
  \cmd \ni C \Coloneqq \skp \mid C_1;C_2 \mid \; C_1 \parop C_2 \; \mid \iftf b{C_1}{C_2} \mid \whl bC \mid a
\end{equation}
Here, the basic actions $a\in\act$ can perform probabilistic operations such as random sampling.
This is not only interesting from a theoretical point of view; for decades, randomization has been used to enhance the capabilities of concurrent programs \cite{rabin1980n-process,rabin1982choice,morris1978counting,hart1983termination}. For example, Dijkstra's famous \emph{Dining Philosophers Problem}---a distributed synchronization scheme---has no deterministic solution \cite{dijkstra1971hierarchical}, but it has a simple randomized one \cite{lehmann1981advantages}.  

Despite the importance of mixing randomization and concurrency, semantic models for probabilistic concurrent programs with unbounded looping are not yet well understood. Such models are necessary to analyze behaviors of many concurrent randomized protocols---including the Dining Philosophers---which avoid deadlock with probability 1, but admit an infinite trace whose probability converges to 0.
In this paper, we develop such a model, drawing insights from prior work on concurrent and probabilistic semantics, but solving challenges unique to their combination.
In a nutshell, the contributions of this paper are:



\begin{enumerate}[leftmargin=*]

\item We introduce \emph{pomsets with formulae}, and prove that they have well behaved domain-theoretic properties (\Cref{sec:semantics}), including an \emph{extension lemma} \cite{meyer1989pomset}, for extending monotone operations on finite pomsets to continuous operations on infinite ones. 

\item We define guarded, sequential, and parallel composition operations on pomsets with formulae (\Cref{sec:operations}) and use those operations to define the denotational semantics for the concurrent imperative language (\ref{fig:syntax}), with uninterpreted actions (\Cref{sec:pomset-semantics}).

\item We define a linearization operation (via the extension lemma), parametric on the computational domain, to convert the pomset semantics into a state transformer (\Cref{sec:linearization}).

\item We prove two adequacy theorems showing that our model captures established models for probabilistic nondeterminism (\Cref{thm:convex}) and pure concurrency (\Cref{thm:powdom}).

\end{enumerate}
We discuss related work in \Cref{sec:discussion}; omitted proofs and details are provided in \cite{zilberstein2025denotational}.

\section{An Overview of Probabilistic Concurrency Semantics}
\label{sec:convex-powerset}

In this section, we introduce the requirements for interpreting programs that are both probabilistic and concurrent. To illustrate the subtleties, consider the following program:
\begin{equation}\label{eq:program}
  C \triangleq x \samp\flip{\tfrac12} ; (y \coloneqq 0 \parop y \coloneqq 1) ; \iftf y\skp{x \coloneqq x+2}
\end{equation}
After executing $C$, the value of $x$ can be any integer between 0 and 3, and that value depends both on random sampling and nondeterminism (from concurrent scheduling). The parity of $x$ is fixed after the coin flip (returning 0 or 1 with probability $\frac 1 2$), but then the scheduler can influence its final value by choosing in which order to run the updates to $y$.
Importantly though, $x$ will be even with probability exactly $\frac12$, regardless of the behavior of the scheduler.

Any semantic domain for this program must encompass probability distributions of program states, while at the same time accounting for nondeterminism introduced from parallel branching. One might attempt to give semantics to this program using a naive composition of distributions and powerset (as monads) but, as it is well known from prior work, such combination can easily lead to non-compositional semantics \cite{varacca_winskel_2006,zwart2019no,zwart2020non,parlant2020monad}.
While there are several domains for this combination of effects, we use the \emph{convex powerset}, which is quite well-studied in that it is a monad \cite{jacobs2008coalgebraic} for sequential composition; it is a directed complete partial order (\dcpo) \cite{tix2009semantic,tix2000convex,keimel2017mixed,tix1999continuous} for finding fixed points of iterated computations; and has well behaved equational laws \cite{mislove2000nondeterminism,mislove2004axioms,bonchi2021presenting}.

We now give a basic account of the convex powerset; for more details on the \dcpo\ structure and monad operations, refer to \Cref{app:convex} and \cite[\S3]{zilberstein2025demonic}.
A discrete probability distribution $\mu \in \D(X) = X \to [0,1]$ is a countable map from $X$ to probabilities such that $\sum_{x\in X}\mu(x) = 1$. Convex combinations are defined for distributions $(\mu \oplus_p \nu)(x) = p\cdot \mu(x) + (1-p) \cdot \nu(x)$ and sets of distributions $S \oplus_p T \triangleq \{ \mu \oplus_p \nu \mid \mu \in S, \nu\in T\}$.
A set $S \subseteq \D(X)$ is \emph{convex} if $S = S \oplus_p S$ for all $p\in [0,1]$, \ie $S$ contains all convex combinations of its elements. 
Our computational domain---\emph{the convex powerset}---consists of nonempty convex sets of distributions\footnote{The \dcpo\ structure relies of a few more properties, which we report in \Cref{app:convex}.}.
\[
  \C(X) \triangleq \set{ S \subseteq \D(X \cup \{\bot\}) \;\;\middle|\;\; \text{$S$ is nonempty, convex, \textellipsis}}
\]
This computational domain allows us give semantics to the coin flip actions in program (\ref{eq:program}) $\de{a}_\act \colon \st \to \C(\st)$ where $\st \triangleq \mathsk{var} \to \mathsk{val}$ is the set of variable valuations. The semantics is given by a set containing a single distribution, where $x$ is set to 1 with probability $p$ and it is set to 0 with probability $1-p$. Being a \emph{set of distributions}, $\C(\st)$ can also represent nondeterministic choice, which we denote by $\nd$. Operationally, $S\nd T$ is not simply a choice between $S$ and $T$, but rather it is a choice of a \emph{distribution} over $S$ and $T$, corresponding to a scheduler that can use biased coin flips to make decisions \cite{varacca2002powerdomain}.
\begin{align*}
  \de{x\samp \flip{p}}_\act(s) &\triangleq \set{
    \left[\footnotesize
    \begin{array}{lll}
      s[x \coloneqq 1] & \mapsto & p
      \\
      s[x \coloneqq 0] & \mapsto & 1-p
    \end{array}
    \right]
  }
  &
    S \nd T &\triangleq \bigcup_{p\in[0,1]} S\oplus_p T
\end{align*}
Our goal is to develop a pomset model to compositionally reason about concurrency in a manner that is compatible with the convex powerset. Existing pomset models do not track control flow, and instead rely on the following equation, stating that a command followed by an if-statement can be decomposed into two traces, with the split lifted to the top:
\begin{equation}\label{eq:problem}
  C ; \iftf b{C_1}{C_2} \ \ \equiv\ \  (C; \assume b ; C_1) \nd (C; \assume{\lnot b} ; C_2)
\end{equation}
It is well known that this equation is invalid in probabilistic contexts, especially $\C$ \cite{mislove2000nondeterminism}. 
More concretely, prior denotational models for concurrency $\de{-}_\powdom \colon \cmd \to \mathcal{P}(\pom)$ (shown in \Cref{fig:pomlang-semantics}) map programs to \emph{pomset languages}---sets of pomsets---where each $\Alpha \in \de{C}_\powdom$ corresponds to a particular resolution of the tests in the program. For program (\ref{eq:program}), this approach yields the semantics on the left of \Cref{fig:pomset-example}, where the problematic \Cref{eq:problem} is implicitly applied. Arrows $a_1 \to a_2$ indicate causality, \ie $a_1$ must occur before $a_2$; when there is no arrow between two actions, they can be interleaved in any order.
There is a significant problem in $\de{C}_\powdom$; as we will soon see, there is no straightforward way to piece together the entire distribution of outcomes in the program.

\begin{figure}
\[
  \de{C}_\powdom = \set{
  \vcenter{\vbox{\footnotesize\xymatrix@R=6pt@C=-15pt{
  &\fork
  \\
  & \assume{y=1} \ar[u]
  \\
  y \coloneqq 0 \ar[ur]
  &&
  y\coloneqq 1 \ar[ul]
  \\
  &\fork \ar[ul] \ar[ur]
  \\
  & x \samp\flip{\frac12} \ar[u]
}}}
\;\;,\;\;
  \vcenter{\vbox{\footnotesize\xymatrix@R=6pt@C=-15pt{
  &x \coloneqq x+2
  \\
  & \assume{y\neq 1} \ar[u]
  \\
  y \coloneqq 0 \ar[ur]
  &&
  y\coloneqq 1 \ar[ul]
  \\
  &\fork \ar[ul] \ar[ur]
  \\
  & x \samp\flip{\frac12} \ar[u]
}}}
}
\qquad
  \de{C} =
  \vcenter{\vbox{\footnotesize\xymatrix@R=6pt@C=-15pt{
  ~~\fork~~ && x \coloneqq x+2
  \\
  & ?(y=1) \guardarrows
  \\
  y \coloneqq 0 \ar[ur]
  &&
  y\coloneqq 1 \ar[ul]
  \\
  &\fork \ar[ul] \ar[ur]
  \\
  & x \samp\flip{\frac12} \ar[u]
}}}
\]
%
\caption{Pomset language model (left) vs pomset with formulae model (right) for program (\ref{eq:program}) (we use a $\fork$ with 0 branches to represent $\skp$).}
\label{fig:pomset-example}
\end{figure}


Crucially, our pomset structure tracks the results of control-flow tests (which can rely on randomization). This is done symbolically by associating a Boolean formula to each node of the pomset, recording the resolution of tests needed to reach that point. We call these new structures {\em pomsets with formulae}, which capture both parallel composition and guarded branching\footnote{Formulae are more expressive than conflict relations, \eg in event structures: conflict can only exclude execution of branches, whereas formulae record the entire outcome of tests that lead to every node.}. As such, this structure encodes many traces, avoiding the need for a set of traces and simultaneously remaining compositional in the presence of probabilistic actions.

The semantics in our new model $\de{-}\colon \cmd\to\pom$ for program (\ref{eq:program}) is shown on the right of \Cref{fig:pomset-example}. Compared to $\de{C}_\powdom$, our new model has merged the set of structures into a single structure, where the two opposing tests are replaced by a single node labelled $?(y=1)$. The outgoing edges are labelled $\mathsf{T}$ and $\mathsf{F}$, indicating that they can only be followed if the test passes or fails, respectively.
Now, we may wish to know the possible probabilities for each outcome resulting from running this program. To this end, in \Cref{sec:linearization} we develop a \emph{linearization} procedure $\lin \colon \pom \to \st \to \C(\st)$ for interpreting a pomset as a state transformer.
Linearizing this structure, we get the following set of distributions:
\[
 \lin(\de{C})(s) = \set{
   \left[\begin{array}{lll}
     s[x \coloneqq 0, y\coloneqq 0] & \mapsto p
     \\
     s[x \coloneqq 1, y\coloneqq 0] & \mapsto q
     \\
     s[x \coloneqq 2, y\coloneqq 1] & \mapsto \frac12 - p
     \\
     s[x \coloneqq 3, y\coloneqq 1] & \mapsto \frac12 - q
   \end{array}\right]
   \middle| \;\; 
   0 \le p,q \le \frac12
 }
\]
Linearization provides new insights about the program, which were not obvious in the pomset structure. The scheduler can make $y$ equal to 0 or 1 with any probability, which matches our operational understanding of the program: the scheduler can choose to execute the commands in either the order $y \coloneqq 0;y\coloneqq 1$ or $y\coloneqq1;y\coloneqq 0$ (or some convex combination thereof). However, regardless of the scheduler's choice of $p$ and $q$, $x$ is even (or odd) with exactly probability $\frac12$. This formal semantics matches the intuition for the program, since the parity of $x$ is fixed by the initial coin flip.

We now contrast the semantics above with pomset language semantics to show that the program $C$ cannot be meaningfully interpreted using prior techniques.
Recall that in $\de{C}_\powdom$ each pomset corresponds to a particular resolution of the test, \ie $y=1$ and $y\neq 1$. So, in the first pomset we will always have $x \le 1$ whereas in the second pomset $x \ge 2$. Let us now try to obtain a meaningful state transformer semantics from this program by linearizing each structure and then merging the results. To do so, we use the following semantics for \code{assume}, where failed tests evaluate to a special $\unk$ symbol, which halts the program execution:
\[
  \de{\assume b}_\act(s) = \left\{
    \begin{array}{cl}
      \eta(s) & \text{if}~ \de{b}_\test(s) = 1
      \\
      \eta(\unk) & \text{if}~ \de{b}_\test(s) = 0
    \end{array}
  \right.
\]
Now, letting $\Alpha_1$ and $\Alpha_2$ be the left and right pomsets in $\de{C}_\powdom$, we apply linearization to both structures to obtain:
\begin{align*}
\lin(\Alpha_1)(s) &= \set{
   \left[\begin{array}{cll}
     s[x \coloneqq 0, y\coloneqq 0] & \mapsto& p_1
     \\
     s[x \coloneqq 1, y\coloneqq 0] & \mapsto& q_1
     \\
     \unk & \mapsto & 1- p_1 - q_1
   \end{array}\right]
   \middle| \;\; 
   0 \le p_1,q_1\le \frac12
 }
 \\
\lin(\Alpha_2)(s) &= \set{
   \left[\begin{array}{cll}
     s[x \coloneqq 2, y\coloneqq 1] & \mapsto& p_2
     \\
     s[x \coloneqq 3, y\coloneqq 1] & \mapsto& q_2
     \\
     \unk & \mapsto & 1- p_2 - q_2
   \end{array}\right]
   \middle| \;\; 
   0 \le p_2,q_2\le \frac12
 }
\end{align*}
Already, we can begin to see a problem; in $\lin(\de{C})(s)$ the scheduler picks two probabilities, but here the scheduler picks four probabilities, giving it a higher degree of freedom to influence the outcome of the program. We make this formal by attempting to define an operation $\bowtie$ to merge the two semantics. As a first attempt, we generate a new set of distributions by summing the non-$\unk$ probability mass in every pair of distributions from the two sets such that the mass of $\unk$ in one is equal to the non-$\unk$ mass in the other:
\[
S \bowtie T \triangleq \set{
  \mu \bowtie \nu
  \mid
  \mu \in S, \nu \in T,
  \mu(\unk) = 1- \nu(\unk)
}
\quad
(\mu\bowtie\nu)(s) \triangleq \left\{
  \begin{array}{ll}
    \mu(s) + \nu(s) & \text{if}~ s \neq \unk
    \\
    0 & \text{if}~ s = \unk
  \end{array}
\right.
\]
Indeed, this gives us $\lin(\de{C})(s) \subseteq \lin(\Alpha_1)(s) \bowtie \lin(\Alpha_2)(s)$: take any $\mu \in \lin(\de{C})(s)$, which is generated by fixing some probabilities $0 \le p,q\le\frac12$. Picking $p_1 = p$, $q_1 = q$, $p_2 = \frac12-p$, and $q_2 = \frac12-q$, we see that $\mu \in \lin(\Alpha_1)(s) \bowtie \lin(\Alpha_2)(s)$.
However, $\lin(\Alpha_1)(s) \bowtie \lin(\Alpha_2)(s)$ contains extra distributions that are not in $\lin(\de{C})(s)$, and are not correct outcomes of the program. For example, letting $p_1 = p_2 = \frac12$ and $q_1 = q_2=0$, we get:
\[
  \left[\begin{array}{lll}
     s[x \coloneqq 0, y\coloneqq 0] & \mapsto \frac12
     \\
     s[x \coloneqq 1, y\coloneqq 0] & \mapsto 0
     \\
     s[x \coloneqq 2, y\coloneqq 1] & \mapsto \frac12
     \\
     s[x \coloneqq 3, y\coloneqq 1] & \mapsto 0
   \end{array}\right]
   \in \lin(\Alpha_1)(s) \bowtie \lin(\Alpha_2)(s)
\]
But this distribution cannot be an outcome of running the program, since $x$ is even with probability 1, deviating from the expected operational behavior!

So a different implementation of $\bowtie$ is needed, which is more discerning about the compatibility of pairs of distributions from the two sets. But as it stands, there is no information in $\lin(\Alpha_1)(s)$ and $\lin(\Alpha_2)(s)$ to indicate compatibility.
While we do not claim that it is impossible to record this information during linearization, our investigation suggests that it would not be straightforward. A proper linearization of the pomset language $\de{C}_\powdom$ would essentially amount to a lockstep execution of the two pomsets, so that both branches can be taken with the appropriate probabilities. That lockstep execution corresponds to a straightforward linearization of the single pomset with formulae $\de{C}$, where each test already has information about both branches. We therefore conclude that pomsets with formulae are the correct semantic structure for this domain.
In the remainder of the paper, we develop the formal details of our pomset with formulae model and associated linearization procedure.


\section{Semantic Structures}\label{sec:semantics}

In this section, we define a particular kind of labelled partial order, which will form the core of our semantic model. We call this structure a {\em labeled partial order with formulae} (\lpof), as it includes a special labelling, assigning a Boolean formula to each node. We recall basic definitions from domain theory in \Cref{app:ord}; for a complete treatment refer to \cite{abramsky1995domain}.

\subsection{Labelled Partial Orders with Formulae}

Let $\Nodes$ be a countable universe of nodes, that will be denoted by $x, y, z, \ldots$, whereas subsets of $\Nodes$ will be denoted by $N, X, Y, \ldots$.
We start by defining the class of Boolean formulae using nodes as the variables:
\[
  \Form \ni \psi \Coloneqq \tru \mid \fls \mid \psi_1 \land \psi_2 \mid \psi_1\vee\psi_2 \mid \lnot \psi \mid x 
\]
Formulae are interpreted by valuations $v \colon \Nodes \to \mathbb{B}$, where $\mathbb{B} = \{0, 1\}$; 
the satisfaction relation, written $v \vDash \psi$, is defined in the standard way (see \Cref{def:formulae} in \Cref{app:ord}).
The variables of a formula $\psi$, written $\free(\psi)$, are those nodes that appear in $\psi$. We write $\sat(\psi)$ iff there exists a valuation $v$ such that $v\vDash \psi$; $\psi \Rightarrow \psi'$ iff $v\vDash \psi'$, for every $v\vDash \psi$; and $\psi \Leftrightarrow\psi'$ iff $\psi\Rightarrow \psi'$ and $\psi'\Rightarrow\psi$.
In what follows, given a strict poset $\tuple{X, <}$ and any $x \in X$, we let
$\predplus x \triangleq \{ y \in X  \mid  y < x \}$,
$\succplus x \triangleq \{ y \in X  \mid  x < y \}$,
$\succ(x) \triangleq \{ y \in X\ \mid\ x < y\ \wedge \not\exists z. (x < z < y) \}$, and
$\min \triangleq \{ x \in X \mid \predplus x = \emptyset \}$.
Further, the level of $x$, written $\lev(x)$, is the length of the longest path from a minimal element to $x$. See \Cref{app:ord} for more details.

\begin{definition}[\lpof]\label{def:lpof}
  Let $\tuple{L, \leq}$ be a pointed, finitely preceded  \dcpo\ with bottom element $\bot$. A labelled partial order with formulae (\lpof) over $L$ is a 4-tuple $\alpha = \tuple{N, <, \lambda, \varphi}$ where:
  \begin{enumerate}[leftmargin=*]
 \item  $N \subseteq \Nodes$ is a countable set of nodes; 
\item $\tuple{N,\mathord{<}}$ is a (strict) poset such that:
  \begin{enumerate}
    \item it is finitely preceded, that is $|\predplus x| < \infty$, for every $x \in N$; 
    \item every level has a finite number of nodes, \ie $|\lev^{-1}(n)| < \infty$ for all $n \in \mathbb N$; and 
    \item it is single-rooted, that is: $|\min| = 1$.
  \end{enumerate}
\item $\lambda \colon N \to L$ is a labelling function such that $\succ(x) = \emptyset$ whenever $\lambda(x) = \bot$.
\item $\varphi \colon N \to \Form$ is a formula function satisfying:
  \begin{enumerate}
    \item $\sat(\varphi(x))$ and $\free(\varphi(x)) \subseteq \predplus{x}$, for all $x \in N$;
    \item $\varphi(y) \Rightarrow \varphi(x)$, for all $x < y$.
  \end{enumerate}
  \end{enumerate}
We denote by $\lpo(L)$ the set of all \lpof s over $L$.
\end{definition}

For some \lpof\ $\alpha = \tuple{N,<,\lambda,\varphi}$, we use $N_\alpha$, $<_\alpha$, $\lambda_\alpha$, and $\varphi_\alpha$ to refer to its constituent parts; similarly, 
we annotate all functions with the \lpof\ 
they refer to, \eg $\predplus x_\alpha$, $\succ_\alpha(x)$, $\min_\alpha$, etc. We now explain the conditions in \Cref{def:lpof}. 

As is typical in pomset semantics, the partial order $\tuple{N,<}$ denotes \emph{causality}: $x < y$ iff $x$ must be scheduled before $y$. If $x \not< y$ and $y \not< x$, then $x$ and $y$ execute concurrently, and can be interleaved in any order.
With this in mind, conditions (2a) and (2b) are standard; in particular,
since a node represents an action of a program, having finitely many predecessors ensures that every action may happen in a finite amount of time.
Moreover, requiring that every level has a finite number of nodes is a weakening of the usual finite branching property: indeed, we allow a node to have infinitely many successors, but they cannot be all at the same level.
We defer the discussion on Condition (2c) to \Cref{rem:singlerooted} later on.
 
Concerning Condition (3), labels correspond to \emph{Boolean tests} and \emph{actions}, performed during a program execution (\eg \Cref{fig:pomset-example}).
Unlike standard pomset models, actions can be probabilistic, and our model is consistent with known sequential probabilistic semantics (\Cref{thm:convex}).
For the moment, we leave labels unspecified, only assuming the presence of $\bot$, used to denote a nonterminating computation, which is fundamental in approximating the fixed point semantics of while-loops. Since $\bot$ denotes nontermination, nodes labelled with $\bot$ cannot be a predecessor of any node, as the successors of $\bot$ would never be executed.

Condition (4) is about formulae, which are crucial for modeling conditional constructs (if-then-else and while-loops), such as
$
\iftf {b_1} {a_1}{\iftf {b_2} {a_2}{a_3}}
$,
where $b_1$ and $b_2$ are tests and $a_1,a_2,a_3$ are actions.
This program corresponds to the \lpof\ below, where $\lambda(x_\ell) = \ell$ for all $\ell \in \{b_1, b_2, a_1, a_2, a_3\}$,
$w {\color{NavyBlue}\xrightarrow{{\color{black}\Tarr}}} w_1$ means that $\varphi(w_1) \Leftrightarrow \varphi(w) \land w$ and $w {\color{Maroon}\xrightarrow{{\color{black}\Farr}}} w_2$ means that $\varphi(w_2) \Leftrightarrow \varphi(w)\land\lnot w$ (\ie  $w$ is labeled with a test, whose outcome leads to $w_1$ or $w_2$, depending on the Boolean outcome depicted on the colored arc):
\begin{equation}\rm
\label{ex:lpoFormulae}
\vcenter{\vbox{\xymatrix@R=3pt@C=5pt{
\scalebox{0.8}{$\varphi(x_{a_2}) = \neg x_{b_1} \wedge x_{b_2}$}
&& x_{a_2} && x_{a_3}
&
\scalebox{0.8}{$\varphi(x_{a_3}) = \neg x_{b_1} \wedge \neg x_{b_2}$}
\\
\scalebox{0.8}{$\varphi(x_{a_1}) = x_{b_1}$}
&
x_{a_1} && x_{b_2}\guardarrows 
&&
\scalebox{0.8}{$\varphi(x_{b_2}) = \neg x_{b_1}$}
\\
\scalebox{0.8}{$\varphi(x_{b_1}) = \tru$}
&
& x_{b_1} \guardarrows 
&
}}}
\end{equation}
Hence, every node $x$ labelled with a Boolean test yields a (binary) branch and its (two) successors have the same formula as $x$, with an extra conjunct that is either $x$ or $\neg x$, according to the outcome of the test.
With this in mind, \Cref{def:lpof}(4) is quite intuitive:
\begin{itemize}[leftmargin=*]
\item Requiring $\sat(\varphi(x))$ amounts to expressing the fact that every node $x$ is reachable, \ie there exists a truth assignment $v$ such that $v \models \varphi(x)$. The truth assignment tells, for every node labeled with a Boolean test, whether that test passes or not. In our example, $x_{a_3}$ is reachable when both the tests $b_1$ and $b_2$ fail; this is exactly what satisfiability of $\varphi(x_{a_3})$ entails (viz., that both $b_1$ and $b_2$ must be false to satisfy $\neg x_{b_1} \wedge \neg x_{b_2}$).
\item This intuition justifies the other two requirements of \Cref{def:lpof}(4): reachability of a node can only depend on the values of the tests that precede it (and so $\varphi(x)$ can only use nodes that precede $x$) and, since along a path we shall only add conjuncts, the formula of a higher-level node is stronger than the formulae of all of its predecessors.
\end{itemize}

Next, we define an order on \lpof s, which we will use for the construction of fixed points to give semantics to unbounded loops. The intuition is that $\alpha \lelpo \beta$ iff
$\beta$ has \emph{more behaviors} than $\alpha$, meaning that $\beta$ can be obtained by expanding $\bot$ nodes in $\alpha$ into larger structures (each  unfolding of a while-loop is obtained in this way, as shown in Section \ref{sec:pomset-semantics}).
Since \lpof s can only be expanded from $\bot$ nodes, the label set $L$ must be pointed (\ie $\bot\in L$); however, the order on $L$ need not be flat: labels themselves can also become larger when passing from $\alpha$ to $\beta$ (\eg if actions are nondeterministic assignments from a given set, this corresponds to enlarging the set of possible choices). This can also be useful to model \emph{invariant sensitive execution}, introduced in Probabilistic Concurrent Outcome Logic \cite{zilberstein2024probabilistic}.
Notationally,
$\Bot_\alpha$ denotes the set of nodes labelled with $\bot$, \ie $\Bot_\alpha \triangleq  \{x \in N_\alpha \mid \lambda_\alpha(x) = \bot\}$.

\begin{definition}[Ordering on \lpof s]
\label{def:lelpo}
We let $\alpha \lelpo \beta$ iff:
\begin{enumerate}
\item\label{def:lelpo:N} $N_\alpha$ is a downward-closed subset of $N_\beta$, written $N_\alpha \dwclosed N_\beta$;
\medskip
\item \label{def:lelpo:<} $\mathord{<_\alpha} = \mathord{<_\beta} \cap (N_\alpha \times N_\alpha)$;
\medskip
\item $\forall x\in N_\alpha$:
\begin{enumerate*}
\item\label{def:lelpo:lambda} $\lambda_\alpha(x) \le \lambda_\beta(x)$;
\item $\varphi_\alpha(x) = \varphi_\beta(x)$;
\item\label{def:lelpo:succ} $\succ_\alpha(x) = \succ_\beta(x) \setminus \succplus{\Bot_\alpha}_\beta$.
\end{enumerate*}
\end{enumerate}
\end{definition}

\begin{example}\rm
Consider the following three \lpof s:
\[
\alpha =
\vcenter{\vbox{\xymatrix@R=5pt@C=6pt{
&
\\
y_1(\bot) && y_2
\\
& x \ar[ul]\ar[ur]
}}}
\qquad\qquad
\beta = 
\vcenter{\vbox{\xymatrix@R=5pt@C=6pt{
&& z
\\
y_1(\bot) && y_2 \ar[u]
\\
& x \ar[ul]\ar[ur]
}}}
\qquad\qquad
\gamma =
\vcenter{\vbox{\xymatrix@R=5pt@C=6pt{
& z
\\
y_1\ar[ur] && y_2 \ar[ul]
\\
& x \ar[ul]\ar[ur]
}}}
\]
where node $y_1$ is labelled with $\bot$ in $\alpha$ and $\beta$, whereas all other nodes have non-$\bot$ labels.
Following the intuition given above, both $\alpha$ and $\beta$ are attempts to specify that the computation in $\gamma$ is truncated at $y_1$.
However, only $\alpha \lelpo \gamma$; this is because, if $\bot$ represents a diverging computation, then $z$ will never have all the causes it needs to be executed and so it must disappear. Indeed, $\beta \nlelpo \gamma$ because, even though $N_\beta \dwclosed N_\gamma$, condition (3c) of \Cref{def:lelpo} is violated: $z \in \succ_\beta(y_2)$ but $z \not\in \succ_\gamma(y_2) \setminus \succplus{\Bot_\beta}_\gamma$, since $y_1 \in \Bot_\beta$ and $y_1 <_\gamma z$.
\end{example}

Finally, in order to find fixed points, we will need to have a directed complete partial order (\dcpo) structure and a notion of continuity for functions, which we now define. 

\begin{definition}[\dcpo]
Given a poset $\tuple{X,\leq}$, a subset $D \subseteq X$ is called {\em directed} iff it is not empty and, for every two elements $x_1,x_2 \in D$, there exists $x \in D$ such that $x_1, x_2 \leq x$.

$\tuple{X,\leq}$ is a {\em directed complete partial order} (\dcpo) iff, for every directed set  $D$, $\sup D$ exists; furthermore, $\tuple{X,\leq}$ is called {\em pointed} if there exists an element $\bot \in X$ such that $\bot \le x$ for all $x\in X$.
\end{definition}

\begin{definition}[Scott Continuity]
Given two \dcpo s $\tuple{X, \le_X}$ and $\tuple{Y, \le_Y}$, a function $f \colon X \to Y$ is \emph{Scott Continuous} if it is monotone: $f(x) \le_Y f(x')$ for all $x \le_X x'$; and preserves suprema of directed sets: $\sup_{x\in D} f(x) = f(\sup D)$ for all $D \subseteq X$ directed.
\end{definition}

\begin{restatable}{lemma}{lposup}
\label{lem:lpo-sup}
For any directed set $D \subseteq \lpo(L)$, $\sup D = \tuple{N, <, \lambda, \varphi}$, where: 
  \begin{align*}
    N &\triangleq \bigcup_{\beta\in D} N_{\beta}
    &
    \mathord{<} &\triangleq \bigcup_{\beta\in D} \mathord{<_{\beta}}
    &
    \lambda(x) &\triangleq \sup_{\beta\in D\,:\,x \in N_{\beta}} \lambda_{\beta}(x)
    &
    \varphi(x) & \triangleq \psi_x
  \end{align*}
  and $\psi_x = \varphi_\beta(x)$, for all $\beta \in D$ such that $x\in N_\beta$.
\end{restatable}
In \Cref{lem:lpo-sup},  $\varphi$ is well defined since, for every $x \in \Nodes$ and $\beta,\beta' \in D$ that contain $x$, it must be that $\varphi_\beta(x) = \varphi_{\beta'}(x)$. Indeed, since $D$ is directed, there must exist a $\gamma \in D$ such that $\beta,\beta' \lelpo \gamma$ and, by \Cref{def:lelpo}(3b), we have that $\varphi_\beta(x) = \varphi_\gamma(x) = \varphi_{\beta'}(x)$.

Notice that $\lpo(L)$ is {\em not} pointed, since the node at the root can vary; hence, there is no single \lpof\ that is smaller (w.r.t. $\lelpo$) than all other \lpof s. This situation will change when moving to pomsets \Arefp{lem:pomdcpo}, which abstract away from the specific nodes.

\subsection{Pomsets with Formulae}

Although \lpof s express causality and branching behavior, they are not ideal semantic structures because information about node identifiers is extraneous. We instead work with \emph{partially ordered multisets} (pomsets), which abstract away from the specific names used.

\begin{definition}[\lpof\ Isomorphism]
\label{def:equivLPO}
Let $\alpha$ and $\beta$ be \lpof s and $f \colon N_\alpha \to N_\beta$ be a bijection between their nodes. Define $f(\alpha) = \tuple{ N, <, \lambda, \varphi }$ as
\begin{align*}
  N &\triangleq \{ f(x) \mid x\in N_\alpha \}
  &
  \mathord{<} &\triangleq \{ (f(x), f(y)) \mid x <_\alpha y \}
  &
  \lambda &\triangleq \lambda_\alpha \circ f^{-1}
  &
  \varphi &\triangleq f \circ \varphi \circ f^{-1}
\end{align*}
where $f(\psi)$ syntactically renames the variables of $\psi$ in the obvious way;
$\alpha$ and $\beta$ are isomorphic, written $\alpha \equiv \beta$, iff there is a bijection $f \colon N_\alpha \to N_\beta$ such that $f(\alpha) = \beta$.
\end{definition}

\begin{definition}[Pomsets with Formulae]
Let $[\alpha] \triangleq \{ \beta \in \lpo(L) \mid \alpha\equiv\beta \}$ be the isomorphism class of $\alpha$.
A pomset with formulae (or, simply, pomset) $\Alpha \in \pom(L)$ is an isomorphism class of \lpof s:
$
  \pom(L) \triangleq \{ [\alpha] \mid \alpha \in \lpo(L) \}
$.
Let $\pom_\fin(L)$ be the set of pomsets with finitely many nodes. Finally, let $\Alpha \lepom \Beta$ iff $\forall \alpha \in \Alpha.\, \exists\beta\in \Beta.\ \alpha \sqsubseteq_\lpo \beta$.
\end{definition}



The semantics of programs depends on a variety of pomset composition operations (sequential, guarded, and parallel composition). We will also later \emph{linearize} pomsets into state transformer functions. Some of these operations cannot be defined inductively on infinite structures, so we need ways to extend functions on finite structures to infinite ones. We first show that infinite pomsets correspond to the suprema of their finite approximations.

%

\begin{restatable}[Finite Approximations]{lemma}{finapproxx}\label{lem:fin-approx}
For any $\Alpha \in \pom(L)$, $\Alpha = \sup \finapprox\Alpha$, where
$\finapprox\Alpha \triangleq \{ \Beta \in \pom_\fin(L) \mid  \Beta\lepom \Alpha \}$ is the set of finite approximations of $\Alpha$. 
\end{restatable}

Finite approximations coincide with the approximation order $\ll$ of \cite{abramsky1995domain} instantiated to pomsets
(see \Aref{sec:approx} for more details); thus, from now on, we let
$\Alpha \ll \Beta$ denote $\Alpha \in \finapprox\Beta$.
We are now ready to show that every operation $f$ on finite pomsets can be extended to infinite ones by defining the operation on $\Alpha$ as the sup of the images through $f$ of all the approximations of $\Alpha$. This will be useful in the remainder of the paper.


\begin{restatable}[Extension]{lemma}{extension}\label{lem:extension}
Let $f\colon \pom_\fin(L)^n \to T$ be a monotone function on the \dcpo\ $\tuple{T, \le}$. Then $f^\ast \colon \pom(L)^n \to T$, shown below, is well-defined and Scott continuous:
\[
  f^*(\Alpha_1, \ldots, \Alpha_n) \triangleq \sup_{\Alpha'_1 \ll \Alpha_1} \cdots \sup_{\Alpha'_n \ll \Alpha_n} f(\Alpha'_1, \ldots, \Alpha'_n)
\]
\end{restatable}

%

\section{Pomset Operations}
\label{sec:operations}

In this section, we define pomset operations that mirror the program syntax $C\in\cmd$ from (\ref{fig:syntax}).
We first define the singleton \lpof\  $\singleton{\ell}_x$ on node $x$ with label $\ell \in L$; this generalizes to $\singleton\ell$ on pomsets and will be used to give the semantics to $\skp$ and atomic actions:
\begin{align*}
  \singleton{\ell}_x & \triangleq \tuple{ \{x\}, \emptyset, [x \mapsto \ell], [x \mapsto \tru]} \in \lpo(L)
  &
  \singleton\ell &\triangleq \{ \singleton{\ell}_x \mid x \in \Nodes \} \in \pom(L)
\end{align*}

\subsection{Guarded Choice}

We now define guard,
recording the causality between a test and the two branches of computation defined on it.
Assuming that $N_\alpha \cap N_\beta = \emptyset$ and $x \not\in N_\alpha \cup N_\beta$, we let
$\guard(x,\ell,\alpha,\beta) = \tuple{ N, <, \lambda,\varphi }$ where 
we overload $\land$ s.t. $\psi\land\tru=\tru\land\psi=\psi$, and:
\begin{align*}
  N &\triangleq \{x\} \cup N_{\alpha} \cup N_{\beta}
  &
  \mathord{<} &\triangleq \mathord{<_{\alpha}} \cup \mathord{<_{\beta}} \cup \left ( \{x\} \times (N_{\alpha} \cup N_{\beta}) \right )
  \\
  \lambda(y) &\triangleq \left\{
    \begin{array}{ll}
      \ell & \text{if} ~y = x
      \\
      \lambda_{\alpha}(y)  & \text{if}~ y \in N_{\alpha}
      \\
      \lambda_{\beta}(y)  & \text{if}~ y \in N_{\beta}
    \end{array}
  \right.
  &
  \varphi(y) &\triangleq \left\{
    \begin{array}{ll}
      \tru & \text{if} ~y = x
      \\
      \varphi_{\alpha}(y)\land x  & \text{if}~ y \in N_{\alpha}
      \\
      \varphi_{\beta}(y)\land \lnot x  & \text{if}~ y \in N_{\beta}
    \end{array}
  \right.
\end{align*}
So, $\guard(x, b, \alpha, \beta)$ joins $\alpha$ and $\beta$ with a new root node $x$, whose label is $b$, and additionally, the formulae in $\alpha$ and $\beta$ are updated to require that $b$ passes or fails, respectively.
As an example, 
$
\guard(x_{b_1},b_1,\singleton{a_1}_{x_{a_1}},
   \guard(x_{b_2},b_2, \singleton{a_2}_{x_{a_2}},\singleton{a_3}_{x_{a_3}})
)
$
produces the \lpof\ depicted in (\ref{ex:lpoFormulae}).
For finite pomsets, we define the $\guard$ operation as follows:
\[
  \guard(\ell,\Alpha,\Beta) \triangleq \left\{ \guard(x,\ell,\alpha,\beta) \ 
      \middle| \ 
      \alpha \in\Alpha,\ \beta\in\Beta, N_\alpha \cap N_\beta = \emptyset, x \notin N_\alpha\cup N_\beta 
  \right\}
\]
Since $\guard$ is monotone \Arefp{lem:monGuard}, we can extend it to infinite pomsets:
\[
  \guard(\ell,\Alpha, \Beta) \triangleq \sup_{\Alpha' \ll \Alpha} \sup_{\Beta' \ll \Beta} \guard(\ell,\Alpha', \Beta')
\]

\subsection{Sequential Composition}
\label{sec:seq}

The sequential composition operation $\Alpha\fatsemi\Beta$ is meant to enforce that all actions in $\Alpha$ must occur before any of those in $\Beta$. This causality dependency interacts with guarded branching, as we will now see, which makes the formal definition of $\fatsemi$ more challenging. Consider the following \lpof s (where we denote with $x$ the singleton $\singleton \ell _x$ when the label is not relevant):
\begin{equation}\label{eq:seq-ex1}\small
\mathsf{a.}\
x\fatsemi y=
\vcenter{\vbox{
\xymatrix@R=8pt@C=0pt{
y \\ x\ar[u]
}}}
\qquad
\mathsf{b.}\
\left(
\vcenter{\vbox{
\xymatrix@R=6pt@C=0pt{
y_1 && y_2
\\
& x \ar[ul]\ar[ur]
}}}
\right) \fatsemi z =
\vcenter{\vbox{
\xymatrix@R=6pt@C=0pt{
& z
\\
y_1 \ar[ur] && y_2 \ar[ul]
\\
& x \ar[ul]\ar[ur]
}}}
\qquad
\mathsf{c.}\
\left(
\vcenter{\vbox{
\xymatrix@R=6pt@C=0pt{
y_1 && y_2
\\
& x \guardarrows
}}}
\right) \fatsemi z =
\vcenter{\vbox{
\xymatrix@R=6pt@C=0pt{
z_x && z_{\lnot x}
\\
y_1\ar[u] && y_2 \ar[u]
\\
& x \guardarrows
}}}
\end{equation}
The sequential composition of two singletons (\ref{eq:seq-ex1}a) simply creates a causality between them. When the program has forked into parallel (unguarded) branches (\ref{eq:seq-ex1}b), $\fatsemi$ introduces a diamond structure, as is typical in pomset semantics. However, if the program has a \emph{guarded} branch (\ref{eq:seq-ex1}c), then the following actions must be copied after each branch. Copying is essential to correctly account for loops: in the program $(\whl ba)~;~a'$, if we only created a single node labelled with $a'$, then that node would not be finitely preceded, invalidating \Cref{lem:fin-approx}; see \Cref{fig:trace-sem} and \Aref{app:loops} for a more detailed account. When a structure contains $\bot$ nodes, so that some paths are \emph{stuck}, the behavior of $\fatsemi$ is more complicated. Consider the following \lpof s, where node $y_1$ is labelled with $\bot$ and all other nodes have non-$\bot$ labels:
\begin{equation}\label{eq:seq-ex2}\small
\mathsf{a.}\
\left(
\vcenter{\vbox{
\xymatrix@R=6pt@C=0pt{
y_1(\bot) && y_2
\\
& x \guardarrows
}}}
\right) \fatsemi z =
\vcenter{\vbox{
\xymatrix@R=6pt@C=0pt{
 && z
\\
y_1(\bot) && y_2 \ar[u]
\\
& x \guardarrows
}}}
\qquad\qquad
\mathsf{b.}\
\left(
\vcenter{\vbox{
\xymatrix@R=6pt@C=0pt{
y_1(\bot) && y_2
\\
& x \ar[ul]\ar[ur]
}}}
\right) \fatsemi z =
\vcenter{\vbox{
\xymatrix@R=6pt@C=0pt{
y_1(\bot) && y_2
\\
& x \ar[ul]\ar[ur]
}}}
\end{equation}
When composing the singleton $z$ after a guarded branch with one stuck path (\ref{eq:seq-ex2}a), $z$ is only added to the non-stuck branch. In the program execution, only one branch will be taken, so $z$ can be safely executed as long at the test $x$ is false. However, the same is not true for parallel branching; in (\ref{eq:seq-ex2}b), both $y_1$ and $y_2$ will eventually be scheduled, so $z$ must occur after $y_1$---which is equivalent to not occurring at all---thus it does not appear in the final structure. This behavior ensures monotonicity of $\fatsemi$ \Arefp{lem:monSeq}.

We first define $\fatsemi$ for finite \lpof s, and then extend our definition to infinite ones using \Cref{lem:extension}. In the rest of this section, we assume that $\alpha,\beta\in \lpo_\fin(L)$. 
We start by defining the notion of stuck computations, extensible nodes, and branches.
\[\arraycolsep=2pt
\def\arraystretch{1.5}
\begin{array}{c}
  \stuck_\alpha \triangleq {\bigvee_{x\in \Bot_\alpha}}\varphi_\alpha(x)
  \\
  \extens_\alpha \triangleq\left\{ x \in N_\alpha \mid  \varphi_\alpha(x) \not\Rightarrow \stuck_\alpha \right\}
\end{array}
  \quad
    \br_\alpha \triangleq \set{ \varphi_\alpha(S) \ \middle|\
    \footnotesize
    \begin{array}{l}
      \emptyset \subset S \subseteq \extens_\alpha, \
       \varphi_\alpha(S) \Rightarrow \lnot\stuck_\alpha, \\
        \ \forall T.\ S \subset T\subseteq \extens_\alpha \Rightarrow \lnot\mathsf{sat}(\varphi_\alpha(T))
        \end{array}
    }
\]
The formula $\stuck_\alpha$ indicates which nodes are guaranteed to encounter a $\bot$ later in their execution.
Extensible nodes $x\in\extens_\alpha$ are not stuck, so there is some computation path that they can take without encountering any $\bot$ node. In (\ref{eq:seq-ex1}), all nodes are extensible, whereas in (\ref{eq:seq-ex2}a) only $x$ and $y_2$ are extensible and in (\ref{eq:seq-ex2}b) none are extensible.
The set of branches $\br_\alpha$ contains all of the maximal formulae that can be obtained as the conjunction of formulae of extensible nodes:
there, for a set $S \subseteq N_\alpha$, we define $\varphi_\alpha(S) \triangleq \bigwedge_{x\in S} \varphi_\alpha(x)$.
For example, in (\ref{eq:seq-ex1}a,b), $\br = \{ \tru\}$ since there are no tests; in (\ref{eq:seq-ex1}c), $\br=\{ x, \lnot x\}$ since both outcomes of the test are not stuck; in (\ref{eq:seq-ex2}a), $\br = \{\lnot x\}$ since the program is stuck if $x$ passes; and in (\ref{eq:seq-ex2}b), $\br = \emptyset$, since all paths are stuck.
The set of branches increases monotonically with $\lelpo$. Refer to \Aref{app:seq-example} for more examples of extensibility and branches.
As we saw in (\ref{eq:seq-ex1}c), we will need an isomorphic copy $\beta_\psi \equiv \beta$ for each branch $\psi \in\br_\alpha$. These copies are generated by a function $f \colon \br_\alpha \to [\beta]$ such that the nodes of $f(\psi)$ are disjoint from those of $\alpha$ and of every $f(\psi')$ where $\psi \neq\psi'$. The function $f$ is drawn from the following set:
\[
  \copylpo_{\alpha,\beta} \triangleq \big\{ f \colon \br_\alpha\! \to [\beta]\,
    \ \big|\
         \forall \psi \in \br_\alpha.\, (N_{f(\psi)} \cap N_\alpha = \emptyset\ \wedge
        \ \forall \psi' \neq \psi. N_{f(\psi)} \cap N_{f(\psi')} = \emptyset)
    \big\}
\]
%
Given any $\alpha,\beta\in \lpo_\fin(L)$ and $f\in\copylpo_{\alpha,\beta}$, we define sequential composition as
$
  \alpha \fatsemi_f \beta = \tuple{N, <, \lambda, \varphi}
$
where $\beta_\psi \triangleq f(\psi)$ and the components are defined as follows:
\begin{align*}
  N &\triangleq N_\alpha \cup \bigcup_{\psi \in \br_\alpha} N_{\beta_\psi}
  &
  \mathord{<} &\triangleq \mathord{<_\alpha} \cup \smashoperator{\bigcup_{\psi \in \br_\alpha}} \left( \mathord{<_{\beta_{\psi}}} \cup
      (\{ x\in N_\alpha \mid \psi \Rightarrow \varphi_\alpha(x) \} \times N_{\beta_\psi})
   \right)
  \\
  \lambda(x) &\triangleq \left\{
  \arraycolsep=2pt
    \begin{array}{ll}
      \lambda_\alpha(x) & \text{if} ~ x \in N_\alpha
      \\
      \lambda_{\beta_\psi}(x) & \text{if} ~ x \in N_{\beta_\psi}
    \end{array}
  \right.
  &
  \varphi(x) &\triangleq \left\{
  \arraycolsep=2pt
    \begin{array}{ll}
      \varphi_\alpha(x) & \text{if} ~ x \in N_\alpha
      \\
      \varphi_{\beta_\psi}(x)\land \psi & \text{if} ~ x \in N_{\beta_\psi}
    \end{array}
  \right.
\end{align*}
(where, again, $\psi\land\tru=\tru\land\psi=\psi$).
So, the nodes of $\alpha\fatsemi_f\beta$ are the nodes of $\alpha$, and all the nodes of the isomorphic copies $\beta_\psi$. The new order preserves the causalities in $\alpha$ and in each $\beta_\psi$, and additionally requires that $\beta_\psi$ occurs after all the nodes in the branch $\psi$. All labels are preserved, and formulae in $\beta_\psi$ are updated to also include $\psi$.
Sequential composition for finite pomsets is defined below (left), and is defined for infinite pomsets by extension (right).
\[
\Alpha \fatsemi \Beta \triangleq \{ \alpha \fatsemi_f \beta \mid \alpha \in\Alpha,\ \beta\in\Beta,\ f\in \copylpo_{\alpha,\beta} \}
\qquad
\Alpha \fatsemi \Beta \triangleq \sup_{\Alpha' \ll \Alpha} \sup_{\Beta' \ll \Beta} \Alpha' \fatsemi \Beta'
\]

\begin{remark}\rm
\label{rem:singlerooted}
We can now explain Condition (2c) in \Cref{def:lpof}, requiring single-rootedness, which may seem strange in a framework with concurrency; indeed, the usual semantics of two commands put in parallel is obtained by taking the (disjoint) union of their pomsets (and this yields several possible minimum elements in the resulting pomset). However, having multi-rooted \lpof s would make sequential composition not monotone. For example,
$y \lelpo y\ z\ $ but 
$\ x \fatsemi y \triangleq
\vcenter{\xymatrix@R=5pt@C=10pt{
y
\\
x \ar[u]
}}
$
$\ \nlelpo\!\!$
$
\vcenter{\xymatrix@R=5pt@C=4pt{
y&&z
\\
&x \ar[ul]\ar[ur]
}} \triangleq x \fatsemi (y\ z)
$. This is undesirable and easily fixed with the mild requirement of single-rootedness in \Cref{def:lpof}.
\end{remark}

\subsection{Parallel Composition}
\label{sec:parComp}

We finally define parallel composition. We assume that $L$ contains a special (non-$\bot$) label $\fork$ that denotes thread forking and let
$
\nFork_\alpha \triangleq  N_\alpha \setminus \{x \in \textstyle\min_\alpha \mid \lambda_\alpha(x) = \fork\}
$
denote the set of nodes of $\alpha$ without its $\fork$-minimal elements (if any).
Then, for both finite and infinite \lpof s $\alpha$ and $\beta$, and for any $x \not\in N_\alpha \cup N_\beta$,
we let $\alpha \parallel_x \beta \triangleq \tuple{N, <,\lambda, \varphi  }$, where
\begin{align*}
N & \triangleq \{x\} \cup \nFork_\alpha \cup \nFork_\beta
\\
  < &\triangleq
      (\{x\} \times (\nFork_\alpha \cup \nFork_\beta))
       \cup (\mathord{<_\alpha} \cap (\nFork_\alpha \times \nFork_\alpha))
       \cup (\mathord{<_\beta} \cap (\nFork_\beta \times \nFork_\beta))
\\
 \lambda(y) &\triangleq \left\{
   \begin{array}{ll}
     \fork & \text{if} ~ y = x \\
     \lambda_\alpha(y) & \text{if}~ y \in \nFork_\alpha\\
     \lambda_\beta(y) & \text{if}~ y \in \nFork_\beta
   \end{array}\right.
 \qquad\qquad
 \varphi(y) \triangleq  \left\{
   \begin{array}{ll}
     \tru & \text{if} ~ y = x \\
     \varphi_\alpha(y) & \text{if}~ y \in \nFork_\alpha\\
     \varphi_\beta(y) & \text{if}~ y \in \nFork_\beta
   \end{array}\right.
\end{align*}
Notice that the branching that arises from $\parallel$ is conceptually (and practically) different from the branching that arises from a test: the branching due to $\parallel$ does not exclude any branch, whereas the two branches due to a test are mutually exclusive. This is apparent by the different way in which we handle the formulae associated to the nodes after the branch: they remain the same under parallel composition, whereas they are extended with a new conjunct (specifying the value of the test at the branch) under a guard.
The complication arising in our handling of parallel composition through $\fork$ is that we want one single node labelled with $\fork$, even if we put in parallel many \lpof s. 

For example, by putting in parallel the singleton \lpof s made up, respectively, by nodes $y_1$ and $y_2$ and by nodes $z_1$, $z_2$ and $z_3$
(and where $\fork$ is the label of the root nodes, say $x_1$ and $x_2$ respectively), we obtain the \lpof s $\alpha$, $\beta$, and their parallel composition as follows:
\[
\alpha = 
\vcenter{\vbox{\xymatrix@R=8pt@C=4pt{
y_1 && y_2
\\
& x_1(\fork) \ar[ul]\ar[ur]
}}}
\qquad
\beta = 
\vcenter{\vbox{\xymatrix@R=8pt@C=4pt{
z_1 & z_2 & z_3
\\
& x_2(\fork) \ar[ul] \ar[ur] \ar[u]
}}}
\qquad
\alpha\parallel_x\beta =
\vcenter{\vbox{\xymatrix@R=8pt@C=4pt{
y_1 & y_2 & z_1 & z_2 & z_3
\\
&& x(\fork) \ar[ull] \ar[ul] \ar[ur] \ar[urr] \ar[u]
}}}
\]
where $x_1$ and $x_2$ disappear and the new root is $x$, labelled $\fork$.

As a side note, observe that the way in which we defined parallel composition entails that $\alpha \lelpo \alpha \parallel_x \beta$ if and only if $\alpha = \singleton \bot _x$: the \lpof\ $\alpha$ has just a node $x$ labelled with $\bot$; the second \lpof\ has node $x$ labelled with $\fork$ that precedes a node labelled with $\bot$ and all the nodes of $\beta$. 
Also, a somewhat unusual notion of associativity for parallel of \lpof s holds, viz.
$(\alpha \parallel_x \beta) \parallel_y \gamma = \alpha \parallel_y (\beta \parallel_z \gamma)$;
the expectable notion of associativity is recovered for pomsets, where we abstract from the specific node set\footnote{
We remark that $\fatsemi$ on \lpof s also satisfies a non-straightforward associativity, that however becomes the usual one when passing to pomsets.
}.
Like before, we define parallel composition for (both finite and infinite) pomsets on top of that for \lpof s:
\[
\qquad
  \Alpha \parallel \Beta \triangleq \{ \alpha \parallel_x \beta \mid \alpha \in\Alpha,\ \beta\in\Beta,\ x\not\in N_\alpha \cup N_\beta \}
\]

\begin{remark} \rm
\label{rem:parNotMon}
It is impossible for $\parallel$ to be \emph{both} associative
\emph{and} Scott continuous. If it were continuous, then $f(\Alpha) \triangleq \singleton{\ell} \parallel \Alpha$ would have a least fixed point $\mathsf{lfp}(f) = \singleton{\ell}\parallel(\singleton{\ell} \parallel \cdots)$, and, due to associativity, this pomset is infinitely branching, which invalidates \Cref{lem:fin-approx}.
\end{remark}

\section{Denotational Semantics}
\label{sec:denotational}

We now define denotational semantics for commands $C\in\cmd$ (see (\ref{fig:syntax})) based on $\pom$.
Although actions at this stage are uninterpreted, $\pom$ allows for actions with a variety of different \emph{computational effects}, which we will see in \Cref{sec:linearization}.
Those actions $a\in\act$ are drawn from a \dcpo\ $\tuple{\act, \le_\act}$ and tests $b\in\test$ are drawn from a set $\test$ (with a flat order), which is closed under conjunction, disjunction, and negation. We first define a compositional denotational model for this language in \Cref{sec:pomset-semantics}, and then in \Cref{sec:linearization} we show how to interpret pomsets as state transformers, to learn their behavior on particular inputs.

\subsection{Pomset Semantics}
\label{sec:pomset-semantics}

\begin{figure}
\begin{multicols}{2}
\begin{align*}
  \de{\skp} &\triangleq \singleton\fork \\
  \de{a} &\triangleq \singleton{a} \\ 
  \de{C_1 ; C_2} &\triangleq \de{C_1} \fatsemi \de{C_2} \\
  \de{C_1 \parop C_2} &\triangleq \de{C_1} \parallel \de{C_2} \\
  \de{\iftf b{C_1}{C_2}} &\triangleq \mathsf{guard}(b, \de{C_1}, \de{C_2})\\
  \de{\whl bC} &\triangleq \mathsf{lfp}\left(\Phi_{\tuple{C,b}}\right)
  \end{align*}
\break
\[
  \left(
  \vcenter{\vbox{\xymatrix@R=5pt@C=-1pt{
  \vdots \\
  a \ar[u] && \fork
 \\
 & b \guardarrows
 \\
 &a \ar[u] && \fork
 \\
 && b \guardarrows
}}}\right)
  \!\fatsemi a'
  =  \!
  \vcenter{\vbox{\xymatrix@R=5pt@C=-1pt{
  \vdots && a' \\
  a \ar[u] && \fork\ar[u]
 \\
 & b \guardarrows && a'
 \\
 &a \ar[u] && \fork\ar[u]
 \\
 && b \guardarrows 
}}}
\]
\end{multicols}
\caption{Left: pomset semantics of commands $\de{-} \colon \cmd \to \pom$. Right: semantics of the program $(\whl ba)\; ;\, a'$.}
\label{fig:trace-sem}
\end{figure}

We will now discuss the semantics of programs $C\in\cmd$ as pomsets, which record both the causality between atomic actions and the branching behavior of tests. The label set $\lab \triangleq \act \cup \test \cup \{ \fork, \bot\}$ consists of actions, tests, fork, and bottom nodes, which forms a pointed \dcpo, with $\bot$ as the bottom. Actions are ordered according to $\le_\act$; for any other label $\ell \in \test \cup \{\fork\}$, we only have $\bot \le \ell$ and $\ell \le\ell$. We will also henceforth use $\pom \triangleq \pom(\lab)$ to refer to pomsets over this label set.

The semantic function $\de{-} \colon \cmd \to \pom$, mapping commands to pomsets, is shown on the left of \Cref{fig:trace-sem}. Its definition is straightforward using the pomset operations; in particular, $\skp$ is interpreted as a singleton $\fork$ and
while-loops are interpreted as the least fixed point of the characteristic function $\Phi_{\tuple{C,b}} : \pom \to \pom$ defined as follows:
\[
\Phi_{\tuple{C,b}}(\Alpha) \triangleq \mathsf{guard}(b, \de{C}\fatsemi \Alpha, \de{\skp})
\]
Clearly $\Phi_{\tuple{C,b}}$ is Scott continuous, as it is defined by Scott continuous operations $\guard$ and $\fatsemi$ \Arefptwo{cor:guardCont}{cor:seqCont}. So, by Kleene's fixed point theorem, $\Phi_{\tuple{C,b}}$ has a least fixed point given by:
\[
  \de{\whl bC}
  =
  \mathsf{lfp}\left(\Phi_{\tuple{C,b}}\right)
  =
  \sup_{n\in\mathbb N} \Phi_{\tuple{C,b}}^n\left(\bot_\pom\right)
\]
where $f^0 \triangleq \mathsf{id}$ and $f^{n+1} \triangleq f \circ f^n$.
Considering a loop $\de{\whl ba}$ with a single action $a$ as its body, we can unroll the definition several times to get the following infinite chain:
\[\arraycolsep=4pt\def\arraystretch{1.25}
\begin{array}{cccccccc}
\bot
&\hspace{-1em}\lepom&
\vcenter{\vbox{\xymatrix@R=6pt@C=-2pt{
  \bot \\
  a \ar[u] && \fork
 \\
 & b \guardarrows
}}}
&\lepom&
\vcenter{\vbox{\xymatrix@R=6pt@C=-2pt{
  \bot \\
  a \ar[u] && \fork
 \\
 & b \guardarrows
 \\
 &a \ar[u] && \fork
 \\
 && b \guardarrows
}}}
& \hspace{-1em}\lepom\cdots\lepom
&
  \hspace{-1em}\vcenter{\vbox{\xymatrix@R=6pt@C=-2pt{
  \vdots \\
  ~a~ \ar[u] && \fork
 \\
 & b \guardarrows
 \\
 &a \ar[u] && \fork
 \\
 && b \guardarrows
}}}
\\
\Phi_{\tuple{a, b}}^0(\bot_\pom)
&&
\Phi_{\tuple{a, b}}^1(\bot_\pom)
&&
\Phi_{\tuple{a, b}}^2(\bot_\pom)
&&
\hspace{-1em}\sup_{n\in\mathbb{N}} \Phi_{\tuple{a,b}}^n(\bot_\pom)
\end{array}
\]
The supremum of this chain is clearly an infinite structure with a terminating branch for each $n\in \mathbb N$, and an infinitely ascending spine. The $\bot$ node is pushed to a progressively higher level after each unrolling, so in the supremum it does not appear at all.

On the right of \Cref{fig:trace-sem}, we show $\de{(\whl ba)\ ;\, a'}$, containing a loop with a single action in the body followed by another action.
As discussed in \Cref{sec:seq}, sequentially composing another action $a'$ after a loop produces a copy of $a'$ for each branch, resulting in countably many copies of $a'$. This approach is necessary to make $a'$ finitely preceded.

The semantics obeys a binary branching property, guaranteeing that test nodes branch into exactly two successors, one where the test passes and another where it fails. Further, variables can only be introduced into formulae after a test. From now on, we assume that all pomsets have this property, which is clearly preserved by the operations of \Cref{sec:operations}.

\begin{definition}[Binary Branching]\label{def:binary-branching}
  A pomset $\Alpha\in\pom$ has the \emph{binary branching property} if, for every $\alpha\in\Alpha$ and every $x\in N_\alpha$ such that $\lambda_\alpha(x) \in \test$, $\succ_\alpha(x) = \{ y_1, y_2 \}$ such that:
  \begin{enumerate*}
  \item $\varphi_\alpha(y_1) \Leftrightarrow \varphi_\alpha(x) \land x$;
  \item $\varphi_\alpha(y_2) \Leftrightarrow \varphi_\alpha(x) \land \lnot x$; and
  \item $\pred_\alpha(y_1) = \pred_\alpha(y_2) = \{x \}$.
  \end{enumerate*}
\noindent
Furthermore, if $x$ is not the successor of a test, then $\varphi_\alpha(x) \Leftrightarrow \bigwedge_{y\in\pred_\alpha(x)} \varphi_\alpha(y)$.
\end{definition}

\subsection{Linearization}
\label{sec:linearization}

We now discuss how to produce a state transformer from the pomset semantics of \Cref{sec:pomset-semantics}, which is useful for understanding the input-output behavior of programs resulting from scheduling concurrent threads.
This state transformer is obtained via \emph{linearization}, which consists of considering all interleavings of the threads in order to interpret the semantics of a \emph{complete} program as a function from inputs to collections of outputs. Once the program is linearized, more threads cannot be composed in parallel, as the causality information is gone.
Nevertheless, linearization is useful to understand the program's behavior, and relate the behavior to other semantic models, as shown in \Cref{sec:powerdomains,sec:convex-powerset}.

When performing linearization, we will interpret programs in a computational domain $D$, which must support nondeterminism---in order to model the different ways that threads can be interleaved---but may support additional computational effects too, which we will model using monads. Hence, we assume some familiarity with basic notions of monads and Kleisli categories; for a detailed introduction, we refer to \cite{awodey2006category}. More precisely, we need the following:

\begin{enumerate}[leftmargin=*]

\item A domain $D$ and set of states $\st$, such that $\tuple{D(\st),\sqsubseteq}$ is a pointed \dcpo\ with bottom $\bot_D$.
\item  An associative, commutative, and monotone nondeterminism operator $\mathord{\nd} \colon D(\st)^2 \to D(\st)$.
\item An additive monad in the category $\textbf{\dcpo}$ $\tuple{D, \eta, (-)^\dagger}$, meaning that:
\begin{enumerate}
\item the operations
$
  \eta \colon X\to D(X)
$
and
$
  (-)^\dagger \colon (X \to D(Y)) \to D(X) \to D(Y)
$
obey the monad laws $f^\dagger \circ \eta = f$, $\eta^\dagger = \mathsf{id}$, and $(g^\dagger \circ f)^\dagger = g^\dagger \circ f^\dagger$.
\item Kleisli extension is Scott Continuous (\ie $\sup_{f\in D}\sup_{d\in D'} f^\dagger(d) = (\sup D)^\dagger(\sup D')$),
strict (\ie $f^\dagger(\bot_D) = \bot_D$), and additive (\ie $f^\dagger(d \nd d') = f^\dagger(d) \nd f^\dagger(d')$) \cite{outcome,zilberstein2024outcome}.
\end{enumerate}

\item An interpretation function for actions $\de{-}_\act \colon \act \to \st \to D(\st)$ that is monotone (\ie $\de{a}_\act\!(s) \sqsubseteq \de{a'}_\act\!(s)$ if $a \le_\act a'$) and one for tests $\de{-}_\test \colon \test \to \st \to \mathbb B$.
\end{enumerate}
There are many domains that obey these properties, depending on which \emph{computational effects} are present. These include: the Hoare, Smyth, and Plotkin powerdomains for nondeterministic computation \cite{plotkin1976powerdomain,smyth1978power,abramsky1995domain} (where the latter two additionally deal with nontermination), the convex powerset \cite{morgan1996refinement,jifeng1997probabilistic} and the powerdomain of indexed valuations \cite{varacca2002powerdomain,varacca2003probability,varacca_winskel_2006} (for mixed probabilistic and nondeterministic computation), and other domains for nondeterminism with exceptions \cite{outcome,zilberstein2024outcome}. In \Cref{thm:convex,thm:powdom}, we will show that standard semantics in two of those domains are recovered from our pomset semantics.

\begin{figure*}
\begin{align*}
  \mathsf{next}(\alpha, \psi, S) &\triangleq \{ x \in N_\alpha \setminus S \mid \predplus{x}_\alpha \subseteq S, \psi \Rightarrow \varphi_\alpha(x) \}
  \\
  \linlpo(\alpha,\psi, S)(s) &\triangleq \eta(s) \quad\text{if}~ \mathsf{next}(\alpha, \psi, S) = \emptyset
  \\
  \linlpo(\alpha,\psi, S)(s) &\triangleq
    \smashoperator[l]{\bignd_{x \in \mathsf{next}(\alpha, \psi, S)}}\!\! \left\{
      \def\arraystretch{1.2}
      \arraycolsep=2pt
      \begin{array}{ll}
        \linlpo\left(\alpha,\psi, S \cup\{ x \}\right)^\dagger\left(\de{a}_\act(s)\right) & \text{if} ~ \lambda_\alpha(x) =a\in\act
        \\
        \linlpo(\alpha,\psi \land \sem{x = {\de{b}}_\test(s)}, S \cup\{x\})(s) & \text{if}~ \lambda_\alpha(x)=b\in\test
        \\
        \bot_D & \text{if}~ \lambda_\alpha(x) = \bot
        \\
        \linlpo(\alpha,\psi, S\cup\{x\})(s) & \text{if} ~ \lambda_\alpha(x) = \fork
      \end{array}
    \right.
\end{align*}
\caption{Linearization for finite \lpof s $\linlpo \colon \lpo_\fin(\lab)\times\Form\times\mathcal{P}(\Nodes) \to \st \to D(\st)$.}
\label{fig:linlpo}
\end{figure*}

We will start by defining a linearization operation on finite \lpof s, which is shown in \Cref{fig:linlpo}.
Linearization is defined recursively, until there are no more nodes to schedule. It must be defined on finite structures, otherwise the recursion is not well-founded, but we can extend linearization to infinite structures using the extension lemma.
In $\linlpo(\alpha, \psi, S)$, the set $S\subseteq N_\alpha$ contains all the nodes that have already been processed and $\psi$ is a path condition, indicating the outcomes of the tests associated to the nodes in $S$.

The function $\next(\alpha, \psi, S)$ gives the nodes that are ready to be scheduled, which includes all nodes that obey the path condition $\psi$, and whose predecessors are in $S$. If $\next(\alpha, \psi, S)$ is empty, then $\linlpo(\alpha,\psi,S)$ simply returns the current state $s$ using the monad unit $\eta$. If not, then it nondeterministically selects a next node, where $\bignd_{i \in I} d_i \triangleq d_{i_1} \nd \cdots \nd d_{i_n}$, for any finite index set $I = \{ i_1, \ldots, i_n \}$. If the next node is an action, then it interprets the action using $\de{-}_\act$, and then uses Kleisli composition to compose the result with the linearization of the remaining \lpof. If the next node is a test, then that test is evaluated and the result is added to the path condition, where $\sem{x = 1} \triangleq x$ and $\sem{x=0} \triangleq \lnot x$. If the next node is $\bot$, then the linearization is $\bot_D$, and $\fork$ is treated like a no-op.

We next use $\linlpo$ to define linearization on finite pomsets $\linfin \colon \pom_\fin(\lab) \to \st \to D(\st)$ as 
$\linfin([\alpha]) \triangleq \linlpo(\alpha, \tru, \emptyset)$.
Clearly $\linlpo(\alpha, \tru, \emptyset) = \linlpo(\beta, \tru, \emptyset)$ if $\alpha \equiv \beta$, so $\linfin(\Alpha)$ can be defined for any arbitrary representative \lpof\ $\alpha \in \Alpha$. Finally, linearization is extended to infinite pomsets $\lin\colon \pom \to \st \to D(\st)$ by taking the supremum over all of the finite approximations: 
$\lin(\Alpha) \triangleq \linfin^*(\Alpha)$.
Linearization over finite pomsets is monotone \Arefp{lem:monLin}, therefore the extension is well-defined and Scott continuous.

To ensure that linearization acts as desired, we provide a sanity check lemma to relate the linearized semantics to well-known counterparts. Linearizing $\skp$ gives the monad unit; linearizing a singleton action has the same behavior as interpreting the action; linearizing a sequential composition is equal to Kleisli composition of the individual linearizations;
linearizing an if-statement is equal to linearizing one of the two branches, depending on the truth of the guard; and
the linearization of a while-loop is equal to the least fixed point over a different characteristic function $\Psi_{\tuple{f, b}}$. The function $\Psi_{\tuple{f, b}}$ inherits Scott continuity from the Kleisli composition operator, therefore the fixed point exists.

\begin{restatable}[Linearization]{lemma}{linProps}\label{lem:lin-props}
The following properties hold:
\rm
\begin{mathpar}
  \lin(\de{\skp}) = \eta
  
  \lin(\de{a}) = \de{a}_\act

  \lin(\de{C_1; C_2}) = \lin(\de{C_2})^\dagger \circ \lin(\de{C_1})

  \lin(\de{\iftf b{C_1}{C_2}})(s)
  = \left\{
    \begin{array}{ll}
      \lin(\de{C_1})(s) & \text{if}~ \de{b}_\test(s) = 1
      \\
      \lin(\de{C_2})(s) & \text{if}~ \de{b}_\test(s) = 0
    \end{array}
  \right.

  \lin(\de{\whl bC}) = \mathsf{lfp}\left( \Psi_{\tuple{\lin(\de{C}),b}} \right)
  ~\text{where}~
  \Psi_{\tuple{f, b}}(g)(s) \triangleq \left\{
    \arraycolsep=2pt
    \begin{array}{ll}
      g^\dagger(f(s)) & \text{if}~ \de{b}_\test(s) = 1
      \\
      \eta(s) & \text{if}~ \de{b}_\test(s) = 0
    \end{array}
  \right.
\end{mathpar}
\end{restatable}
\Cref{lem:lin-props} is quite significant; it implies that $\lin(\de{C})$ corresponds to standard monadic semantics in a variety of domains. In particular, it brings us to our first adequacy theorem, showing that our pomset model recovers standard probabilistic semantics $\de{C}_\C \colon \st \to \C(\st)$ (shown in \Cref{fig:convex-powerset-semantics} of \Cref{app:convex}) \cite{jifeng1997probabilistic,zilberstein2025demonic,mciver2005abstraction} based on the convex powerset for a parallel-free fragment of our programming language.

\begin{theorem}\label{thm:convex}
For any program $C\in\cmd$ without parallel composition: $\lin(\de{C}) = \de{C}_\C$.
\end{theorem}
\begin{proof}
By induction on the structure of $C$. The cases follow immediately from \Cref{lem:lin-props}.
\end{proof}

%

\section{Pomset Languages and Purely Nondeterministic Concurrency}
\label{sec:powerdomains}

Powerdomains give semantics to programs that combine nondeterminism with looping or recursion \cite{plotkin1976powerdomain,smyth1978power}. In this section, we create a powerdomain instantiation of our semantics from \Cref{sec:denotational} and prove that the resulting model corresponds to well-known models of non-probabilistic concurrent programs \cite{kappe2020concurrent,laurence2014completeness}.
We work in the \emph{Hoare} or \emph{lower} powerdomain, named as such for its connection to partial correctness and Hoare Logic \cite{hoarelogic,FLOYD67}. Indeed, the Hoare powerdomain identifies the terminating traces of a program, but not that the program \emph{always} terminates. From a concurrency perspective, this corresponds to \emph{safety} properties \cite{lamport1977proving}. For simplicity, we presume that the order over actions $\le_\act$ in this section is flat.


Since typical pomsets do not contain formulae \cite{gischer1988equational,pratt1986modeling,grabowski1981partial}, they cannot encode control flow branching introduced by if-statements and while-loops. As we already saw in \Cref{sec:convex-powerset}, concurrent semantics typically use \emph{pomset languages}---sets of pomsets---where each individual pomset contains only the actions that occur in a single branch of computation \cite{kappe2020concurrent,laurence2014completeness}.

We emulate pomset language semantics using pomsets with formulae. First, we define a label set $\lab_\powdom \triangleq \act_\powdom \cup \test_\powdom \cup \{ \fork, \bot\}$, whose components are as before, but now $\test_\powdom \triangleq \emptyset$, and instead $\act_\powdom \triangleq \act \cup \{ \assume b \mid b \in \test \}$ includes an action $\assume b$ for every test $b\in\test$, which guarantees that a certain condition holds in the current branch.
Pomset languages $\pomlang \triangleq \mathcal{P}(\pom_\fin(\lab_\powdom))$ are sets of finite pomsets over $\lab_\powdom$.

The pomset language semantics $\de{-}_\powdom \colon \cmd \to \pomlang$, which is standard \cite{kappe2020concurrent,laurence2014completeness}, is given in \Cref{fig:pomlang-semantics} of \Cref{app:plsem}. In this model, 
each $\Alpha \in \de{C}_\powdom$ corresponds to a particular sequence of test resolutions. When an $\code{assume}$ fails, the trace is eliminated. When a branch occurs, \ie in an if-statement or while-loop, the set of traces is duplicated to account for both the ``true'' and ``false'' branches, as we saw in \Cref{fig:pomset-example}. Loops are interpreted as the (possibly infinite) union of all finite traces, so infinite pomsets are not needed. 


The Hoare powerdomain consists of sets of states, ordered by subset inclusion $\tuple{\mathcal{P}(\st), \subseteq}$, which is known to be a pointed \dcpo\ with supremum given by union $\cup$ and $\emptyset$ as bottom. Let $\nd \triangleq \cup$ and define the monad operations as $\eta(s) \triangleq \{ s \}$ and $f^\dagger(S) \triangleq \bigcup_{s\in S} f(s)$, which are well known to be Scott continuous, strict, and additive \cite{outcome}.
Accordingly, action evaluation
$\de{-}_{\act_\powdom} \colon \act_\powdom \to \st \to \mathcal P(\st)$
is given below, where $\code{assume}$ actions result in $\{s\}$, if the test is true in state $s$, and $\emptyset$, otherwise. 
We also define a specialized linearization operation $\linpl \colon \pomlang \to \st \to \mathcal{P}(\st)$, which is simply a union over the linearization of all traces.
\begin{align*}
  \de{a}_{\act_\powdom}\!(s) &\triangleq \left\{
  \arraycolsep=2pt
    \begin{array}{ll}
    \de{a}_\act(s) & \text{if} ~ a\in \act
    \\
    \{s\} & \text{if}~ a = \assume b,\ \de{b}_\test\!(s) = 1
    \\
    \emptyset & \text{if}~ a = \assume b,\ \de{b}_\test\!(s) = 0
    \end{array}
  \right.
  &\quad
  \linpl(S)(s) &\triangleq \bigcup_{\Alpha \in S} \linfin(\Alpha)(s)
\end{align*}
Since all formulae are $\tru$, the next nodes to schedule are exactly those whose predecessors have all been processed, \ie $\next(\alpha, \tru, S) = \{ x \in N_\alpha\setminus S \mid \predplus{x}_\alpha \subseteq S \}$. This corresponds to standard interleaving definitions of linearization \cite{pratt1986modeling}.
We now prove that pomset language semantics $\de{-}_\powdom$ (\Cref{fig:pomlang-semantics}) corresponds exactly to the Hoare powerdomain instance of our semantic model from \Cref{sec:denotational}. The proof relies on a translation $\tr \colon \pom \to \pomlang$, which recovers a pomset language from a pomset with formulae, and is defined in \Aref{app:powdom}.

\begin{restatable}[Equivalence of Semantics]{theorem}{powDomChar}
\label{thm:powdom}
The following diagram commutes:
\vspace*{-.2cm}
\[
\xymatrix@R=32pt@C=64pt{
{\cmd} \ar[r]^{\de{-}} \ar[d]_{\de{-}_\powdom}
&
{\pom} \ar[d]^{\lin} \ar[dl]_{\tr}
\\
{\pomlang} \ar[r]_{\linpl}
&
{(\mathcal S \to \mathcal{P}(\mathcal S))}
}
\]
\end{restatable}
The upper commuting triangle does not depend on the Hoare powerdomain, so it is tempting to say that pomset languages give an adequate model in other domains too. However, as we showed in \Cref{sec:convex-powerset}, just because a semantic structure exists does not mean that it conveys the desired meaning.
Indeed, the fact that $\lin = \lin_\powdom \circ \tr$ relies on two particular properties of the Hoare powerdomain. First, if $h(x) = f(x) \nd g(x)$, then $h^\dagger(S) = f^\dagger(S) \nd g^\dagger(S)$, which is invalid in probabilistic domains. In fact, this corresponds to the problematic \Cref{eq:problem} from \Cref{sec:convex-powerset}.
Second, $\lin_\powdom$ is computed as a union over a possibly infinite set (the union is infinite whenever the program contains a while-loop). Unbounded nondeterminism is known to cause problems in many domains---including the Smyth and Plotkin powerdomains \cite{apt1986countable,back1980semantics,s_ondergaard1992non}---thus no infinitary version of $\nd$ exists.

So, although programs in any domain \emph{can} be interpreted as pomset languages, those semantic objects may not give an adequate meaning to the program. Particularly in probabilistic domains---\ie the convex powerset---prior pomset language models \cite{kappe2020concurrent,laurence2014completeness} do not give us a way to find the relative probabilities of the different outcomes after actually running the program.
Our new approach, using a single structure that records both causality and control flow, is more suitable for interpreting the semantics of concurrent programs in new domains, such as probabilistic computation. As we saw in \Cref{sec:convex-powerset} and \Cref{thm:convex}, pomsets with formulae accurately capture the concurrent behavior of these programs, while also allowing us to recover the precise probabilities of each outcome.



\section{Discussion and Related Work}\label{sec:discussion}

We developed a new semantic structure---pomsets with formulae---and studied its domain-theoretic properties, which we used as basis to provide semantics for a probabilistic concurrent imperative language with unbounded loops. Our work builds on a rich history of using partial orders and pomsets in denotational models of concurrent programs \cite{grabowski1981partial,pratt1976semantical,gischer1988equational}, which have been used as a semantic basis for concurrent separation logic \cite{brookes2004semantics}, concurrent Kleene algebras \cite{hoare2011concurrent,kappe2020concurrent,laurence2014completeness,jipsen2016concurrent}, and other applications \cite{kavanagh2018denotational,edixhoven2022branching,jagadeesan2020pomsets}.

While there are links between pomset semantics and operational models of concurrency \cite{brookes2002traces}, pomsets are often preferred for their finer notion of concurrency compared to the interleaving semantics offered by operational models. This is especially useful in the context of weak memory \cite{kavanagh2018denotational,jagadeesan2020pomsets}.
Probabilistic event structures model probabilistic process algebras \cite{katoen1996quantitative,winskel2014probabilistic,varacca2006probabilistic,varacca2007probabilistic}, whereas we model a full imperative language with sequential composition, control flow, and unbounded loops, for which adding formulae was essential.


Concurrency has previously been combined with other computational effects in limited ways. For example, semantics for probabilistic concurrent programs have been defined in both operational \cite{polaris,tassarotti2018verifying} and denotational \cite{zilberstein2024probabilistic} styles, but these approaches are limited to programs with bounded looping constructs.
The inclusion of unbounded looping adds significant complexity in the probabilistic case. Whereas finite-trace models are sufficient for many non-probabilistic scenarios, unbounded looping in probabilistic programs requires infinite traces, as the probability of termination may only become 1 in the limit. A proper theory of infinite traces requires pomsets to be enriched with a \dcpo\ \cite{meyer1988applications,meyer1989pomset} or metric \cite{bakker1990metric} structure; prior work in this area was very informative to our own development, although the inclusion of deterministic branching in our own pomset structure added significant complexity.

An operational model of probabilistic concurrency has been developed based on Markov Decision Processes \cite{fesefeldt2022towards}. This semantics is used to lower bound expected values, but does not straightforwardly extend to more complex domains such as the convex powerset, which give the full set of distributions over outcomes.
In addition, a denotational model has been developed \cite{neves2024adequacy}, where the underlying semantic structure is obtained as the solution to a domain equation. Our construction is more concrete, giving a full account of the causality in the program and thus capturing non-interleaving models of concurrency too.

Branching pomsets were recently introduced to model choices in choreographic programs  \cite{edixhoven2022branching,edixhoven2022realisability,edixhoven2024branching}. While branching pomsets can, in theory, contain infinitely many nodes, the domain theoretic properties needed to approximate infinite structures have not been explored, which was a significant focus of this paper. Indeed, our extension lemma (\Cref{lem:extension}) was necessary for linearization, without which we would not have been able to relate our pomeset model to the known convex powerset semantics  \cite{jifeng1997probabilistic,mciver2005abstraction,zilberstein2025demonic}.

Going forward, it would be interesting to explore the applicability of pomsets with formulae to convex powerdomain constructions other than Smyth \cite{tix2009semantic,tix2000convex,keimel2017mixed,tix1999continuous}, or to other domains for mixing probabilities and nondeterminism including indexed valuations \cite{varacca2003probability,varacca_winskel_2006,varacca2002powerdomain} and multisets of distributions \cite{jacobs2021multisets,kozen2024multisets}. While our linearization procedure is based on an \emph{omniscient} scheduler, we are also interested in exploring more restricted models including oblivious schedulers (which cannot see the outcomes of random sampling), and fair schedulers. This will be challenging, as those models are non-compositional and therefore linearization on infinite pomsets could not be defined via \Cref{lem:extension}. Still, pomsets with formulae are a good basis for studying these questions, as the structure itself makes no assumptions about how scheduling of parallel branches occurs.


\bibliography{refs}

\appendix

\section{Preliminaries and Omitted Definitions}
\label{app:prelim}

\subsection{Order Theory and Domain Theory}
\label{app:ord}

\begin{definition}[Poset]
A partially ordered set, or {\em poset}, is a set equipped with a partial order: that is, a pair $\tuple{X,\le}$ consisting of a set $X$ and an order relation $\mathord{\le} \subseteq X\times X$ which is reflexive, transitive and antisymmetric. When the order is instead irreflexive (\ie $x \not< x$ for all $x\in X$), the poset is called {\em strict} and the order relation is denoted by $<$.
\end{definition}

Given a strict poset $\tuple{X, <}$ and $x \in X$, we define the upward and downward closures:
\begin{align*}
  \succplus x &\triangleq \{ y \in X  \mid  x < y \}
  &
   \predplus x &\triangleq \{ y \in X  \mid  y < x \};
  \end{align*}
and  its set of immediate successors and predecessors as
  \[
   \succ(x) \triangleq \{ y \in \succplus x\ \mid\ \not\exists z. (x < z < y) \} \qquad\qquad
    \pred(x) \triangleq \{ y \in \predplus x\ \mid\ \not\exists z. (y < z < x) \}
      \]
      All the above operations can easily be extended to sets of elements, \eg $\succplus X \triangleq \bigcup_{x\in X} \succplus x$, for any set $X$.
  The set of minimal and maximal nodes of a poset is given by:
        \begin{align*}
  \max \triangleq \{ x \in X \mid \succ(x) = \emptyset \} 
   \qquad
 \min  \triangleq \{ x \in X \mid \pred(x) = \emptyset \}
\end{align*}
To define our new structures we need a few more definitions on posets that will play a role in proving the existence of fixed points, essential for giving a semantics to unbounded loops. First, we recall the notion of downward closed sets.

\begin{definition}[Downward Closure]
Given a poset $\tuple{X, <}$, a set $Y \subseteq X$ is downward closed, written $Y \dwclosed X$, iff for all $y\in Y$ it holds that $\predplus y\, \subseteq Y$.
\end{definition}

Second, we need a notion of a poset being finitely preceded. 

\begin{definition}[Finitely Preceded]
A poset $\tuple{X, <}$ is \emph{finitely preceded} if there are finitely many elements smaller than each element; that is: $|\predplus{x}| < \infty$, for all $x\in X$.
\end{definition}

Third, we define the level of elements in the poset.

\begin{definition}[Level \cite{meyer1989pomset}]
  Given a finitely preceded poset $\tuple{X, <}$, define $\lev \colon X \to \mathbb N$:
  \[
  \lev(x) \triangleq
\sup \left\{ n \ \middle| \!\!
    \begin{array}{l}
    \exists x_0, \ldots, x_n \in X.\ x_0\in\min\ \wedge\  x_n = x \ \wedge\ \forall i < n.\ x_{i+1} \in \succ(x_i) \!\!\!
    \end{array}
    \right\}
  \]
  We also define its inverse $\lev^{-1} \colon \mathbb N \to \mathcal{P}(X)$:
  \[
    \lev^{-1}(n) \triangleq \{ x \in X \mid \lev(x) = n \}
  \]
\end{definition}

\begin{definition}\label{def:formulae}
The satisfaction relation $\mathord{\vDash} \subseteq (\Nodes \to \mathbb B) \times \Form$ for Boolean formulae is given below:
\begin{align*}
  v &\vDash \tru && \text{always}
  \\
  v &\vDash \fls && \text{never}
  \\
  v &\vDash \psi_1 \land\psi_2 && \text{iff} \quad v\vDash \psi_1 ~\text{and}~ v\vDash \psi_2
  \\
  v &\vDash \psi_1 \lor\psi_2 && \text{iff} \quad v\vDash \psi_1 ~\text{or}~ v\vDash \psi_2
  \\
  v &\vDash \lnot \psi && \text{iff} \quad v\not\vDash \psi
  \\
  v &\vDash x && \text{iff} \quad v(x) = 1
\end{align*}
\end{definition}

\subsection{The Convex Powerdomain}
\label{app:convex}

We begin by providing the omitted definitions from \Cref{sec:convex-powerset}, regarding the \dcpo\ structure of the convex powerset. We start by defining the order on probability distributions. For any set $X$, let $X_\bot = X \cup \{\bot\}$ and the order over distributions $\mathord{\sqsubseteq_\D} \subseteq \D(X_\bot) \times \D(X_\bot)$ be defined as follows:
\[
  \mu \sqsubseteq_\D \nu
  \quad\text{iff}\quad
  \forall x\in X.\ \mu(x) \le \nu(x)
\]
Note that this means that $\mu(\bot) \ge \nu(\bot)$ if $\mu \sqsubseteq_\D \nu$, and that the bottom of the order is $\bot_\D = [ \bot \mapsto 1]$. A set of distributions $S \subseteq \D(X_\bot)$ is \emph{upward closed} if, for all $\mu \in S$, it holds that $\mu \sqsubseteq_\D \nu$ implies $\nu \in S$. In addition, $S$ is \emph{Cauchy closed} if it is closed in the product of Euclidean topologies \cite[Definition 5.4.3]{mciver2005abstraction}. Now, the convex powerset is defined as follows:
\[
  \C(X) \triangleq \set{ S \subseteq \D(X_\bot) \;\;\middle|\;\; \text{$S$ is nonempty, convex, upward closed, and Cauchy closed}  }
\]
Convex powersets are ordered via the Smyth order \cite{smyth1978power} $\mathord{\sqsubseteq_\C} \subseteq \C(X) \times \C(X)$, defined below:
\[
  S \sqsubseteq_\C T
  \quad\text{iff}\quad
  \forall \nu \in T.\ \
  \exists \mu \in S.\ \
  \mu \sqsubseteq_\D \nu
\]
Since the sets above are upward closed, the Smyth order collapses to reverse subset inclusion $S \sqsubseteq_\C T$ iff $S \supseteq T$. This makes the \emph{convex powerdomain} $\tuple{\C(X), \sqsubseteq_\C}$ a pointed \dcpo\ with suprema given by set intersection and bottom $\bot_\C = \D(X_\bot)$ being the set of all distributions, for the full proof refer to \cite{zilberstein2025demonic}.

\begin{figure}
\begin{align*}
  \de{\skp}_\C\!(s) &\triangleq \eta(s)
  \\
  \de{C_1; C_2}_\C\!(s) &\triangleq \de{C_2}_\C^\dagger\left(\de{C_1}_\C(s)\right)
  \\
  \de{\iftf b{C_1}{C_2}}_\C\!(s) &\triangleq \left\{
    \begin{array}{ll}
      \de{C_1}_\C\!(s) & \text{if}~ \de{b}_\test(s) = 1
      \\
      \de{C_2}_\C\!(s) & \text{if}~ \de{b}_\test(s) = 0
    \end{array}
  \right.
  \\
  \de{\whl bC}_\C\!(s) &\triangleq \mathsf{lfp}(\Psi_{\tuple{\de{C}_\C, b}})(s)
  \\
  \de{a}_\C\!(s) &\triangleq \de{a}_\act(s)
\end{align*}
\caption{Convex powerset semantics $\de{-}_\C \colon \cmd \to \mathcal{S} \to \C(\mathcal S)$ due to \cite{jifeng1997probabilistic,zilberstein2025demonic,mciver2005abstraction}, 
where $\Psi$ is defined in \Cref{lem:lin-props}.}
\label{fig:convex-powerset-semantics}
\end{figure}

Finally, we give the definitions for the monad operations \cite{jacobs2008coalgebraic}: unit $\eta \colon X \to \C(X)$ and Kleisli extension $(-)^\dagger \colon (X \to \C(Y)) \to \C(X) \to \C(Y)$.
\begin{align*}
  \eta(x) &\triangleq \left[ x \mapsto 1 \right]
  &
  f^\dagger(S) & \triangleq \set{
    \sum_{x\in\supp(\mu)} \mu(x)\cdot \nu_x
    \;\;\middle|\;\;
    \mu \in S, \forall x \in \supp(\mu).\ \nu_x \in f_\bot(x)
  }
\end{align*}
where $(-)_\bot \colon (X \to \C(Y)) \to X_\bot \to \C(Y)$ is defined as follows:
$f_\bot(x) \triangleq f(x)$ if $x\in X$ and $f_\bot(\bot) \triangleq \bot_\C$ otherwise.

As is required in \Cref{sec:linearization}, it is proven in \cite{zilberstein2025demonic} that the Kleisli extension for $\C$ defined above is Scott continuous and additive:
\begin{align*}
  \sup_{f\in D, x\in D'} f^\dagger(S) &= (\sup D)^\dagger(\sup D')
  &
  f^\dagger(S \nd T) &= f^\dagger(S) \nd f^\dagger(T)
\end{align*}
Based on these operations, \Cref{fig:convex-powerset-semantics} shows the convex powerset semantics  for the parallel-free fragment of the language in (\ref{fig:syntax}), $\de{-}_\C \colon \cmd \to \st \to \C(\st)$. This is precisely the same semantics used in prior work for reasoning about programs that are both nondeterministic and probabilistic \cite{jifeng1997probabilistic,mciver2005abstraction,zilberstein2025demonic}. This semantics does not capture parallel computation, which requires the more sophisticated semantic domains that we develop in this paper.
The Scott continuity property above ensures that $\Psi_{\tuple{\de{C}_\C, b}} \colon (\st \to \C(\st)) \to \st \to \C(\st)$ is Scott continuous, and therefore the semantics of loops is well defined.

\subsection{Pomset Languages}
\label{app:plsem}

As we saw in \Cref{sec:powerdomains}, a pomset language is a set of finite pomsets without any tests or formulae. Control flow is determined by $\code{assume}$ actions, which keep or eliminate traces depending on the whether the corresponding test passes or fails.
\[
  \pomlang \triangleq \mathcal{P}(\pom_\fin(\lab_\powdom))
\]
We give the pomset language semantics for a nondeterministic language in \Cref{fig:pomlang-semantics}.

\begin{figure}
\begin{align*}
  \de{\skp}_{\powdom} &\triangleq \{ \singleton\fork \} 
 \\
 \de{a}_\powdom  &\triangleq \{\singleton{a}\}
  \\
  \de{C_1 ; C_2}_{\powdom} &\triangleq \de{C_1}_{\powdom} \fatsemi\de{C_2}_{\powdom}
  \\
  \de{C_1 \parop C_2}_{\powdom} &\triangleq \de{C_1}_{\powdom} \parallel \de{C_2}_{\powdom}
  \\
  \de{\iftf b{C_1}{C_2}}_{\powdom} &\triangleq
      \left(\set{ \singleton{\assume b} } \fatsemi \de{C_1}_\powdom \right) \cup 
      \left(\set{ \singleton{\assume{\lnot b}} } \fatsemi \de{C_2}_\powdom \right)
  \\
  \de{\whl bC}_{\powdom} &\triangleq \mathsf{lfp}\left(\Xi_{\tuple{C,b}}\right)
 \end{align*}
\begin{align*}
\text{where} \quad \Xi_{\tuple{C,b}}(S) &\triangleq
    \left( \set{ \singleton{\assume b}} \fatsemi \de{C}_{\powdom} \fatsemi S \right) \cup
    \left( \{ \singleton{\assume{\lnot b}} \} \fatsemi \de{\skp}_\powdom \right)
\\
S \fatsemi T &\triangleq \{ \Alpha \fatsemi\Beta \mid \Alpha\in S, \Beta \in T\}
\\
S \parallel T &\triangleq \{ \Alpha \parallel\Beta \mid \Alpha\in S, \Beta \in T\}
\end{align*}
\caption{Standard pomset language semantics $\de{-}_{\powdom} \colon \cmd \to \pomlang$.}
\label{fig:pomlang-semantics}
\end{figure}

\ifx\extended\undefined\else

\allowdisplaybreaks
%

\section{Operations on Pomset with Formulae: Examples}

\subsection{Sequential Composition}
\label{app:seq-example}

\begin{figure}[t]
\[\arraycolsep=3pt
\begin{array}{ccccc}
\alpha_1 =\!\!
\vcenter{\vbox{
\xymatrix@R=8pt@C=-6pt{
z_1(\bot) && z_2(\bot) && z_3 && z_4
\\
& y_1 \guardarrows
  &&&& y_2 \guardarrows
\\
&&& x \ar[ull]\ar[urr]
}}}
&\lelpo&
\alpha_2 =\!\!
\vcenter{\vbox{
\xymatrix@R=8pt@C=-6pt{
z_1(\bot) && z_2 && z_3 && z_4
\\
& y_1 \guardarrows
&&&& y_2 \guardarrows
\\
&&& x \ar[ull]\ar[urr]
}}}
&
\lelpo 
&
\alpha_3 =\!\!
\vcenter{\vbox{
\xymatrix@R=8pt@C=-2pt{
z_1 && z_2 && z_3 && z_4
\\
& y_1 \guardarrows
&&&& y_2\guardarrows
\\
&&& x \ar[ull]\ar[urr]
}}}
\\
\stuck_{\alpha_1} = y_1 \lor \lnot y_1 \Leftrightarrow \tru
&
\Leftarrow
&
\stuck_{\alpha_2} = y_1
&
\Leftarrow
&
\stuck_{\alpha_3} = \fls
\\
\extens_{\alpha_1} = \emptyset
&\subseteq&
\extens_{\alpha_2} = N_{\alpha_2} \setminus \{ z_1 \}
&\subseteq&
\extens_{\alpha_3} = N_{\alpha_3}
\\
\br_{\alpha_1} = \emptyset
& \subseteq &
\br_{\alpha_2} = \left\{\arraycolsep=1pt\begin{array}{ccr} \lnot y_1 &\land& y_2, \\ \lnot y_1 &\land& \lnot y_2 \phantom{,} \end{array}\right\}
& \subseteq &
\br_{\alpha_3} = \left\{\arraycolsep=1pt\begin{array}{rcr} y_1 &\land& y_2, \\ y_1 &\land& \lnot y_2, \\ \lnot y_1 &\land& y_2, \\ \lnot y_1 &\land& \lnot y_2\phantom{,} \end{array}\right\}
\end{array}
\]
\begin{align*}
\alpha_1 \fatsemi w
&
=
\vcenter{\vbox{
\xymatrix@R=5pt@C=-3pt{
& z_1(\bot)
&& z_2(\bot)
&& z_3
&& z_4
\\
&& y_1 \guardarrows &&&& y_2 \guardarrows
\\
&&&& x \ar[ull]\ar[urr]
}}}
\qquad
\alpha_2 \fatsemi w
=
\vcenter{\vbox{
\xymatrix@R=5pt@C=-3pt{
&  &&  && w_{\lnot y_1 \land y_2} && w_{\lnot y_1 \land\lnot y_2}
\\\\
& z_1(\bot)
&& z_2 \ar@/_/[uurr]\ar@/_/[uurrrr]
&& z_3 \ar[uu]
&& z_4 \ar[uu]
\\
&& y_1 \guardarrows &&&& y_2 \guardarrows
\\
&&&& x \ar[ull]\ar[urr]
}}}
\\
\alpha_3 \fatsemi w
&
=
\vcenter{\vbox{
\xymatrix@R=5pt@C=0pt{
& w_{y_1 \land y_2} && w_{y_1 \land \lnot y_2} && w_{\lnot y_1 \land y_2} && w_{\lnot y_1 \land\lnot y_2}
\\\\
& z_1\ar[uu]\ar@/_/[uurr]
&& z_2 \ar@/_/[uurr]\ar@/_/[uurrrr]
&& z_3 \ar@/^/[uullll]\ar[uu]
&& z_4 \ar@/^/[uullll]\ar[uu]
\\
&& y_1 \guardarrows &&&& y_2 \guardarrows
\\
&&&& x \ar[ull]\ar[urr]
}}}
\end{align*}
\caption{An example of sequential composition and some of the functions used to define $\fatsemi$.
}
\label{fig:seqcomp}
\end{figure}

In \Cref{fig:seqcomp} we give three \lpof s and, for each of them, we provide the stuck nodes, the extensible ones, the branches, and the result of sequentially composing them with the singleton \lpof\ $w$.
To lighten notation, we assume that $\alpha_1$, $\alpha_2$ and $\alpha_3$ have the same nodes and we only show the $\bot$ labels (associated to $z_1$ and $z_2$ in $\alpha_1$ and just to $z_1$ in $\alpha_2$).
In $\alpha_3$, there are no stuck nodes, since all maximal elements are non-$\bot$; hence, all nodes are extensible and the branches include all the possible Boolean combinations of $y_1$ and $y_2$ (this requires 4 copies of $w$ when building $\alpha_3 \fatsemi w$).
In $\alpha_2$ the only stuck node is $z_1$, arising when $y_1$ is true; all other nodes are extensible but the only branches to be considered are those in which $y_1$ is not true (thus, we only need two copies of $w$ here).
Finally, in $\alpha_1$ all nodes are stuck, since the leftmost branch will always lead to $\bot$ and so the whole \lpof\ cannot be sequentially composed with anything else; for this reason, $\alpha_1$ has no extensible node and no branch (thus, $\alpha_1 \fatsemi w = \alpha_1$).

\subsection{Unrolling Loops}
\label{app:loops}

In this section, we discuss the rationale for the copying behavior of sequential composition by demonstrating the semantic structures that are generated as the result of looping programs. We start with a simple program, which is a loop containing a single action as the body:
\[
  \whl ba
\]
As defined in \Cref{fig:trace-sem}, the semantics of this program is given by the following fixed point:
\[
  \de{\whl ba} = \mathsf{lfp}\left( \Phi_{\tuple{a, b}} \right)
  \qquad\text{where}\qquad
  \Phi_{\tuple{a, b}}(\Alpha) = \guard(b, \singleton{a} \fatsemi \Alpha, \de{\skp})
\]
By the Kleene fixed point theorem, the least fixed point above is given by the supremum over all the finite iterations of $\Phi_{\tuple{a, b}}$ applied to $\bot_\pom$. Unrolling this definition several times, we get the chains shown in \Cref{sec:pomset-semantics}.
Now suppose that we would like to sequentially compose another action after the loops, to obtain the program:
$
  (\whl ba) \;\; ; \;\; a'
$.
According to the definition of sequential composition from \Cref{sec:seq}, the semantics of this program is obtained by composing a copy of $a'$ after each branch of $\de{\whl ba}$. As we just saw, the branches of $\de{\whl ba}$ correspond to all of the terminating traces of the loop, and therefore we make countably many copies of $a'$, as shown in the first option below.
We investigated the possibility of defining sequential composition without such a copying, similar to what was done for finite structures in \cite{zilberstein2024probabilistic}. However, this approach does not translate to the infinite case. To see why not, consider the second option below structure, obtained by composing $a'$ after $\whl ba$, without copying.
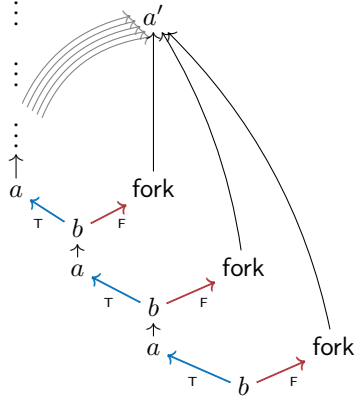
\[
  \de{(\whl ba) \;\; ; \;\; a'}
  \;\;=\;\;
  \left(
  \vcenter{\vbox{\xymatrix@R=3pt@C=-2pt{
  \vdots \\
  ~a~ \ar[u] && \fork
 \\
 & b \guardarrows
 \\
 &a \ar[u] && \fork
 \\
 && b \guardarrows
}}}\right)
  \,\fatsemi\ a'
  \;=  
  \vcenter{\vbox{\xymatrix@R=3pt@C=-2pt{
  \vdots && a' \\
  ~a~ \ar[u] && \fork\ar[u]
 \\
 & b \guardarrows && a'
 \\
 &a \ar[u] && \fork\ar[u]
 \\
 && b \guardarrows 
}}}
\;\;\; \textsf{OR} \;\;\;
  \vcenter{\vbox{\xymatrix@R=2pt@C=0pt{
  \vdots&& a'
  \\
  \vdots
    \ar@[gray]@/^1pc/[urr]
    \ar@<1mm>@[gray]@/^.95pc/[urr]
    \ar@<-1mm>@[gray]@/^1.05pc/[urr]
    \ar@<-2mm>@[gray]@/^1.1pc/[urr]
    \ar@<-3mm>@[gray]@/^1.15pc/[urr]
  \\
  a \ar[u] && \fork\ar[uu]
 \\
 & b \guardarrows
 \\
 &a \ar[u] && \fork \ar@/_/[uuuul]
 \\
 && b \guardarrows
}}}
\]
The node labelled as $a'$ is not finitely preceded; it has countably many predecessors, one for each terminating branch of the loop. Relaxing the finitely preceded requirement for \lpof s is not an option; it breaks many of the proofs from \Aref{sec:approx}, preventing us from approximating a pomset using finite structures and invalidating the extension lemma. Without the extension lemma, it is unclear how linearization could be defined, as well-founded recursion relies on the structure being finite.
Scott continuity of sequential composition comes into question as well. Scott continuity of $\fatsemi$ means that the following equation holds:
\[
  \left(\sup_{n\in \mathbb N} \Phi^n_{\tuple{a,b}} ( \bot_\pom ) \right) \fatsemi \singleton{a'}
  \qquad\stackrel{\?}=\qquad
  \sup_{n\in \mathbb N} \left( \Phi^n_{\tuple{a,b}} ( \bot_\pom ) \fatsemi \singleton{a'} \right)
\]
We investigated two possible expansions for the latter expression. The first possibility retained the requirement that nothing can proceed $\bot$; it is shown as the sequence below, ignoring the grey dotted arrows. The second possibility allowed $\bot$ nodes to have successors, and is shown below, where the grey dotted arrows indicate standard causality.
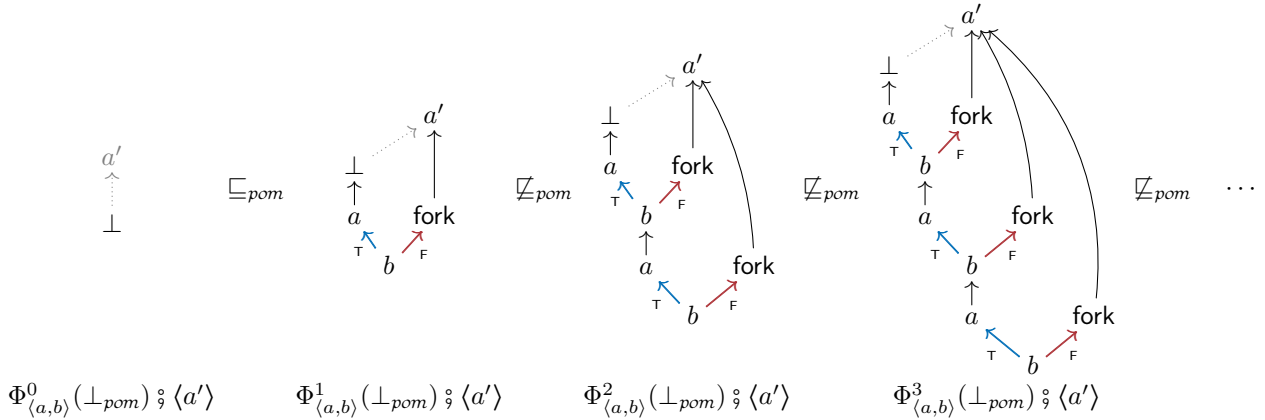
\[\arraycolsep=-1pt
\begin{array}{cccccccc}
\vcenter{\vbox{\xymatrix@R=12pt@C=-2pt{
{\color{gray} a'}
\\
\bot \ar@{.>}@[gray][u]
}}}
&\hspace{-1em}\sqsubseteq_\pom&
\vcenter{\vbox{\xymatrix@R=6pt@C=-2pt{
   && a' \\
  \bot \ar@{.>}@[gray][urr]
  \\
  a \ar[u] && \fork \ar[uu]
 \\
 & b \guardarrows
}}}
&\hspace{-.5em}\not\sqsubseteq_\pom&
\vcenter{\vbox{\xymatrix@R=6pt@C=-2pt{
   && a' \\
   \bot \ar@{.>}@[gray][urr]
   \\
  a \ar[u] && \fork \ar[uu]
 \\
 & b \guardarrows
 \\
 &a \ar[u] && \fork \ar@/_/[uuuul]
 \\
 && b \guardarrows
}}}
&\hspace{-.5em}\not\sqsubseteq_\pom&
\vcenter{\vbox{\xymatrix@R=6pt@C=-2pt{
   && a' \\
   \bot \ar@{.>}@[gray][urr]
   \\
  a \ar[u] && \fork \ar[uu]
 \\
 & b \guardarrows
 \\
 &a \ar[u] && \fork\ar@/_/[uuuul]
 \\
 && b \guardarrows
 \\
 &&a \ar[u] && \fork \ar@/_1.5pc/[uuuuuull]
 \\
 &&& b \guardarrows
}}}
& \not\sqsubseteq_\pom\cdots
\\
\Phi_{\tuple{a, b}}^0(\bot_\pom) \fatsemi \singleton{a'}
&&
\Phi_{\tuple{a, b}}^1(\bot_\pom) \fatsemi \singleton{a'}
&&
\Phi_{\tuple{a, b}}^2(\bot_\pom) \fatsemi \singleton{a'}
&&
\Phi_{\tuple{a, b}}^3(\bot_\pom) \fatsemi \singleton{a'}
\end{array}
\]
The causality from $\bot$ to $a'$ is problematic, since $\bot$ could be expanded into an infinite structure, again causing $a'$ to not be finitely preceded. But the lack of causality from $\bot$ to $a'$ is also problematic, as in the next approximation, $\bot$ will be replaced by the test $b$, which is supposed to be scheduled before $a'$. This complicates monotonicity of linearization, as the allowable interleavings can change as the structure grows.

In any case, the sequence above does not form a chain---at least, not according to our current definition of $\lepom$. The reason why this is not a chain is because the node labelled $a'$ gains more predecessors with each successive structure. We investigated several alternative definitions of $\lepom$, which did allow for this sort of expansion, however each one had significant problems---\eg some were not transitive, and thus not partial orders; some were not \dcpo s; and some caused $\Phi_{\tuple{C,b}}$ not to be monotone, meaning that loops could not be interpreted as fixed points using the Kleene fixed point theorem or the Knaster-Tarski theorem. In the course of our investigation, we found that the copying of structures in sequential composition was the best way to obtain a sound structure, obeying all the necessary laws.

Interestingly, this copying behavior corresponds to the following equational law of Probabilistic Guarded Kleene Algebra with Tests \cite{rozowski2023probabilistic}, where sequencing a command $C$ after a guarded choice is equivalent to including $C$ in both branches:
\[
  (\iftf b{C_1}{C_2}) \;\; ; \;\; C
  \quad\equiv\quad
  \iftf b{C_1;C}{C_2;C}
\]

\section{Pomsets with Tests}

\subsection{Labelled Partial Orders}
\label{app:lpof}

Let  $\nBot_\alpha$ denote the set of nodes not labelled with $\bot$, \ie $\nBot_\alpha \triangleq  \{x \in N_\alpha \mid \lambda_\alpha(x) \neq \bot\}$.

\begin{lemma}
\label{lem:nBot}
Let $\alpha \lelpo \beta$; then, $\nBot_\alpha \subseteq \nBot_\beta$.
\end{lemma}
\begin{proof}
By Def.\,\ref{def:lelpo} (\Cref{def:lelpo:N} and \ref{def:lelpo:lambda}), we have that $\lambda_\alpha(x) \leq \lambda_\beta(x)$, for all $x \in N_\alpha \subseteq N_\beta$.
Hence, if $x \in \nBot_\alpha$, then $\bot < \lambda_\alpha(x) \leq \lambda_\beta(x)$ and so $x \in \nBot_\beta$.
\end{proof}

\begin{restatable}{lemma}{lelpomissing}\label{lem:lelpo-missing}
If $\alpha \lelpo \beta$, then $N_\beta\setminus N_\alpha = \succplus{\Bot_\alpha}_\beta$.
\end{restatable}
\begin{proof}
We show the set inclusion in both ways. 

First, take any $x \in N_\beta \setminus N_\alpha$. Since $x$ is finitely proceeded, there must be a positive integer $n$ and nodes $x_1,\ldots,x_n \in N_\beta$ such that $x_1$ is the root of $\beta$, $x_n = x$, and $x_{i+1} \in \succ_\beta(x_i)$ for all $1 \le i < n$. Since $\alpha \lelpo\beta$, then $x_1$ must also be the root of $\alpha$. Now, let $i$ be smallest index such that $x_i \in N_\alpha$ and $x_{i+1} \not\in N_\alpha$ (note that $1 < i < n$ since $x_1 \in N_\alpha$ but $x = x_n \notin N_\alpha$). From $\alpha \lelpo \beta$, we know that $\succ_{\alpha}(x_i) = \succ_\beta(x_i) \setminus \succplus{\Bot_\alpha}_\beta$; so, $x_{i+1} \in \succplus{\Bot_\alpha}_\beta$ that allows us to conclude $x \in \succplus{\Bot_\alpha}_\beta$, being $x_{i+1} <_\beta x_n = x$.

Now, take $x \in \succplus{\Bot_\alpha}_\beta$, i.e. there exists $z \in \Bot_\alpha$ such that $z <_\beta x$. If $x \in N_\alpha$, being $<_\alpha\ =\ <_\beta \cap\ (N_\alpha \times N_\alpha)$, we would have that $z <_\alpha x$, in contradiction with $z \in \Bot_\alpha$.
\end{proof}

\begin{restatable}{lemma}{lpoposet}\label{lem:lpo-poset}
$\lelpo$ is a partial order.
\end{restatable}
\begin{proof}
Reflexivity is trivial. 
For antisymmetry, we have that $\alpha \sqsubseteq_\lpo \beta$ and $\beta \sqsubseteq_\lpo \alpha$ imply that 
$N_\alpha = N_\beta$; hence:
\[
  \mathord{<_\alpha}
  = \mathord{<_\beta} \cap (N_\alpha \times N_\alpha)
  = \mathord{<_\beta} \cap (N_\beta \times N_\beta)
  = \mathord{<_\beta}
\]
Then, let $x \in N_\alpha (= N_\beta)$. Trivially, $\varphi_\alpha(x) = \varphi_\beta(x)$; furthermore, $\lambda_\alpha(x) \leq \lambda_\beta(x)$ and $\lambda_\beta(x) \leq \lambda_\alpha(x)$, so, $\lambda_\alpha(x) = \lambda_\beta(x)$, by antisimmetry of $\leq$. 

For transitivity, let us assume that $\alpha \sqsubseteq_\lpo \beta$ and $\beta \sqsubseteq_\lpo \gamma$,  and prove that $\alpha \sqsubseteq_\lpo \gamma$. By hypothesis, $N_\alpha \dwclosed N_\beta\dwclosed N_\gamma$, so $N_\alpha \subseteq N_\gamma$; we have to prove that $N_\alpha$ is downwards closed (w.r.t. to $\gamma$).
Take any $x \in N_\alpha$ and $y <_\gamma x$. Since $N_\alpha \subseteq N_\beta$, $x \in N_\beta$ and, being $N_\beta \dwclosed N_\gamma$, then $y \in N_\beta$. Further, since $\mathord{<_\beta} = \mathord{<_\gamma} \cap (N_\beta \times N_\beta)$, then $y <_\beta x$. Therefore $y \in N_\alpha$, since $N_\alpha \dwclosed N_\beta$.

For the second item of Definition \ref{def:lelpo},
\[
\begin{array}{lcl}
<_\alpha &=& 
<_{\beta} \cap\ (N_\alpha \times N_{\alpha})
\\
\quad & = &
(<_\gamma \cap\ (N_\beta \times N_{\beta})) \cap (N_\alpha \times N_{\alpha})
\\
\quad & = &
<_\gamma \cap\ ((N_\beta  \cap N_\alpha ) \times (N_\beta \cap N_\alpha))
\\
&=&
<_\gamma \cap\ (N_\alpha  \times N_\alpha)
\end{array}
\]

Finally, let $x \in N_\alpha$. Trivially, $\varphi_\alpha(x) = \varphi_\beta(x) = \varphi_\gamma(x)$ and $\lambda_\alpha(x) \leq \lambda_\beta(x) \leq \lambda_\gamma(x)$; so, $\lambda_\alpha(x) \leq \lambda_\gamma(x)$, by transitivity of $\leq$.
Moreover,  we have:
\begin{align*}
  \succ_\alpha(x)
  &= \succ_\beta(x) \setminus \succplus{\Bot_\alpha}_\beta
  \\
  &= \left(\succ_\gamma(x) \setminus \succplus{\Bot_\beta}_\gamma\right) \setminus \succplus{\Bot_\alpha}_\beta
  \\
  &= \succ_\gamma(x) \setminus \left( \succplus{\Bot_\beta}_\gamma\cup \succplus{\Bot_\alpha}_\beta \right)
  \\
  &= \succ_\gamma(x) \setminus \succplus{\Bot_\alpha}_\gamma
\end{align*}
where the last equality holds thanks to \Cref{lem:lelpo-missing}: indeed, $\succplus{\Bot_\alpha}_\beta = N_\beta\setminus N_\alpha$, $\succplus{\Bot_\beta}_\gamma = N_\gamma\setminus N_\beta$ and
$\succplus{\Bot_\alpha}_\gamma = N_\gamma\setminus N_\alpha$; we conclude since $(N_\gamma \setminus N_\beta) \cup (N_\beta \setminus N_\alpha) = (N_\gamma \setminus N_\alpha)$.
\end{proof}

\begin{lemma}
\label{lem:succBot}
Let $N_\beta \dwclosed N_\alpha$ and
$<_\beta\ =\ <_\alpha \cap\ (N_\beta \times N_\beta)$. Then:
\begin{enumerate}
\item if $y \in \succ_\alpha(x)$ and $x <_\beta y$, then $y \in \succ_\beta(x)$;
\item if $y \in \succ_\beta(x)$, then $y \in \succ_\alpha(x)$;
\item $\succ_\beta(x) \subseteq \succ_\alpha(x) \setminus \succplus{\Bot_\beta}_\alpha$, for every $x \in N_\beta$.
\end{enumerate}
\end{lemma}
\begin{proof}
For the first claim, if it was $x <_\beta z <_\beta y$, then $x <_\alpha z <_\alpha y$ (being $<_\beta\ \subseteq\ <_\alpha$), in contradiction with $y \in \succ_\alpha(x)$.

For the second claim, $x <_\beta y$ and $\mathord{<_\beta} \subseteq \mathord{<_\alpha}$ imply that $x <_\alpha y$. Now, by contradiction, assume that $x <_\alpha z <_\alpha y$. Since $y \in N_\beta \dwclosed N_\alpha$, then $z \in N_\beta$; hence, being $<_\beta\ =\ <_\alpha \cap\ (N_\beta \times N_\beta)$, we would have that $x <_\beta z <_\beta y$, in contradiction with $y \in \succ_\beta(x)$. So, $y \in \succ_\alpha(x)$.

For the third claim, let $y \in \succ_\beta(x)$; by the second claim, $y \in \succ_\alpha(x)$.
If it was $y \in \succplus{\Bot_\beta}_\alpha$, then there would exist $z \in \Bot_\beta$ such that $z <_\alpha y$; being $<_\beta\ =\ <_\alpha \cap\ (N_\beta \times N_\beta)$, we would have that $z <_\beta y$, i.e. $z \not\in \Bot_\beta$. 
\end{proof}

\begin{lemma}\label{lem:predplus-directed}
If $D \subseteq \lpo(L)$ is a directed set, then $\predplus{x}_\beta = \predplus{x}_{\beta'}$ for any $\beta,\beta' \in D$ such that $x\in N_\beta \cap N_{\beta'}$.
\end{lemma}
\begin{proof}
Take any $\beta,\beta' \in D$ such that $x\in N_\beta$ and $x\in N_{\beta'}$.
Since $D$ is directed, there must be a $\gamma \in D$ such that $\beta \lelpo \gamma$ and $\beta' \lelpo \gamma$. 
Now, take any $y <_\beta x$, clearly this means that $y <_\gamma x$, and since $N_{\beta'} \dwclosed N_\gamma$, then $y \in N_{\beta'}$, and so $y \in \predplus{x}_{\beta'}$, so we have showed that $\predplus{x}_\beta \subseteq \predplus{x}_{\beta'}$. By an equivalent argument, we get that $\predplus{x}_{\beta'} \subseteq \predplus{x}_{\beta}$, therefore $\predplus{x}_\beta = \predplus{x}_{\beta'}$.
\end{proof}

\begin{lemma}\label{lem:directed-level}
For any directed set $D \subseteq \lpo(L)$ and $n \in \mathbb N$, there exists a finite set $X$ such that $\{ x \mid \lev_\beta(x) = n \} \subseteq X$ for all $\beta\in D$.
\end{lemma}
\begin{proof}
The proof is by induction on $n$. Let $n=1$, take any $\beta \in D$, and let $x$ be the root of $\beta$. Since $D$ is directed, then $x$ is the root of every $\beta'\in D$, therefore letting $X = \{x\}$, we have that $\{ y\mid \lev_\beta(y) = 1 \} = X$ for all $\beta\in D$.

Now, suppose that the claim holds for all $k < n$, and let $X_k$ be the finite set such that $\{ x \mid \lev_\beta(x) = k\} \subseteq X_k$ for each $k < n$. Also, let $X = \bigcup_{k<n} X_k$, so $\{ x \mid \lev_\beta(x) < n\} \subseteq X$, which is also finite. This means that only finitely many $\bot$ nodes can appear before level $n$. Now construct $A$ as follows: for each $x \in X$, select a $\beta\in D$ such that $\lambda_\beta(x) \neq \bot$, if such a $\beta$ exists. So $|A| \le |X|$, which is clearly finite. Since $A$ is a finite subset of $D$, which is directed, then there must be some $\gamma \in D$ that is an upper bound of all the elements in $A$. This means that for all $x\in X$, $x \notin \Bot_\beta$ for some $\beta \in D$ iff $x \notin \Bot_\gamma$.

Now let $Y = \{ x \mid \lev_\gamma(x) = n \}$. Take any $\beta \in D$, so there is a $\beta'\in D$ that is an upper bound of $\beta$ and $\gamma$. By construction, $\beta'$ can only extend $\gamma$ after level $n$, and so the elements of $\gamma$ at level $n$ must be the same as the elements of $\beta'$ and since $\beta\lelpo\beta'$, then $ \{ x \mid \lev_{\beta'}(x) = n \} \subseteq Y$.
\end{proof}

\lposup*
\begin{proof}
We need to show that: 
  \begin{enumerate}
  \item\label{lem:lpo-sup:lpo} $\alpha \in \lpo(L)$;
  \item\label{lem:lpo-sup:ub} $\beta \lelpo \alpha$, for all $\beta \in D$; and 
  \item \label{lem:lpo-sup:least} $\alpha\lelpo \gamma$, for all $\gamma$ s.t. $\beta \lelpo \gamma$ (for all $\beta\in D$). 
  \end{enumerate}

\medskip
  We start with property (\ref{lem:lpo-sup:lpo}). First, clearly $\alpha$ is single-rooted, since $D$ is a directed set and therefore every $\beta\in D$ must have the same root. We now show that every element of $N_\alpha$ is finitely proceeded.   
  Take any $x \in N_\alpha$, so $x\in N_\beta$ for some $\beta \in D$.
  By \Cref{lem:predplus-directed}, we know that $\predplus{x}_\beta = \predplus{x}_{\beta'}$ for all $\beta' \in D$, therefore $\predplus{x}_\alpha = \bigcup_{\beta'\in A} \predplus{x}_{\beta'} = \predplus{x}_\beta$, which much be finite.
  
  Next, we show that finitely many elements appear at level $n$. By \Cref{lem:directed-level}, we know that there is a finite set $X$ such that $\{ x \mid \lev_\beta(x) = n \} \subseteq X$ for all $\beta\in D$. Following from \Cref{lem:predplus-directed}, $\lev_\beta(x) = \lev_{\beta'}(x) = \lev_\alpha(x)$ for all $\beta,\beta'\in D$ and $x \in N_\alpha$, so the elements at level $n$ in $\alpha$ are the elements at level $n$ from all $\beta\in D$. This gives us:
  \[
    \{ x \mid \lev_\alpha(x) = n \}
    = \bigcup_{\beta\in D} \{ x \mid \lev_\beta(x) = n \}
    \subseteq X
  \]
  Now, take any $x \in N_\alpha$. If $\lambda_\alpha(x) = \bot$, then $\lambda_\beta(x) = \bot$ for all $\beta\in D$ such that $x\in N_\beta$. This means that $x$ has no successors in any $\beta\in D$, therefore it has no successors in $\alpha$.
  
  If $x <_\alpha y$, then $x <_\beta y$ for some $\beta \in D$. Since $D$ is directed, $\varphi_\beta(x) = \varphi_{\beta'}(x)$ for all $\beta' \in D$ that contain $x$ (and similarly for $y$); therefore, $\varphi_\alpha(y) = \varphi_\beta(y) \Rightarrow \varphi_\beta(x) = \varphi_\alpha(x)$. Finally, take any $x \in N_\alpha$, so $x\in N_\beta$ for some $\beta \in D$. By construction, $\varphi_\alpha(x) = \varphi_\beta(x) = \psi_x$ and, by \Cref{lem:predplus-directed}, $\predplus{x}_\alpha = \predplus{x}_\beta$; so, since $\free(\varphi_\beta(x)) \subseteq \predplus{x}_\beta$, then $\free(\varphi_\alpha(x)) \subseteq \predplus{x}_\alpha$.
  
  \medskip
  We next prove property (\ref{lem:lpo-sup:ub}): we choose any $\beta \in D$ and show all conditions of $\lelpo$. 
  
  First, we prove that $N_\beta \dwclosed N_\alpha$; the inclusion is by construction. Then, take $x\in N_\beta$ and $y<_\alpha x$; by construction, there must exist a $\beta'\in D$ such that $y <_{\beta'} x$. Take $\gamma \in D$ to be an upper bound of $\beta$ and $\beta'$, which must exist since $D$ is a directed set. So $\beta' \lelpo \gamma$ and therefore $y <_\gamma x$.  Since $N_\beta$ is a downward closed subset of $N_\gamma$ (being $\beta \lelpo \gamma$), it must be that $y\in N_\beta$.
  
  For the order, we proceed as follows: 
\begin{align*}             
\mathord{<_\alpha} \cap (N_\beta \times N_\beta) 
&=   \left(\bigcup_{\beta' \in D} \mathord{<_{\beta'}}\right) \cap \left( N_\beta \times N_{\beta} \right)\\
&=   \bigcup_{\beta' \in D}\left( \mathord{<_{\beta'}} \cap \left( N_\beta \times N_{\beta} \right)\right)\\
&\stackrel\ddagger= \mathord{<_\beta} 
             \end{align*}
For the step ($\ddagger$), we prove that, for every $\beta' \in D$ and every $(x,y) \in\ <_{\beta'}\ \cap\ (N_\beta \times N_\beta)$, we have that $(x, y) \in\ <_\beta$. Let $\gamma$ be the upper bound of $\beta$ and $\beta'$; then, $x <_\gamma y$ (being $x <_{\beta'} y$ and $\beta' \lelpo \gamma$). Since $<_\beta\ = \mathord{<_\gamma} \cap (N_\beta\times N_\beta)$ (because $\beta \lelpo \gamma$), we conclude $x<_\beta y$. 

Now, fix $x \in N_\beta$; since $D$ is directed, $\varphi_\alpha(x) = \psi_x = \varphi_\beta(x)$ and, by definition of sup, $\lambda_\beta(x) \leq \lambda_\alpha(x)$. 
Finally, by \Cref{lem:succBot}(3), $\succ_\beta(x) \subseteq \succ_\alpha(x) \setminus \succplus{\Bot_\beta}_\alpha$. We now show the reverse inclusion. Suppose that $y\in \succ_\alpha(x)$ and $y \notin \succplus{\Bot_\beta}_\alpha$; by construction, there must be some $\beta' \in D$ such that $x <_{\beta'} y$. By \Cref{lem:succBot}(1), $y \in \succ_{\beta'}(x)$.
Since $D$ is directed, there is a $\gamma \in D$ that is an upper bound of $\beta$ and $\beta'$.
Since $\beta' \lelpo \gamma$, then $\succ_{\beta'}(x) = \succ_\gamma(x) \setminus \succplus{\Bot_{\beta'}}_\gamma$; so, $y \in \succ_\gamma(x)$. Since $\beta \lelpo \gamma$, then we also have $\succ_{\beta}(x) = \succ_\gamma(x) \setminus \succplus{\Bot_{\beta}}_\gamma$. Since $y \notin \succplus{\Bot_{\beta}}_\alpha$, then $y$ also cannot be in $\succplus{\Bot_{\beta}}_\gamma$ since $\mathord{<_\gamma} \subseteq \mathord{<_\alpha}$ by construction; therefore, it must be that $y\in \succ_{\beta}(x)$.

 \medskip
Let us now move to property (\ref{lem:lpo-sup:least}): we fix $\gamma$ satisfying $\beta \lelpo \gamma$ (for all $\beta\in D$) and show that $\alpha\lelpo \gamma$. 

We know that $N_\beta\subseteq N_\gamma$ for all $\beta\in D$, therefore $N_\alpha = \bigcup_{\beta \in D} N_\beta \subseteq N_\gamma$. We now need to show that $N_\alpha$ is downwards closed. Take any $x \in N_\alpha$ and $y <_\gamma x$. Since $x \in N_\alpha$, then $x \in N_\beta$ for some $\beta \in D$. Since $N_\beta \dwclosed N_\gamma$, then $y \in N_\beta$, and therefore $y \in N_\alpha$.

Second,
\begin{align*}             
\mathord{<_\gamma} \cap (N_\alpha \times N_\alpha) 
&=  \mathord{<_\gamma} \cap \left( \bigcup_{\beta\in D} N_\beta \times  \bigcup_{\beta\in D} N_{\beta} \right)\\
&\stackrel \dagger = \mathord{<_{\gamma}} \cap \left(\bigcup_{\beta\in D}  (N_{\beta} \times N_{\beta})\right ) \\
&= \bigcup_{\beta \in D}  (\mathord{<_{\gamma}} \cap (N_{\beta} \times N_{\beta}))  \\
&= \bigcup_{\beta \in D} \mathord{<_{\beta}} \ =\ \mathord{<_\alpha}
             \end{align*}
where the step marked with ($\dagger$) follows using two set inclusions. The first is $\bigcup_{\beta\in D} (N_{\beta}\times N_{\beta}) \subseteq \bigcup_{\beta\in D} N_{\beta} \times  \bigcup_{\beta\in D} N_{\beta}$, which is immediate by the properties of cartesian product. For the reverse inclusion, take any $(x,y)\in \bigcup_{\beta\in D} N_{\beta} \times  \bigcup_{\beta\in D} N_{\beta}$, which means that there exist $\beta',\beta'' \in D$ such that $x \in N_{\beta'}$ and $y \in N_{\beta''}$.
Since $D$ is a directed set, there exists $\hat\beta \in D$ that is an upper bound of $\beta'$ and $\beta''$; thus, $N_{\beta'} \subseteq N_{\hat\beta}$ and $N_{\beta''} \subseteq N_{\hat\beta}$ and so
$N_{\beta'} \times N_{\beta''} \subseteq N_{\hat\beta} \times N_{\hat\beta}$.
The latter implies that $(x,y) \in  N_{\hat\beta} \times N_{\hat\beta} \subseteq \bigcup_{\beta\in D} (N_{\beta} \times N_{\beta})$, as desired.

To conclude, take any $x\in N_\alpha$; since $\beta \lelpo \gamma$ (for all $\beta$), we know that $\varphi_\beta(x) = \varphi_\gamma(x)$ and $\lambda_\beta(x) \leq \lambda_\gamma(x)$, for all $\beta$ that contain $x$.
This allows us to obtain that $\varphi_\alpha(x) = \varphi_\gamma(x)$ and $\lambda_\alpha(x) \leq \lambda_\gamma(x)$.

Finally, by \Cref{lem:succBot}(3) (with $\alpha$ in place of $\beta$ and $\gamma$ in place of $\alpha$ therein), $\succ_\alpha(x) \subseteq \succ_\gamma(x) \setminus \succplus{\Bot_\alpha}_\gamma$. We now show the reverse inclusion too.
Let $y\in \succ_\gamma(x)$ and $y\notin \succplus{\Bot_\alpha}_\gamma$. Since $y$ is finitely proceeded, we know that $\predplus{y}_\gamma$ is finite, therefore $\predplus{y}_\gamma \cap N_\alpha$ is finite and $\lambda_\alpha(z) \neq \bot$ for each $z \in \predplus{y}_\gamma \cap N_\alpha$ since $y \notin \succplus{\Bot_\alpha}_\gamma$. We now construct a set $A$ as follows: for each $z \in \predplus{y}_\gamma \cap N_\alpha$, select a $\beta_z\in D$ such that $z \in N_{\beta_z}$ and $\lambda_{\beta_z}(z) \neq \bot$, which must exist since $\lambda_\alpha(z) \neq \bot$. Since $A$ is a finite subset of $D$, which is directed, there is an element $\beta \in D$ which is an upper bound of all the elements in $A$; therefore, $\predplus{y}_\gamma \cap N_\alpha \subseteq N_\beta$ and $\lambda_\beta(z) \neq \bot$ for each $z \in \predplus{y}_\gamma \cap N_\alpha$.
Since $\beta \in D$, we also know that $\beta \lelpo \gamma$ and, since $x <_\gamma y$ and $x \in N_\alpha$, then $x \in N_\beta$; hence, $\succ_\beta(x) = \succ_\gamma(x) \setminus \succplus{\Bot_\beta}_\gamma$, so either $y \in \succ_\beta(x)$ or $y \in \succplus{\Bot_\beta}_\gamma$. In the first case, we are done since $\beta \lelpo \alpha$ and therefore $\succ_\beta(x) \subseteq \succ_\alpha(x)$ by \Cref{lem:succBot}(3). The second case is a contradiction, since it implies that there is some $z \in \Bot_\beta$ such that $z <_\gamma y$; hence, $z \in \predplus{y}_\gamma \cap N_\beta \subseteq \predplus{y}_\gamma \cap N_\alpha$, where the last inclusion holds since $\beta \lelpo \alpha$. This implies that $\lambda_\alpha(z) \neq \bot$, i.e.  $z \not\in \Bot_\beta$.
\end{proof}

\begin{corollary}
\label{lem:dcpoLPO}
  $\tuple{ \lpo(L), \lelpo }$ is a \dcpo.
\end{corollary}
\begin{proof}
Follows immediately from \Cref{lem:lpo-poset,lem:lpo-sup}.
\end{proof}

\subsection{Truncation of Lpos}
\label{app:lpo-trun}

First, we define some notation for lpos, similar to that of pomsets:
\[
  \lpo_\fin(L) \triangleq \{ \alpha \in \lpo(L) \mid |N_\alpha| < \infty \}
  \qquad\qquad
  \finapprox{\alpha} \triangleq \{ \beta \in \lpo_\fin(L) \mid \beta \lelpo \alpha \}
\]

\begin{definition}[Truncation]
Let $\alpha \in \lpo(L)$ and $n\in\mathbb N$. Then, $\trun \alpha n = \tuple{N,<,\lambda,\varphi}$ where
\begin{itemize}
\item $N \triangleq \{x \in N_\alpha \mid \lev_\alpha(x) \leq n\}$;
\item $<\ \triangleq\ <_\alpha \cap\ (N \times N)$;
\item for all $x \in N$, let $\lambda(x) \triangleq \lambda_\alpha(x)$, if $\lev_\alpha(x) < n$, and $\lambda(x) \triangleq \bot$, otherwise;
\item for all $x \in N$, let $\varphi(x) \triangleq \varphi_\alpha(x)$.
\end{itemize}
\end{definition}


Essentially, the truncation of $\alpha$ with $n$ consists in keeping $\alpha$ unchanged for all nodes $x\in N_\alpha$ such that $\lev_\alpha(x) < n$ and changing the labels of all nodes $x \in \lev^{-1}(n)$ to $\bot$. 
For any $\alpha \in \lpo(L)$, $\trun{\alpha}0$ is a singleton \lpof\ with $\bot$ at the root.

For a set of \lpof s $S \subseteq \lpo(L)$, we also let $\trun Sn = \{ \trun\alpha{n} \mid \alpha \in S \}$. Note that since a pomset is a set of \lpof s, this definition applies to pomsets too.


\begin{restatable}{lemma}{nodeInApprox}\label{lem:node-in-approx}
For any $\alpha \in \lpo(L)$ and $n > 0$, it holds that $\trun{\alpha}n \in \finapprox{\alpha}$.
\end{restatable}
\begin{proof}
To lighten notation, let $\beta = \trun \alpha n$.
First, note that $N_\beta$ is finite, since all nodes occur at level at most $n$, and there can only be finitely many nodes up to that level.

Now, we prove that $\beta \in \lpo(L)$.
First, $\beta$ is single rooted, is finitely proceeded and has finitely many nodes at every finite level, because it inherits these properties from $\alpha$. For the same reason, the conditions on $\varphi_{\beta}$ are satisfied.
Finally, for any $x\in N_\beta$ such that $\lambda_{\beta}(x) = \bot$, we consider two cases:
\begin{enumerate}
\item $\lev_\alpha(x) < n$: by construction, $\lambda_{\beta}(x) = \lambda_\alpha(x)$ and so $\succ_{\alpha}(x) = \emptyset$. By the definition of $<_{\beta}$, we have that $\succ_{\beta}(x) \subseteq \succ_\alpha(x)$ and so also $\succ_{\beta}(x) = \emptyset$.

\item $\lev_\alpha(x) = n$: obviously $x$ has no successors, since everything above level $n$ has been deleted.
\end{enumerate}

Finally, we show that $\beta \lelpo \alpha$.
First, we show that $N_\beta \dwclosed N_\alpha$. Clearly $N_\beta \subseteq N_\alpha$ by construction. Now, take any $y \in N_\beta$ and $z <_\alpha y$. Since $\lev_\alpha(y) \le n$, then $\lev_\alpha(z) < n$ and therefore $z \in N_\beta$.
By construction, we also have that $\mathord{<_\beta} = \mathord{<_\alpha} \cap (N_\beta\times N_\beta)$, $\lambda_\beta(y) \le \lambda_\alpha(y)$, and $\varphi_\beta(y) = \varphi_\alpha(y)$ for all $y \in N_\beta$. It remains only to show the condition on successors. Take any $y \in N_\beta$. By \Cref{lem:succBot}(3), $\succ_\beta(y) \subseteq \succ_\alpha(y) \setminus \succplus{\Bot_\beta}_\alpha$; we now show the reverse inclusion.
Suppose that $z\in \succ_\alpha(y)$ and $z\notin \succplus{\Bot_\beta}_\alpha$. Since $y\in N_\beta$, then $\lev_\alpha(z) \le n+1$. However, it cannot be that $\lev_\alpha(z) = n+1$, since that would imply that there exists $w \in \pred_\alpha(z)$ such that $\lev_\alpha(w) = n$, but then $\lambda_\beta(w) = \bot$, which is a contradiction. So, $\lev_\alpha(z) \le n$, which means that $z \in N_\beta$.
\end{proof}

An immediate consequence of \Cref{lem:node-in-approx} is that $\trun\alpha{n} \in \lpo(L)$ and $\trun\alpha{n} \lelpo \alpha$.

\begin{lemma}[Monotonicity of truncation]\label{lem:trunc-mono}
If $\alpha \lelpo \beta$, then $\trun \alpha n \lelpo \trun \beta n$, for every $n \in \mathbb N$.
\end{lemma}
\begin{proof}
By \Cref{lem:predplus-directed}, we have that $\lev_\alpha(x) = \lev_\beta(x)$, for every $x \in N_\alpha \subseteq N_\beta$; hence, 
$N_{\trun \alpha n} \subseteq N_{\trun \beta n}$.
For downward closure, let $x \in N_{\trun \alpha n}$ and $y <_{\trun \beta n} x$; thus, $ y \in \predplus{x}_{\trun \beta n}$ and again by \Cref{lem:predplus-directed} we easily conclude.

For the order relation, we have that
\begin{align*}
<_{\trun \beta n} \cap\ (N_{\trun \alpha n} \times N_{\trun \alpha n})
& = (<_{ \beta} \cap\ (N_{\trun \beta n} \times N_{\trun \beta n})) \cap\ (N_{\trun \alpha n} \times N_{\trun \alpha n})
\\
& =\ <_{ \beta} \cap\ ((N_{\trun \beta n} \times N_{\trun \beta n}) \cap (N_{\trun \alpha n} \times N_{\trun \alpha n}))
\\
& =\ <_{ \beta} \cap\ (N_{\trun \alpha n} \times N_{\trun \alpha n})
\\
& =\ <_{ \beta} \cap\ ((N_{\alpha} \times N_{ \alpha}) \cap (N_{\trun \alpha n} \times N_{\trun \alpha n}))
\\
& = (<_{ \beta} \cap\ (N_{\alpha} \times N_{ \alpha})) \cap (N_{\trun \alpha n} \times N_{\trun \alpha n})
\\
& =\ <_{ \alpha} \cap\ (N_{\trun \alpha n} \times N_{\trun \alpha n}) =\ <_{\trun \alpha n}
\end{align*}
where all steps follow by definition of truncation, associativity of $\cap$, $\alpha \lelpo \beta$, $N_{\trun \alpha n} \subseteq N_{\trun \beta n}$ and $N_{\trun \alpha n} \subseteq N_\alpha$.

Then, let $x \in N_{\trun \alpha n}$. By definition and $\alpha \lelpo \beta$, we have that $\lambda_{\trun \alpha n}(x) = \lambda_\alpha(x) \leq \lambda_\beta(x) = \lambda_{\trun \beta n}(x)$ and $\varphi_{\trun \alpha n}(x) = \varphi_\alpha(x) = \varphi_\beta(x) = \varphi_{\trun \beta n}(x)$. By \Cref{lem:succBot}(3), $\succ_{\trun \alpha n}(x) \subseteq \succ_{\trun \beta n}(x) \setminus \succplus{\Bot_{\trun \alpha n}}_{\trun \beta n}$;
let us prove the reverse inclusion. So, let $y \in \succ_{\trun \beta n}(x)$; this means that $\lev_\beta(x) < n$ and $\lev_\beta(y) \leq n$.
By assumption, $x \in N_\alpha$.
If also $y \in N_\alpha$, then $y \in \succ_{\trun \alpha n}(x)$, because  $\lev_\beta(y) = \lev_\alpha(y)$ and by \Cref{lem:succBot}(1).
If $y \notin N_\alpha$, then $y \in \succplus{\Bot_\alpha}_\beta$, because $\alpha\lelpo\beta$; this means that there exists $z \in \Bot_\alpha$ such that $z <_\beta y$. Hence, $\lev_\alpha(z) = \lev_\beta(z) <  \lev_\beta(y) = \lev_\alpha(y) \leq n$; thus, $z \in N_{\trun \alpha n} \subseteq N_{\trun \beta n}$ and so $z <_{\trun \beta n} y$. By definition of truncation, $\lambda_{\trun \alpha n}(z) = \lambda_\alpha(z) = \bot$; hence, we can conclude that $y \in \succplus{\Bot_{\trun \alpha n}}_{\trun \beta n}$.
\end{proof}

We extend the truncation to a set of lpos in the obvious way: $\trun D n \triangleq \{\trun \beta n \mid \beta \in D\}$.

\begin{lemma}[Scott continuity of truncation]\label{lem:trunc-continuous}
If $D$ is a directed set of \lpof s, then $\sup\trun D n$ exists and it is $\trun {\sup D} n$, for every $n \in \mathbb N$.
\end{lemma}
\begin{proof}
By \Cref{lem:lpo-sup}, $\sup D = \alpha$, where: 
  \begin{align*}
    N_\alpha &\triangleq \bigcup_{\beta\in D} N_{\beta}
    &
    \mathord{<}_\alpha &\triangleq \bigcup_{\beta\in D} \mathord{<_{\beta}}
    &
    \lambda_\alpha(x) &\triangleq \sup_{\beta\in D\,:\,x \in N_{\beta}} \lambda_{\beta}(x)
    &
    \varphi_\alpha(x) & \triangleq \psi_x
  \end{align*}
We now want to prove that $\gamma$, given by
  \begin{align*}
    N_\gamma &\triangleq \bigcup_{\beta \in D} N_{\trun \beta n}
    &
    \mathord{<}_\gamma &\triangleq \bigcup_{\beta\in D} \mathord{<_{\trun \beta n}}
    &
    \lambda_\gamma(x) &\triangleq \sup_{\beta\in D\,:\,x \in N_{\trun \beta n}} \lambda_{\trun \beta n}(x)
    &
    \varphi_\gamma(x) & \triangleq \psi_x
  \end{align*}
is the sup of $\trun D n$ and that it coincides with $\trun \alpha n$.

The fact that $\gamma = \sup\trun D n$ follows by \Cref{lem:lpo-sup}, once we prove that $\trun D n$ is directed; so, take $\beta_1,\beta_2 \in \trun D n$ and prove that they admit an upper bound in $\trun D n$. By definition, there exist $\beta_1',\beta_2' \in D$ such that $\beta_1 = \trun {\beta_1'} n$ and $\beta_2 = \trun {\beta_2'} n$. Since $D$ is directed, there exist a $\beta \in D$ such that $\beta_1',\beta_2' \lelpo \beta$. By \Cref{lem:trunc-mono}, $\beta_1,\beta_2 \lelpo \trun \beta n$ and, by definition, $\trun \beta n \in \trun D n$.

We are left with proving that $\gamma = \trun \alpha n$.
First, $x \in N_\gamma$ iff there exists $\beta \in D$ such that $x \in N_{\trun \beta n}$, i.e. $x \in N_\beta$ and $\lev_\beta(x) = n$; this happens iff $x \in \bigcup_{\beta \in D} N_\beta$ and $\lev_\alpha(x) = n$, being $\lev_\alpha(x) = \lev_\beta(x)$ (since $\beta \lelpo \alpha$), i.e. $x \in N_{\trun \alpha n}$.
Then,
\begin{align*}
<_{\trun \alpha n}
& =\ <_\alpha \cap\ (N_{\trun \alpha n} \times N_{\trun \alpha n})
\\
& = \left(\bigcup_{\beta \in D} <_\beta\right)\ \cap\ (N_{\trun \alpha n} \times N_{\trun \alpha n})
\\
& = \bigcup_{\beta \in D}  (<_\beta\ \cap\ (N_{\trun \alpha n} \times N_{\trun \alpha n}))
\\
& = \bigcup_{\beta \in D}  ((<_\alpha \cap\ (N_\beta \times N_\beta))\ \cap\ (N_{\trun \alpha n} \times N_{\trun \alpha n}))
\\
& = \bigcup_{\beta \in D}  (<_\alpha \cap\ ((N_\beta \times N_\beta)\ \cap\ (N_{\trun \alpha n} \times N_{\trun \alpha n}))
\\
& \stackrel \dagger = \bigcup_{\beta \in D}  (<_\alpha \cap\ (N_{\trun \beta n} \times N_{\trun \beta n}))
\\
& \stackrel \ddagger = \bigcup_{\beta \in D}  <_{\trun \beta n} =\ <_\gamma
\end{align*}
Step ($\dagger$) holds because $N_{\trun \alpha n} = N_\gamma = \bigcup_{\beta' \in D} N_{\trun {\beta'} n}$ and so 
$N_{\trun \alpha n} \cap N_\beta = (\bigcup_{\beta' \in D} N_{\trun {\beta'} n}) \cap N_\beta = \bigcup_{\beta' \in D} (N_{\trun {\beta'} n} \cap N_\beta) = N_{\trun {\beta} n}$, once we prove that $N_{\trun {\beta'} n} \cap N_\beta \subseteq N_{\trun {\beta} n}$, for every $\beta' \in D$ (indeed, obviously $N_{\trun {\beta} n} \subseteq \bigcup_{\beta' \in D} (N_{\trun {\beta'} n} \cap N_\beta)$). Fix $\beta' \neq \beta$ and $x \in N_{\trun {\beta'} n} \cap N_\beta$;
this means that $\lev{\beta'}(x) \leq n$, and so $\lev_\beta(x) = \lev_\alpha(x) = \lev{\beta'}(x) \leq n$, since $\beta,\beta' \lelpo \alpha$. This allows us to conclude that $x \in N_{\trun {\beta} n}$, as desired.
Step ($\ddagger$) is proved by double inclusion. 
Trivially, $x <_{\trun \beta n} y$ means that $x <_\beta y$ and $\lev_\beta(x),\lev_\beta(y) \leq n$; hence, $x,y \in N_{\trun \beta n}$ and, since $\beta \lelpo \alpha$, we also have $x <_\alpha y$, as desired.
Conversely, if $x <_\alpha y$ and $x,y \in N_{\trun \beta n}$, then $x,y \in N_{\beta}$ 
and, since $\beta \lelpo \alpha$, also $x <_\beta y$; thus, by definition, $x <_{\trun \beta n} y$.

Finally, fix $x \in N_\gamma = N_{\trun \alpha n}$. By definition, $\varphi_\gamma(x) = \varphi_{\trun \alpha n}(x)$.
For the labeling, first observe that $\lev_\gamma(x) = \lev_{\trun \beta n}(x) = \lev_\alpha(x)$, for every $\beta \in D$ such that $x \in N_\beta$.
If $\lev_\gamma(x) < n$, then by construction $\lambda_\gamma(x) = \sup_{\beta\in D\, :\, x \in N_{\trun \beta n}} \lambda_{\trun \beta n}(x) = \sup_{\beta\in D\, :\, x \in N_{ \beta}} \lambda_\beta(x) = \lambda_{\alpha}(x)= \lambda_{\trun \alpha n}(x)$.
Otherwise, $x \in N_{\trun \beta n}$ and  $\lambda_{\trun \beta n}(x) = \bot$, for all $\beta$ such that $x \in N_\beta$;
hence, $\lambda_\gamma(x) = \sup_{\beta\in D\, :\, x \in N_{\trun \beta n}} \lambda_{\trun \beta n}(x) = \bot = \lambda_{\trun \alpha n}(x)$.
\end{proof}

\section{Pomsets}

The following lemma provides alternative characterizations of the pomset order that will be convenient in some proofs and definitions later on.

\begin{restatable}{lemma}{altLePom}
\label{lem:altLePom}
The following three statements are equivalent:
\begin{enumerate}
\item $\Alpha \lepom \Beta$;
\item $\exists \alpha \in \Alpha,\beta\in \Beta.\ \alpha \sqsubseteq_\lpo \beta$; and
\item $\forall \beta\in\Beta.\ \exists\alpha\in\Alpha.\ \alpha \lelpo\beta$
\end{enumerate}
\end{restatable}
\begin{proof}
Clearly (1) implies (2) since $\Alpha$ is nonempty.

We now show that (2) implies (3), assume that there exist $\alpha \in \Alpha$ and $\beta\in\Beta$ such that $\alpha \sqsubseteq_\lpo \beta$. Now, take any $\beta'\in\Beta$, so $\beta \equiv\beta'$ and therefore there exists a bijection $f \colon N_\beta \to N_{\beta'}$ that satisfies the properties in \Cref{def:equivLPO}. Since $\alpha\lelpo\beta$, then $N_\alpha \subseteq N_\beta$, and so we can construct a bijection $g$ such that $g(x) = f(x)$ for all $x\in N_\alpha$. Let $\alpha' \equiv \alpha$ be the lpo generated via the bijection $g$, so clearly $\alpha' \lelpo \beta'$ since both are generated by the same bijection.

Now, we show that (3) implies (1), assume that forall $\beta\in \Beta$ there exists an $\alpha\in\Alpha$ such that $\alpha \lelpo\beta$.
Then, take any $\alpha' \in \Alpha$; by definition, there exists a bijection $f : N_\alpha \rightarrow N_{\alpha'}$ that satisfies the properties in Definition \ref{def:equivLPO}. Now, consider the LPO $\beta'$ obtained from $\beta$ by replacing every $x \in N_\alpha$ with $f(x)$ and every $x \in N_\beta \setminus N_\alpha$ with a node that does not occur in $N_{\alpha'}$. 
Clearly, this induces an extension of $f$ that is a bijection between $N_\beta$ and $N_{\beta'}$ that satisfies the properties in Definition \ref{def:equivLPO}; hence, $\beta' \in \Beta$ and, trivially, $\alpha' \sqsubseteq_\lpo \beta'$.
\end{proof}

\subsection{Truncation}

As mentioned in Appendix \ref{app:lpo-trun}, we extend the truncation operation to pomsets in the expected way: $\trun \Alpha n \triangleq \{\trun \alpha n \mid \alpha \in \Alpha\}$.

\begin{lemma}[Monotonicity of Truncation]\label{lem:trun-pom-mono}
If $\Alpha \lepom \Beta$, then $\forall n\in\mathbb N.\ \trun \Alpha n \lepom \trun \Beta n$.
\end{lemma}
\begin{proof}
By \Cref{lem:altLePom} we know that there exist $\alpha \in \Alpha$ and $\beta \in \Beta$ such that $\alpha \lelpo \beta$; by \Cref{lem:trunc-mono}, this entails that $\trun \alpha n \lelpo \trun \beta n$ and so, again by \Cref{lem:altLePom}, that $\trun \Alpha n \lepom \trun \Beta n$.
\end{proof}

\begin{lemma}\label{lem:pom-trun-limit}
If $\trun{\Alpha}n \lepom \Beta$ for all $n\in\mathbb N$, then $\Alpha \lepom\Beta$.
\end{lemma}
\begin{proof}({\it Adapted from \cite[Proposition 3.1.10]{bakker1990metric})}
Fix any $\alpha \in \Alpha$ and $\beta \in \Beta$. We will show that there exists $N \dwclosed N_\beta$ and a bijection $f \colon N_\alpha \to N$ such that $f(\alpha) \lelpo \beta$; therefore, by \Cref{lem:altLePom}, $\Alpha \lepom \Beta$.

We construct a tree of bijections as follows. Each node in the tree is a triple $(f, X, n)$ where $n\in \mathbb N$, $X \dwclosed N_\beta$, and $f \colon N_{\trun{\alpha}n} \to X$ is a bijection such that $f(\trun{\alpha}n) \lelpo \beta$. We add an edge from $(f, X, n)$ to $(g, Y, n+1)$ if $X \subseteq Y$ and $g(x) = f(x)$ for all $x\in N_{\trun{\alpha}n}$ (\ie $f$ is $g$ restricted to $N_{\trun{\alpha}n}$ or, equivalently, $g$ is an extension of $f$).
Let $x_\alpha$ and $x_\beta$ be the roots of $\alpha$ and $\beta$ respectively, and $f_0 \colon \{x_\alpha\} \to \{x_\beta\}$ be the obvious bijection. We will now show that the structure we described above defines an infinite tree with $(f_0, \{x_\beta\}, 0)$ as the root.

We first show that, for each $n\in \mathbb N$, there exists a possible node with $n$ as the third element; hence, there are infinitely many nodes (but we do not yet know that they can be all reached from the root). Fix any $n\in \mathbb N$. Since $\trun{\Alpha}n \lepom \Beta = [\beta]$, then there exists an $\alpha' \in \Alpha$ such that $\trun{\alpha'}n \lelpo \beta$. Since $\alpha \equiv \alpha'$, then there must be a bijection $g \colon N_\alpha \to N_{\alpha'}$. Now, let $X = \{ g(x) \mid x \in N_{\trun{\alpha}n} \}$ and $f \colon N_{\trun{\alpha}n} \to X$ be defined as $f(x) = g(x)$, which is clearly a bijection since $g$ is. Clearly
\[
  f(\trun{\alpha}n) = \trun{g(\alpha)}n = \trun{\alpha'}n \lelpo \beta
\]
and so $(f, X, n)$ is a node.

We now show that any node $(f, X, n)$ as described above can be reached from the root. The proof is by induction on $n$. If $n=0$, then the node must be $(f_0, \{x_\beta\}, 0)$, hence it is the root. If $(f, X, n+1)$ is a node, then $f \colon N_{\trun{\alpha}{n+1}} \to X$ is a bijection. Now, let $Y = \{ f(x) \mid x \in N_{\trun{\alpha}n} \}$. So, $g \colon N_{\trun{\alpha}n} \to Y$ defined as $g(x) = f(x)$ is also a bijection. We also know that $f(\trun{\alpha}{n+1}) \lelpo \beta$; since $g(\trun{\alpha}n) \lelpo f(\trun{\alpha}{n+1})$,  we have that $g(\trun{\alpha}n) \lelpo \beta$ by transitivity. Therefore, $(g, Y, n)$ is a valid node and there is an edge from $(g, Y, n)$ to $(f, X, n+1)$. By the induction hypothesis, $(g, Y, n)$ is reachable from the root, therefore so is $(f, X, n+1)$.

Now, let $k = |N_{\trun{\alpha}n}|$, which must be finite since only finitely many nodes can occur at each level. Therefore, there can only be finitely many elements in each level of the tree, since there are only only finitely many downward closed subsets of $N_\beta$ of size $k$; so, the tree must be finitely branching and, by K\"onig's Lemma, there exists an infinite path:
\[
  (f_n, X_n, n)_{n\in\mathbb N}
  \quad\text{where}\quad
  X_n \subseteq X_{n+1}
  \quad\text{and}\quad
  \forall x\in X_n.\ f_n(x) = f_{n+1}(x)
  \quad\text{for all}~ n\in \mathbb N
\]
Now let $X = \bigcup_{n\in\mathbb N} X_n$ and let $f\colon N_\alpha \to X$ be defined as $f(x) = f_{\lev_\alpha(x)}(x)$. We now show that $f(\alpha) \lelpo \beta$. First, clearly $N_{f(\alpha)} = X \dwclosed N_\beta$ since it is the union of downward closed sets. Now, for the order, we have:
\begin{align*}
  \mathord{<_{f(\alpha)}}
  &= \bigcup_{n\in\mathbb N} \mathord{<_{f_n( \trun{\alpha}n )}} 
  \\
  &= \bigcup_{n\in\mathbb N} \left( \mathord{<_{\beta}}  \cap (N_{f_n(\trun{\alpha}n)} \times N_{f_n(\trun{\alpha}n)})\right) 
  \\
  &= \mathord{<_{\beta}}  \cap \bigcup_{n\in\mathbb N} \left(N_{f_n(\trun{\alpha}n)} \times N_{f_n(\trun{\alpha}n)}\right)  
  \\
  &= \mathord{<_{\beta}}  \cap (N_{f(\alpha)} \times N_{f(\alpha)})
\end{align*}
Now, for each $x \in N_{f(\alpha)}$, let $n = \lev_\alpha(x)$. We have that:
\begin{align*}
  \lambda_{f(\alpha)}(x)
  &= \lambda_{\alpha}\left(  f^{-1}(x) \right)
  \\
  &= \lambda_\alpha\left( f^{-1}_{n}(x) \right)
  \\
  &= \lambda_\alpha\left( f^{-1}_{n+1}(x) \right)
  \\
  &= \lambda_{\trun{\alpha}{n+1}}\left( f^{-1}_{n+1}(x) \right)
  \\
  &= \lambda_{f_{n+1}(\trun{\alpha}{n+1})}(x)
  \\
  &\le \lambda_\beta(x)
\end{align*}
By an analogous argument, $\varphi_{f(\alpha)}(x) = \varphi_\beta(x)$. Now, by \Cref{lem:succBot}, we know that $\succ_{f(\alpha)}(x) \subseteq \succ_\beta(x) \setminus \succplus{\Bot_{f(\alpha)}}_\beta$, it just remains to show the reverse inclusion. Take any $y \in \succ_{\beta}(x)$ such that $y \notin \succplus{\Bot_{f(\alpha)}}_\beta$. Let $m = \lev_\beta(y)$ and note that $y \notin \Bot_{f_m(\trun{\alpha}m)}$ since $f(\alpha)$ and $f_m(\trun{\alpha}m)$ are identical up to level $m$. Since $f_{m}(\trun{\alpha}{m}) \lelpo \beta$, then $y \in \succ_{f_m(\trun{\alpha}m)}(x)$, and therefore $y \in \succ_{f(\alpha)}(x)$.
\end{proof}

\subsection{The Pomset Dcpo Structure}

\begin{lemma}
\label{lem:directedPom}
Let $D = \{\alpha_i \mid i \in I\} \subseteq \lpo(L)$ be a directed set and let $\DD = \{ [\alpha_i] \mid  i \in I \} \subseteq \pom(L)$; then, $\sup \DD$ exists and $\sup \DD = [\sup D]$.
\end{lemma}
\begin{proof}
By \Cref{lem:lpo-sup}, we know that $D$ admits a sup, let us call it $\alpha$ and let $\Alpha = [\alpha]$; we have to prove that $\Alpha = \sup\DD$, \ie
\begin{enumerate}
\item for all $i \in I$ it holds that $[\alpha_i] \lepom [\alpha]$, and
\item for every $\Beta$, if $[\alpha_i] \lepom \Beta$ for all $i \in I$, then $\Alpha \lepom \Beta$. 
\end{enumerate}
The first property is trivial, thanks to \Cref{lem:altLePom}(2).

For the second property, take some $\Beta$ such that $[\alpha_i] \lepom \Beta$ for all $i\in I$ and
let $A_n = \{ \beta\in\Beta \mid \trun{\alpha}n \lelpo \beta \}$. 

We first prove that $A_n \neq\emptyset$ for all $n\in\mathbb N$.
Let $N = \{ x\in N_\alpha \mid \lev_\alpha(x) \le n \}$, which by definition is a finite set. By the definition of supremum for lpos, for each $x\in N$ there exists an $i_x\in I$ such that $x \in N_{\alpha_{i_x}}$ and $\lambda_{\alpha_{i_x}}(x) = \lambda_\alpha(x)$. Since $\{ \alpha_{i_x} \mid x\in N \}$ is a finite subset of $D$, which is directed, there must be a $\alpha_j \in D$ such that $\alpha_{i_x}\lelpo\alpha_j$ for all $x\in N$. By construction, $\trun{\alpha_j}n = \trun{\alpha}n$. Since $\alpha_j\in D$, we also know that $[\alpha_j] \lepom \Beta$, so there must be a $\beta \in \Beta$ such that $\alpha_j \lelpo \beta$. This gives us $\trun{\alpha}n = \trun{\alpha_j}n \lelpo \alpha_j \lelpo \beta$, so $\beta\in A_n$ and therefore $A_n \neq \emptyset$.

Now, for each $n\in \mathbb N$, since $A_n \neq\emptyset$, we know that there exists a $\beta_n \in \Beta$ such $\trun{\alpha}n \lelpo \beta_n$; therefore, by \Cref{lem:altLePom}, we get that $\trun{\Alpha}n = [\trun{\alpha}n] \lepom [\beta_n] = \Beta$ so by \Cref{lem:pom-trun-limit}, we get that $\Alpha \lepom \Beta$.
\end{proof}


\begin{lemma}\label{lem:pom-lpo-chain}
For every transfinite chain of pomsets $(\Alpha_\delta)_{\delta < \zeta}$ (such that $\Alpha_\delta \lepom \Alpha_{\delta'}$ if $\delta < \delta'$), there exists a corresponding chain of \lpof s $(\alpha_\delta)_{\delta <\zeta}$ such that $\alpha_\delta \in \Alpha_\delta$, for all $\delta < \zeta$.
\end{lemma}
\begin{proof}
We construct the chain $(\alpha_\delta)_{\delta <\zeta}$ by transfinite induction as follows.
First, fix an arbitrary $\alpha_0 \in \Alpha_0$.
Second, for any ordinal $\delta$ such that $\alpha_\delta$ is fixed, we know that there exists an $\alpha_{\delta+1} \in \Alpha_{\delta+1}$ such that $\alpha_\delta \lelpo \alpha_{\delta+1}$; so, fix this element.
Finally, for any limit ordinal $\xi$, consider $\{ \alpha_\delta \mid \delta < \xi \}$ that is a chain (thus, also a directed set); we have to choose an $\alpha_\xi \in \Alpha_\xi$ that is bigger (w.r.t. $\lelpo$) than any $\alpha_\delta$. By \Cref{lem:directedPom}, we know that $\sup_{\delta < \xi} \Alpha_\delta = [ \sup_{\delta<\xi} \alpha_\delta ]$. By definition, $\Alpha_\xi$ is an upper bound of $\{ \Alpha_\delta \mid \delta < \xi \}$, therefore $\sup_{\delta < \xi} \Alpha_\delta \lepom \Alpha_\xi$; hence, there exists an $\alpha_\xi \in \Alpha_\xi$ such that $\sup_{\delta<\xi} \alpha_\delta \lelpo \alpha_\xi$, which also means that $\alpha_\delta \lelpo \alpha_\xi$ for all $\delta<\xi$.
\end{proof}

\begin{restatable}{lemma}{pomdcpo} 
\label{lem:pomdcpo}
$\tuple{\pom(L), \lepom}$ is a pointed \dcpo.
\end{restatable}
\begin{proof}
By \cite[Proposition 2.1.15]{abramsky1995domain}, 
$\tuple{\pom(L), \lepom}$ is a dcpo if any transfinite chain $(\Alpha_\delta)_{\delta<\zeta}$ has a supremum. By \Cref{lem:pom-lpo-chain}, there is a corresponding chain of lpos $(\alpha_\delta)_{\delta < \zeta}$ such that $\alpha_\delta \in \Alpha_\delta$ for each $\delta < \zeta$. Therefore, by \Cref{lem:directedPom}, we get that $\sup_{\delta < \zeta} \Alpha_\delta = [ \sup_{\delta<\zeta} \alpha_\delta]$, so the chain has a supremum.

Finally, pointedness follows from the fact that $\singleton{\bot} \lepom \Alpha$, for every $\Alpha$, where we let 
$\singleton{\bot}$ to denote the isomorphism class of the singleton LPOs whose node is labelled with $\bot$.
\end{proof}

\section{Approximation and Extension}
\label{sec:approx}

We start by recalling the notions of approximation and compactness taken from \cite{abramsky1995domain}.

\begin{definition}[Approximation and Compactness]
Given a \dcpo\ $\tuple{D, \le}$, the {\em approximation order} $\mathord{\ll}\subseteq D \times D$ is defined as follows: $d_1 \ll d_2$ iff, for any directed set $D' \subseteq D$ such that $d_2 \le \sup D'$, then $d_1 \le d$, for some $ d \in D'$. 

An element $d \in D$ is \emph{compact} if it approximates itself ($d \ll d$).
We denote the set of compact elements of $D$ as $\mathsf K(D) \triangleq \{ d \in D \mid d \ll d \}$.
\end{definition}

The approximation order is sometimes called the \emph{way-below} relation or the \emph{order of definite refinement} \cite{smyth1986finite}. Intuitively, $d_1 \ll d_2$ means that $d_1$ is a ``simpler representation'' of $d_2$. 
We now want to show that such an order coincides with the notion of finite approximations (see \Cref{lem:fin-approx}); this requires some preliminary results for \lpof s.

\begin{lemma}\label{lem:fin-approx-directed-lpo}
For any $\alpha \in \lpo(L)$, the set $\finapprox\alpha$ is directed.
\end{lemma}
\begin{proof}
Take any $\beta,\gamma \in \finapprox\alpha$, so $\beta,\gamma\in \lpo_\fin(L)$ and $\beta,\gamma \lelpo \alpha$.
Let $n$ be the maximum level of any node in $\beta$ or $\gamma$, and $\alpha' = \trun{\alpha}{n+1}$. Clearly, $\alpha' \in \finapprox{\alpha}$ by \Cref{lem:node-in-approx}.

We conclude the proof by showing that $\alpha'$ is an upper bound for $\beta$ and $\gamma$; we will only show $\beta \lelpo \alpha'$, as the case for $\gamma$ is identical. 
Clearly $N_\beta \subseteq N_{\alpha'}$ since $\beta$ contains a subset of the nodes from $\alpha$ at level at most $n$ and $\alpha'$ contains all nodes from $\alpha$ up to level $n+1$. We will now show that it is downward closed. Take any $x \in N_\beta$ and $y <_{\alpha'} x$; this means that $y <_\alpha x$. Being $N_\beta \dwclosed N_\alpha$, we have that $y \in N_\beta$.

Next, since $\alpha' \lelpo\alpha$ and $N_\beta \subseteq N_{\alpha'}$, we have that:
\begin{align*}
  \mathord{<_{\alpha'}} \cap (N_\beta \times N_\beta)
  &= (\mathord{<_{\alpha}} \cap (N_{\alpha'} \times N_{\alpha'})) \cap (N_\beta \times N_\beta)
  \\
  &= \mathord{<_{\alpha}} \cap ((N_{\alpha'} \cap N_{\beta}) \times (N_{\alpha'} \cap N_\beta))
\\
  &= \mathord{<_{\alpha}} \cap (N_\beta \times N_\beta)
\\
  &= \mathord{<_\beta}
\end{align*}
Finally, let $x \in N_\beta$. By construction, $\varphi_\beta(x) = \varphi_\alpha(x) = \varphi_{\alpha'}(x)$ and $\lambda_\beta(x) \le \lambda_\alpha(x) = \lambda_{\alpha'}(x)$, since $\lev_\alpha(x) \le n$.
For the last requirement, by \Cref{lem:succBot}(3) we have that $\succ_\beta(x) \subseteq \succ_{\alpha'}(x) \setminus \succplus{\Bot_\beta}_{\alpha'}$; we need to prove the reverse inclusion. 
Let $y \in \succ_{\alpha'}(x)$ and $y \not\in \succplus{\Bot_\beta}_{\alpha'}$.
 We have two cases:
\begin{itemize}
\item $y \in N_\beta$: in this case, $y \in \succ_{\beta}(x)$ by \Cref{lem:succBot}(1).
\item $y \not\in N_\beta$: Because $y$ is finitely proceeded in $\alpha'$, there exist $y_1,\ldots,y_n$ such that $y_1 = \min_{\alpha'}$, $y_n = y$ and $y_{i+1} \in \succ_{\alpha'}(y_i)$. Because of \Cref{lem:succBot}(2), $y_{i+1} \in \succ_{\alpha}(y_i)$, for all $i$.
Being $y_1=\min_\beta$, there exists an $i$ ($1 <i < n$) such that $y_i \in N_\beta$ and $y_{i+1} \not\in N_\beta$. Since $\beta \lelpo \alpha$, we have that $\succ_{\beta}(y_i) \subseteq \succ_\alpha(y_i) \setminus \succplus{\Bot_{\beta}}_\alpha$; so, $y_{i+1} \in \succplus{\Bot_{\beta}}_\alpha$ and, consequently, $y \in \succplus{\Bot_{\beta}}_\alpha$. This latter fact, together with $y \in N_{\alpha'}$, implies that $y \in \succplus{\Bot_{\beta}}_{\alpha'}$: contradiction.
\qedhere
\end{itemize}
\end{proof}

\begin{lemma}
\label{lem:fin-approx-directed}
For any $\Alpha \in \pom(L)$, the set $\finapprox\Alpha$ is directed.
\end{lemma}
\begin{proof}
Take any $\Beta,\Beta' \in \finapprox\Alpha$, so $\Beta,\Beta' \in \pom_\fin(L)$ and $\Beta \lepom\Alpha$ and $\Beta'\lepom\Alpha$. Now take any $\alpha \in \Alpha$, by \Cref{lem:altLePom} we know that there is a $\beta\in\Beta$ and $\beta'\in \Beta'$ such that $\beta \lelpo\alpha$ and $\beta' \lelpo\alpha$. Therefore, $\beta,\beta' \in \finapprox{\alpha}$, so by \Cref{lem:fin-approx-directed-lpo}, there exists a $\gamma \in \finapprox{\alpha}$ such that $\beta \lelpo \gamma$ and $\beta' \lelpo\gamma$. Since $\gamma \in \finapprox\alpha$, then $[\gamma] \in \finapprox{\Alpha}$, and clearly $\Beta \lepom [\gamma]$ and $\Beta' \lepom[\gamma]$.
\end{proof}


\begin{lemma}\label{lem:fin-approx-lpo}
For any $\alpha \in \lpo(L)$, $\alpha = \sup \finapprox\alpha$.
\end{lemma}
\begin{proof}
By \Cref{lem:fin-approx-directed-lpo}, $\finapprox\alpha$ is a directed set, therefore its supremum is given by the construction in \Cref{lem:lpo-sup}. We now show that this gives us exactly $\alpha$. First, we show this for the nodes. We know that $N_{\sup \finapprox\alpha}  = \bigcup_{\beta \in \finapprox\alpha} N_\beta$.
Clearly, $\bigcup_{\beta \in \finapprox\alpha} N_\beta \subseteq N_\alpha$ since each $N_\beta \subseteq N_\alpha$ by definition. 
For the reverse inclusion, take any $x \in N_\alpha$ and let $n = \lev_\alpha(x)$. Obviously, $x \in N_{\trun\alpha{n}}$, and by \Cref{lem:node-in-approx}, we get that $\trun{\alpha}n \in \finapprox\alpha$. Therefore, $x\in \bigcup_{\beta \in \finapprox\alpha} N_\beta$, so $N_\alpha \subseteq \bigcup_{\beta \in \finapprox\alpha} N_\beta$ and thus $\bigcup_{\beta \in \finapprox\alpha} N_\beta = N_\alpha$

By definition, $\beta \lelpo \alpha$, for each $\beta\in\finapprox\alpha$; so clearly $\alpha$ is an upper bound for $\finapprox\alpha$ and therefore $\sup \finapprox\alpha \lelpo \alpha$ (since the supremum is the least upper bound). Hence, ${<_{\sup \finapprox\alpha}} \subseteq {<_\alpha}$; we show the reverse inclusion as follows. Take any $x <_\alpha y$ and let $\gamma = \trun{\alpha}{\lev_\alpha(y)}$. Clearly $y \in N_\gamma$ and by \Cref{lem:node-in-approx}, $\gamma \in \finapprox{\alpha}$. Since $N_\gamma \dwclosed N_\alpha$, then $x \in N_\gamma$ too, which means that $x <_\gamma y$. This gives us:
\begin{align*}
  (x,y)
  \quad \in \quad {<_\gamma}
  \quad \subseteq\quad \bigcup_{\beta\in \finapprox\alpha} {<_\beta}
  \quad =\quad   {<_{\sup \finapprox\alpha}}
\end{align*}
Finally, for any $x \in N_\alpha$, we have:
\begin{itemize}
\item $\varphi_{\sup \finapprox\alpha}(x) = \varphi_\alpha(x)$: Take any $\beta \in \finapprox\alpha$, so $\beta \lelpo \sup \finapprox\alpha$ and $\beta \lelpo \alpha$, therefore $\varphi_{\finapprox\alpha} = \varphi_\beta = \varphi_\alpha$.
\item
\begin{align*}
  \lambda_{\sup \finapprox\alpha}(x)
  &= \sup_{\beta\in\finapprox\alpha : x\in N_\beta} \lambda_\beta(x)
  \\
  &= \sup \{
    \lambda_\beta(x)
    \mid
    \beta \in \lpo_\fin(L), \beta\lelpo \alpha, x \in N_\beta
  \}
\\
  &\stackrel 1 = \sup \{
    \ell
    \mid
    \ell \le \lambda_\alpha(x)
  \}
\\
  &\stackrel 2 = \lambda_\alpha(x)
\end{align*}
where (1) holds by the definition of $\lelpo$, that implies that the set above must include all the labels that are less than $\lambda_\alpha(x)$, and (2) holds since the set contains $\lambda_\alpha(x)$ itself.
\qedhere
\end{itemize}
\end{proof}

%

\finapproxx*
\begin{proof}
Fix any $\alpha \in \Alpha$ and consider the set $\finapprox\alpha$: by \Cref{lem:fin-approx-directed-lpo} it is directed and by \Cref{lem:fin-approx-lpo} its sup is $\alpha$. Then, by \Cref{lem:directedPom}, $\sup\finapprox\Alpha$ exists and it is $[\alpha]$, i.e. $\Alpha$.
\end{proof}

\begin{lemma}\label{lem:lelpo-finite}
If $\alpha \in \lpo_\fin(L)$, then there are only finitely many $\beta$ such that $\beta \lelpo \alpha$.
\end{lemma}
\begin{proof}
Any $\beta\lelpo \alpha$ is such that $N_\beta \subseteq N_\alpha$; in other words, $N_\beta \in \mathcal{P}(N_\alpha)$. If $|N_\alpha|=n$. then
 $|\mathcal{P}(N_\alpha)| = 2^n$ (actually, there are fewer possibilities, since $N_\beta$ is downward closed, but $2^n$ is a sufficient upper bound).
The order and formula of $\beta$ are deterministically determined by $N_\beta$, since $\mathord{<_\beta} = \mathord{<_{\alpha}} \cap (N_\beta \times N_\beta)$ and $\varphi_\beta(x) = \varphi_{\alpha}(x)$.
Let $k = \max_{x\in N_\alpha} |\predplus{\lambda_\alpha(x)}_L|$, which must be a finite integer since it is the maximum of a finite set of integers. So, there are at most $2^n$ subsets of $N_\alpha$ and for each of those subsets, there are at most $k$ choices for the labels of each node. Hence, there are at most $2^{n}\cdot n \cdot k$ possible $\beta \lelpo \alpha$.
\end{proof}

%

\begin{lemma}\label{lem:trun-directed-finite}
If $\DD \subseteq \pom(L)$ is a directed set, then $\trun\DD{n}$ is directed and finite, for any $n\in\mathbb N$.
\end{lemma}
\begin{proof}
We first show that $\trun\DD{n}$ is directed. Take any $\Beta_1,\Beta_2 \in \trun\DD{n}$. This means that there is $\GGamma_1,\GGamma_2 \in \DD$ such that $\Beta_i = \trun{\GGamma_i}n$ for each $i\in\{1,2\}$, and since $\DD$ is directed, then there is a $\GGamma\in\DD$ such that each $\GGamma_i \lepom \GGamma$. By monotonicity of truncation, we have that $\Beta_i = \trun{\GGamma_i}n \lepom \trun{\GGamma}n$, and clearly $\trun{\GGamma}n \in \trun\DD{n}$, therefore $\trun\DD{n}$ is directed.

Because of  \Cref{lem:trun-pom-mono}, we know that $\Beta \lepom \trun\Alpha{n}$ for all $\Beta \in \trun\DD{n}$.
Fix any $\alpha \in \trun{\Alpha}n$, so clearly $\alpha$ is finite. By \Cref{lem:altLePom}, we know that, for every $\Beta \in \trun\DD{n}$, there is a $\beta \in \Beta$ such that $\beta \lelpo \alpha$. By \Cref{lem:lelpo-finite}, there can only be finitely many such elements. Thus we have that $\{ \beta \mid \beta \lelpo \alpha \}$ is finite, and also clearly $\trun\DD{n} \subseteq \{ [\beta] \mid \beta \lelpo \alpha \}$; therefore, $\trun\DD{n}$ is also finite.
\end{proof}

\begin{lemma}\label{lem:trun-pom-cont}
If $\DD$ is a directed set of pomsets, then $\sup \trun \DD n = \trun {\sup \DD} n$ for all $n\in\mathbb N$.
\end{lemma}
\begin{proof}
We show that $\trun{-}n$ is chain continuous; indeed, by \cite[Corollary 3]{markowsky1976chain-complete}, any chain continuous function on a \dcpo\ is also Scott continuous, and thus $\sup \trun{\DD}n = \trun{\sup\DD}n$.

Take any pomset chain $(\Alpha_\delta)_{\delta < \zeta}$; by \Cref{lem:pom-lpo-chain}, there is a corresponding chain of \lpof s $(\alpha_\delta)_{\delta < \zeta}$ such that $\alpha_\delta \in \Alpha_\delta$ for all $\delta < \zeta$. 
Since truncation is monotonic (\Cref{lem:trun-pom-mono}), then $(\trun{\Alpha_\delta}n)_{\delta < \zeta}$ is also a chain and, by \Cref{lem:trunc-continuous}, we know that $\sup_{\delta <\zeta} \trun{\alpha_\delta}n = \trun{\sup_{\delta<\zeta} \alpha_\delta}n$. Clearly $\trun{\alpha_\delta}n \in \trun{\Alpha_\delta}n$ for all $\delta<\zeta$, therefore by \Cref{lem:directedPom}, we get both $\sup_{\delta<\zeta}\Alpha_\delta = [\sup_{\delta <\zeta} \alpha_\delta]$ and $\sup_{\delta<\zeta}\trun{\Alpha_\delta}n = [\sup_{\delta <\zeta} \trun{\alpha_\delta}n]$. Putting all these facts together, we get:
\[\textstyle
  \sup_{\delta<\zeta} \trun{\Alpha_\delta}n
  = [\sup_{\delta <\zeta} \trun{\alpha_\delta}n]
  = [\trun{\sup_{\delta<\zeta} \alpha_\delta}n]
  = \trun{ [\sup_{\delta<\zeta} \alpha_\delta] }n
  = \trun{ \sup_{\delta<\zeta}\Alpha_\delta }n
  \qedhere
\]
\end{proof}

\begin{lemma}[Compactness]
\label{lem:compactness}
$\mathsf{K}(\pom(L)) = \pom_\fin(L)$.
\end{lemma}
\begin{proof}
We first show that any compact pomset is finite, \ie $\mathsf{K}(\pom(L)) \subseteq \pom_\fin(L)$. 
Take any $\Alpha \in \mathsf{K}(\pom(L))$, so for any directed set $A$, if $\Alpha \lelpo\sup A$ then $\Alpha \lelpo\GGamma$ for some $\GGamma \in A$. Now, let $A = \finapprox\Alpha$, which is directed by \Cref{lem:fin-approx-directed}. By \Cref{lem:fin-approx}, we also know that $\Alpha = \sup A$, so $\Alpha \lepom \sup A$. Since $\Alpha$ is compact, there must be some $\GGamma \in A$ such that $\Alpha \lepom\GGamma$. Since $\GGamma \in A$, then $\GGamma \in \pom_\fin(L)$; therefore, since $\Alpha\lelpo \GGamma$, then $\Alpha \in \pom_\fin(L)$ too.

We now show that any finite pomset is compact, \ie $\pom_\fin(L) \subseteq \mathsf{K}(\pom(L))$. Taking any $\Alpha \in \pom_\fin(L)$, we must show that $\Alpha \ll\Alpha$. That is, for any directed set $A$, if $\Alpha \lepom \sup A$, then there exists $\GGamma \in A$ such that $\Alpha \lepom \GGamma$.

To this aim, let $A$ be a directed set such that $\Alpha \lepom \sup A$. Since $\Alpha$ is finite, let $\alpha\in\Alpha$ be an arbitrary representative lpo, then $n = \max_{x\in N_\alpha} \lev(x) + 1$ must be finite. Moreover, $\trun\Alpha{n} = \Alpha$, since $n$ is one level above the maximum level of $\Alpha$ and therefore the truncation only converts labels to $\bot$ above the maximum level of $\Alpha$ and removes nodes even above that. Therefore, by montonicity (\Cref{lem:trun-pom-mono}) and Scott continuity (\Cref{lem:trun-pom-cont}) of truncation, we have:
\[
  \Alpha = \trun\Alpha{n} \lepom \trun{\sup A}n = \sup \{ \trun\Beta{n} \mid {\Beta\in A} \}
\]
By \Cref{lem:trun-directed-finite}, we know that $\{ \trun\Beta{n} \mid {\Beta\in A} \}$ is a finite directed set, and therefore it must contain its own supremum. That is, there exists $\GGamma \in \{ \trun\Beta{n} \mid {\Beta\in A} \}$ such that $\GGamma = \sup \{ \trun\Beta{n} \mid {\Beta\in A} \}$. This means that there exists a $\GGamma' \in A$ such that $\GGamma = \trun{\GGamma'}n$, and so clearly $\GGamma \lepom \GGamma'$, and therefore by transitivity, $\Alpha \lepom\GGamma'$.
\end{proof}

\begin{corollary}[Approximation]
\label{cor:lefin-approx}
$\Alpha \ll \Beta$ iff $\Alpha \in \finapprox\Beta$.
\end{corollary}
\begin{proof}
Suppose that $\Alpha \ll \Beta$. Since $\finapprox\Beta$ is a directed set and $\sup \finapprox\Beta = \Beta$ (\Cref{lem:fin-approx-directed,lem:fin-approx}), then there must be some $\GGamma \in \finapprox\Beta$ such that $\Alpha \lelpo\GGamma$. Since $\finapprox\Beta$ is downward closed by definition, then $\Alpha \in \finapprox\Beta$ too.

Now suppose that $\Alpha \in \finapprox\Beta$. By \Cref{lem:compactness}, $\Alpha$ is compact, so $\Alpha\ll\Alpha$. By the definition of $\finapprox\Beta$, $\Alpha \lelpo\Beta$ too. Therefore, by \cite[Proposition 2.2.2(2)]{abramsky1995domain}, we get that $\Alpha \ll \Beta$.
\end{proof}

For any \dcpo s $\tuple{X, \le_X}$ and $\tuple{Y, \le_Y}$, the product poset $\tuple{X\times Y, \le}$ where $(x,y) \le (x', y')$ iff $x\le_X x'$ and $y\le_Y y'$ is also a \dcpo \cite[Proposition 3.2.2]{abramsky1995domain}. We will use this fact in the next lemma. We write $\vec x \in X_1 \times \cdots \times X_n$ to denote elements of Cartesian products. Let $\vec\Alpha = (\Alpha_1, \ldots, \Alpha_n)$ and $\vec\Beta = ( \Beta_1, \ldots, \Beta_n)$. We write $\vec\Alpha \ll \vec\Beta$ to mean that $\Alpha_i \ll \Beta_i$ for all $1 \le i \le n$.

\begin{lemma}
\label{lem:fin-approx2}
For any directed set $D \subseteq \pom(X)^n$:
\[
\left\{ \vec\Beta\mid \exists\vec\Alpha \in D. \vec\Beta \ll \vec\Alpha \right\}
  =
\left\{ \vec\Beta\mid \vec\Beta \ll \sup D \right\}
\]
\end{lemma}
\begin{proof}

We show the set inclusion in both directions. First, take any element $\vec\Beta = ( \Beta_1, \ldots, \Beta_n)$ from the first set, so we know that $\Beta_i \ll \Alpha_i$ for all $1 \le i \le n$ and some $\vec\Alpha = (\Alpha_1, \ldots, \Alpha_n)  \in D$. Clearly $\vec\Alpha \lepom \sup D$, so $\vec\Beta \ll \sup D$ \cite[Proposition 2.2.2(2)]{abramsky1995domain} and therefore the forward inclusion holds.

Now take any element $\vec\Beta = ( \Beta_1, \ldots, \Beta_n)$ from the second set and let $D_1, \ldots, D_n$ be the projections of $D$.
Since the supremum on the product order is coordinatewise \cite[Proposition 2.2.2(2)]{abramsky1995domain}, then $\sup D = (\sup D_1, \ldots, \sup D_n)$.
So, $\Beta_i \ll \sup D_i$ for all $1 \le i\le n$. Therefore, by the definition of $\ll$, there must be some $\Alpha_i \in D_i$ such that $\Beta_i \lepom \Alpha_i$. Now, for each $i$, let $\vec{\Alpha}_i \in D$ be the tuple whose $i-$th projection is $\Alpha_i$.
Since $D$ is directed, then there is some $\vec\Alpha \in D$, which is an upper bound of all the $\vec{\Alpha}_i$s.
From $\Beta_i \ll\sup D_i$ and \Cref{cor:lefin-approx}, we know that $\Beta_i$ is finite, therefore $\Beta_i \in \finapprox{\Alpha_i}$. By \Cref{cor:lefin-approx} again, $\Beta_i \ll \Alpha_i$, and so $\vec\Beta \ll \Alpha$ and the reverse inclusion holds too.
\end{proof}

\extension*
\begin{proof}
By \Cref{cor:lefin-approx}, $\{ (\Alpha_1',\ldots,\Alpha_n')  \mid \forall i.\,\Alpha_i' \ll \Alpha_i \} = \finapprox{\Alpha_1} \times \cdots \times \finapprox{\Alpha_n}$, that is a directed set by \Cref{lem:fin-approx-directed}. Since $f$ is monotone, then $\{ f(\Alpha_1',\ldots,\Alpha_n') \mid \forall i.\,\Alpha_i' \ll \Alpha_i \}$ is also a directed set, and so clearly the supremum exists. 

We now establish that $f^\ast$ is Scott continuous. Let $D \subseteq \pom(L)^n$ be a directed set and:
\begin{align*}
\sup_{\vec\Alpha \in D} f^\ast(\vec\Alpha)
&= \sup_{\vec\Alpha\in D}\sup_{\vec\Beta \ll \vec\Alpha} f(\vec\Beta)
\\
&= \sup \left\{ f(\vec\Beta) \mid \exists \vec\Alpha \in D.\, \vec\Beta \ll \vec\Alpha \right\}
&&\text{By \cite[Proposition 2.1.4(3)]{abramsky1995domain}.}
\\
&= \sup \left\{ f(\vec\Beta) \mid \vec\Beta \ll \sup D \right\}
&&\text{By \Cref{lem:fin-approx2}.}
\\
&= f^\ast(\sup D)
\end{align*}
\end{proof}

\section{Operations}

\subsection{Sequential Composition}

\begin{lemma}[Antitonicity of $\stuck$]\label{lem:stuck-antitone}
If $\alpha \lelpo\beta$, then $\stuck_\beta \Rightarrow \stuck_\alpha$.
\end{lemma}
\begin{proof}
For any $z\in \Bot_\beta$, either $z \in \Bot_\alpha$ or $z \notin N_\alpha$. In the first case, $\varphi_\beta(z) = \varphi_\alpha(z)$, being $\alpha \lelpo \beta$; in the second case, by \Cref{lem:lelpo-missing}, there is some $w \in \Bot_\alpha$ such that $w <_\beta z$, which implies that $\varphi_\beta(z) \Rightarrow\varphi_\alpha(w)$. This means that every term in $\bigvee_{z \in \Bot_\beta} \varphi_\beta(z)$ implies some term in $\bigvee_{y\in \Bot_\alpha} \varphi_\alpha(y)$, therefore $\bigvee_{z \in \Bot_\beta} \varphi_\beta(z) \Rightarrow \bigvee_{y\in \Bot_\alpha} \varphi_\alpha(y)$.
\end{proof}

\begin{lemma}[Monotonicity of $\extens$]\label{lem:extens-mono}
If $\alpha \lelpo \beta$, then $\extens_\alpha = \{ x \in \extens_\beta \mid \varphi_\beta(x) \not\Rightarrow \stuck_\alpha\}$.
\end{lemma}
\begin{proof}
We show the inclusion in both directions.

Take any $x\in\extens_\alpha$, this means that $x \in N_\alpha$ and $\varphi_\alpha(x) \not\Rightarrow \stuck_\alpha$; so, there is a valuation $v$ such that $v \vDash \varphi_\alpha(x)$ and $v \not\vDash \stuck_\alpha$. By \Cref{lem:stuck-antitone}, $\stuck_\beta \Rightarrow \stuck_\alpha$, therefore $v$ cannot satisfy $\stuck_\beta$ either. Since $\varphi_\alpha(x) = \varphi_\beta(x)$ (because $\alpha \lelpo \beta$), we have that 
$\varphi_\beta(x) \not\Rightarrow \stuck_\alpha$ and that
$v \vDash \varphi_\beta(x)$ and $v \not\vDash \stuck_\beta$; therefore, $\varphi_\beta(x) \not\Rightarrow \stuck_\beta$ and so $x\in \extens_\beta$.

For the reverse inclusion, take any $x \in \extens_\beta$ such that $\varphi_{\beta}(x) \not\Rightarrow \stuck_\alpha$. By \Cref{lem:lelpo-missing}, we know that either $x\in N_\alpha$ or $x \in \succplus{\Bot_\alpha}_\beta$. In the former case, we know that $\varphi_{\alpha}(x) = \varphi_\beta(x) \not\Rightarrow \stuck_\alpha$, therefore $x \in \extens_\alpha$. The latter case is a contradiction, since $x \in \succplus{\Bot_\alpha}_\beta$ means that there is some $y \in \Bot_\alpha$ such that $y <_\beta x$, which entails that $\varphi_\beta(x) \Rightarrow \varphi_\beta(y)$; however, $\varphi_\beta(y) = \varphi_\alpha(y) \Rightarrow \stuck_\alpha$ and, by combining these facts, we get that $\varphi_\beta(x) \Rightarrow \stuck_\alpha$.
\end{proof}

\begin{restatable}[Monotonicity of $\br$]{lemma}{brmono}\label{lem:br-mono}
If $\alpha \lelpo \beta$, then $\br_\alpha = \{ \psi \in \br_\beta \mid \psi \Rightarrow \lnot \stuck_\alpha \}$.
\end{restatable}
\begin{proof}

We show the inclusion in both directions.

For the forward inclusion, take $\psi \in \br_\alpha$, i.e. $\psi = \varphi_\alpha(S)$, for some $S \subseteq_\fin \extens_\alpha$ maximally satisfiable and that does not imply $\stuck_\alpha$; we want to show that $\psi = \varphi_\beta(S) \in \br_\beta$.
By \Cref{lem:extens-mono}, we know that $\extens_\alpha \subseteq \extens_\beta$, therefore $S \subseteq_\fin \extens_\beta$. Further, since $\varphi_\alpha(x) = \varphi_\beta(x)$ for all $x\in N_\alpha$, then $\varphi_\beta(S) = \varphi_\alpha(S) \Rightarrow \lnot \stuck_\alpha$; by using the contrapositive of \Cref{lem:stuck-antitone}, we get that $\varphi_\beta(S) \Rightarrow \lnot \stuck_\beta$.
Last, take any $T$ such that $S \subset T \subseteq \extens_\beta$. If $T \subseteq \extens_\alpha$, we already know that $\varphi_\alpha(T)$ cannot be satisfiable; so suppose 
that there is some node $x \in T$ such that $x \notin \extens_\alpha$, meaning that either $x \notin N_\alpha$ or $\varphi_\alpha(x) \Rightarrow \stuck_\alpha$.
In the first case, by \Cref{lem:lelpo-missing} $x \in \succplus{\Bot_\alpha}_\beta$, so there is a $z \in \Bot_\alpha$ such that $z <_\beta x$, which also means that $\varphi_\beta(x) \Rightarrow \varphi_\beta(z) \Rightarrow \stuck_\alpha$. So, in either case, we have $\varphi_\beta(x) \Rightarrow \stuck_\alpha$, while, since $\varphi_\alpha(S) \in\br_\alpha$, we know that $\varphi_\alpha(S) \Rightarrow \lnot\stuck_\alpha$. These give us:
\begin{align*}
  \varphi_\beta(T)
  &\Rightarrow \varphi_\beta(S) \land \varphi_\beta(x)
  \\
  &= \varphi_\alpha(S) \land \varphi_\beta(x)
  \\
  &\Rightarrow \lnot\stuck_\alpha \land \stuck_\alpha  
  \\
  &\Rightarrow \fls
\end{align*}
So, $\varphi_\beta(T)$ is not satisfiable, and therefore we have shown that $\varphi_\beta(S) \in \br_\beta$ (and, clearly, $\varphi_\beta(S) \Rightarrow\lnot\stuck_\alpha$).

For the reverse inclusion, take any $\psi\in \br_\beta$ such that $\psi \Rightarrow \lnot \stuck_\alpha$. We know that $\psi = \varphi_\beta(S)$ where $S \subseteq \extens_\beta$. By \Cref{lem:extens-mono}, 
for each $x \in S$, either $x\in\extens_\alpha$ or $\varphi_\beta(x) \Rightarrow \stuck_\alpha$. But the latter case is not possible, otherwise we would be able to contradict $\varphi_\beta(S) = \psi \Rightarrow \lnot\stuck_\alpha$. Therefore, $S \subseteq \extens_\alpha$. Now take any $T$ such that $S \subset T \subseteq_\fin \extens_\alpha$. Since $\psi \in \br_\beta$, we know that $\varphi_\beta(T)$ is not satisfiable and, since $\varphi_\alpha(T) = \varphi_\beta(T)$, then $\varphi_\alpha(T)$ is not satisfiable either. Therefore, $\psi\in \br_\alpha$.
\end{proof}

\begin{lemma}[Monotonicity of $\fatsemi$ for \lpof s]\label{lem:mon-seq-lpo}
For all $\alpha,\alpha',\beta,\beta' \in \lpo_\fin(L)$ and $f \in \copylpo_{\alpha',\beta'}$ such that $\alpha \lelpo \alpha'$ and $\beta\lelpo\beta'$, there exists $g\in\copylpo_{\alpha,\beta}$ such that $\alpha \fatsemi_g\beta \lelpo \alpha'\fatsemi_f\beta'$.
\end{lemma}
\begin{proof}
Let $\beta'_\psi = f(\psi)$ for each $\psi \in \br_{\alpha'}$. Since $\beta' \equiv f(\psi)$ for all $\psi$, then for each $\psi \in \br_{\alpha'}$ there exists a bijection $f_\psi \colon N_{\beta'} \to N_{\beta'_\psi}$. Since $\beta \lelpo \beta'$, then $N_\beta \subseteq N_{\beta'}$; so, we can restrict $f_\psi$ to $N_\beta$ and obtain an lpo-isomorphism $g_\psi$ (that clearly satisfies $g_\psi(x) = f_\psi(x)$ for all $x \in N_\beta$). Now, let $\beta_\psi$ be an isomorphic copy of $\beta$ generated by the bijection $g_\psi$. Clearly, $N_{\beta_\psi} \cap N_\alpha = \emptyset$, since $N_\alpha \subseteq N_{\alpha'}$ and $N_{\beta_\psi} \subseteq N_{\beta'_\psi}$; similarly, all the $\beta_\psi$s have pairwise disjoint nodes. By \Cref{lem:br-mono}, $\br_\alpha \subseteq \br_{\alpha'}$; so, we have such a $\beta_\psi$ for all $\psi \in \br_\alpha$, and we can define $g \in \copylpo_{\alpha,\beta}$ as $g(\psi) = \beta_\psi$ for all $\psi \in \br_\alpha$.
Also, $\beta_\psi = g_\psi(\beta) \lelpo f_\psi(\beta') = \beta'_\psi$ for all $\psi \in \br_\alpha$, since $\beta \lelpo \beta'$ and $g_\psi \subseteq f_\psi$.

We start with proving that $N_{\alpha\fatsemi_g\beta} \dwclosed N_{\alpha'\fatsemi_f \beta'}$. Take any $x \in N_{\alpha\fatsemi_g\beta}$; so, either $x \in N_\alpha (\subseteq N_{\alpha'})$ or $x \in N_{\beta_\psi} (\subseteq N_{\beta'_\psi})$, for some $\psi \in \br_\alpha$. Clearly $x \in N_{\alpha'\fatsemi_f \beta'}$, so the subset inclusion holds; we now show that it is downward closed. Take any $y <_{\alpha' \fatsemi_f \beta'} x$.
If $x \in N_\alpha$, then clearly $y <_{\alpha'} x$, since $\fatsemi$ does not add any causality into elements of $N_{\alpha'}$, and we therefore conclude that $y \in N_\alpha \subseteq N_{\alpha \fatsemi_g\beta}$, since $N_\alpha \dwclosed N_{\alpha'}$.
If instead $x \in N_{\beta_\psi}$ for some $\psi\in \br_\alpha$, then we know by \Cref{lem:br-mono} that $\psi \in \br_{\alpha'}$ and therefore $x \in N_{\beta'_\psi} \subseteq N_{\alpha'\fatsemi_f \beta'}$. Since $y <_{\alpha' \fatsemi_f \beta'} x$, then either $y <_{\beta'_\psi} x$, or $\psi \Rightarrow \varphi_{\alpha'}(y)$. In the former case, then clearly $y \in N_{\beta_\psi} \subseteq N_{\alpha\fatsemi_g\beta}$ since $N_{\beta_\psi} \dwclosed N_{\beta'_\psi}$. In the latter case, suppose for the sake of contradiction that $y \notin N_\alpha$; then, by \Cref{lem:lelpo-missing} there is some $z\in\Bot_\alpha$ such that $z <_{\alpha'} y$, which also means that $\varphi_{\alpha'}(y) \Rightarrow \varphi_{\alpha'}(z)$. Since $\alpha \lelpo \alpha'$ and $z\in\Bot_\alpha$, we have that $\psi \Rightarrow \varphi_{\alpha'}(y) \Rightarrow \varphi_{\alpha'}(z) \Leftrightarrow \varphi_\alpha(z) \Rightarrow \stuck_\alpha$; this contradicts $\psi \Rightarrow \lnot\stuck_\alpha$, that holds since $\psi \in \br_\alpha$. 

Next, we show the condition on the order:
\begin{align*}
  & \mathord{<_{\alpha'\fatsemi_f\beta'}} \cap (N_{\alpha\fatsemi_g\beta} \times N_{\alpha\fatsemi_g\beta})
  \\
   &\qquad = \left( \mathord{<_{\alpha'}}  \cup
   \bigcup_{\psi \in \br_{\alpha'}} \left( \mathord{<_{\beta'_{\psi}}}
      \cup (\{ x \mid \psi \Rightarrow \varphi_{\alpha'}(x) \} \times N_{\beta'_\psi}) \right) \right)
       \cap (N_{\alpha\fatsemi_g\beta} \times N_{\alpha\fatsemi_g\beta})
   \\
   &\qquad = \left(\mathord{<_{\alpha'}} \cap (N_{\alpha\fatsemi_g\beta} \times N_{\alpha\fatsemi_g\beta})\right)
       \\
       &\qquad \cup \bigcup_{\psi \in \br_{\alpha'}}
         \left( \mathord{<_{\beta'_{\psi}}} \cap (N_{\alpha\fatsemi_g\beta} \times N_{\alpha\fatsemi_g\beta}) \right)
          \cup \left( (\{ x \mid \psi \Rightarrow \varphi_{\alpha'}(x) \} \times N_{\beta'_\psi}) \cap (N_{\alpha\fatsemi_g\beta} \times N_{\alpha\fatsemi_g\beta}) \right) \qquad (\dagger)
\end{align*}
Now, $N_\alpha$ is a subset of $N_{\alpha'}$ and it is disjoint from all the $N_{\beta_\psi}$; moreover, $\alpha \lelpo \alpha'$. So,
\[
\mathord{<_{\alpha'}} \cap (N_{\alpha\fatsemi_g\beta} \times N_{\alpha\fatsemi_g\beta}) = \mathord{<_{\alpha'}} \cap (N_{\alpha} \times N_{\alpha}) = \mathord{<_{\alpha}}
\]
Then, by \Cref{lem:br-mono} $\br_{\alpha} \subseteq \br_{\alpha'}$, and, for all $\psi \in \br_{\alpha'} \setminus \br_{\alpha}$, the nodes of $\beta'_\psi$ do not appear in $N_{\alpha \fatsemi_g \beta}$; hence, we can reduce the bounds of the union in ($\dagger$) and obtain:
\begin{align*}
& \bigcup_{\psi \in \br_{\alpha'}}
         \left( \mathord{<_{\beta'_{\psi}}} \cap (N_{\alpha\fatsemi_g\beta} \times N_{\alpha\fatsemi_g\beta}) \right)
          \cup \left( (\{ x \mid \psi \Rightarrow \varphi_{\alpha'}(x) \} \times N_{\beta'_\psi}) \cap (N_{\alpha\fatsemi_g\beta} \times N_{\alpha\fatsemi_g\beta}) \right)  
          \\
   &\qquad = \bigcup_{\psi \in \br_{\alpha}}
         \left( \mathord{<_{\beta'_{\psi}}} \cap (N_{\alpha\fatsemi_g\beta} \times N_{\alpha\fatsemi_g\beta}) \right)
          \cup \left( (\{ x \mid \psi \Rightarrow \varphi_{\alpha'}(x) \} \times N_{\beta'_\psi}) \cap (N_{\alpha\fatsemi_g\beta} \times N_{\alpha\fatsemi_g\beta}) \right)
   \\
   &\qquad = \bigcup_{\psi \in \br_{\alpha}}
         \left( \mathord{<_{\beta'_{\psi}}} \cap (N_{\beta_\psi} \times N_{\beta_\psi}) \right)
          \cup  (\{ x \mid \psi \Rightarrow \varphi_{\alpha'}(x) \} \times N_{\beta_\psi})
   \\
   & \qquad = \bigcup_{\psi \in \br_{\alpha}}
         \left(\mathord{<_{\beta_{\psi}}}
          \cup  (\{ x \mid \psi \Rightarrow \varphi_{\alpha}(x) \} \times N_{\beta_\psi})\right)
\end{align*}
where the last two equalities hold because $\beta_\psi \lelpo \beta'_\psi$ and because the nodes of $\beta_\psi'$ are disjoint from those of $\alpha$ and of any other isomorphic copy of $\beta$.
Thus, ($\dagger$) is equal to $\mathord{<_{\alpha}} \cup \bigcup_{\psi \in \br_{\alpha}} \left ( \mathord{<_{\beta_{\psi}}} \cup  (\{ x \mid \psi \Rightarrow \varphi_{\alpha}(x) \} \times N_{\beta_\psi})\right )$ that, by definition, is $\mathord{<_{\alpha\fatsemi\beta}}$.

Now, take any $x \in N_{\alpha \fatsemi_g \beta}$. 
Since $\alpha \lelpo \alpha'$ and $\beta_\psi \lelpo \beta_{\psi}'$, we have that $\lambda_\alpha(x) \le \lambda_{\alpha'}(x)$, if $x \in N_\alpha$, and $\lambda_{\beta_\psi}(x) \le \lambda_{\beta'_\psi}(x)$, if $x \in N_{\beta_\psi}$. This gives us:
\begin{align*}
  \lambda_{\alpha\fatsemi_g \beta}(x)
  &= \left\{
    \begin{array}{ll}
      \lambda_\alpha(x) & \text{if} ~ x \in N_\alpha
      \\
      \lambda_{\beta_\psi}(x) & \text{if} ~ x \in N_{\beta_\psi}
    \end{array}
  \right\}\ 
  \le\ \left\{
    \begin{array}{ll}
      \lambda_{\alpha'}(x) & \text{if} ~ x \in N_\alpha
      \\
      \lambda_{\beta'_\psi}(x) & \text{if} ~ x \in N_{\beta_\psi}
    \end{array}
  \right\}
  = \lambda_{\alpha'\fatsemi_f \beta'}(x)
\end{align*}
Similarly, we have that $\varphi_\alpha(x) = \varphi_{\alpha'}(x)$, if $x \in N_\alpha$, and $\varphi_{\beta_\psi}(x) = \varphi_{\beta'_\psi}(x)$, if $x \in N_{\beta_psi}$; so
\begin{align*}
  \varphi_{\alpha\fatsemi_g \beta}(x)
  &= \left\{
    \begin{array}{ll}
      \varphi_\alpha(x) & \text{if} ~ x \in N_\alpha
      \\
      \varphi_{\beta_\psi}(x) \land \psi & \text{if} ~ x \in N_{\beta_\psi}
    \end{array}
  \right\}\ 
  = \ \left\{
    \begin{array}{ll}
      \varphi_{\alpha'}(x) & \text{if} ~ x \in N_\alpha
      \\
      \varphi_{\beta'_\psi}(x) \land\psi & \text{if} ~ x \in N_{\beta_\psi}
    \end{array}
  \right\}
  = \varphi_{\alpha'\fatsemi_f \beta'}(x)
\end{align*}
For successors, we know by \Cref{lem:succBot}(3) that $\succ_{\alpha\fatsemi_g\beta}(x) \subseteq \succ_{\alpha'\fatsemi_f\beta'}(x) \setminus \succplus{\Bot_{\alpha\fatsemi_g\beta}}_{\alpha'\fatsemi_f\beta'}$, we now must show the reverse inclusion.
Take any $y\in \succ_{\alpha'\fatsemi_f\beta'}(x)$ such that $y \notin \succplus{\Bot_{\alpha\fatsemi_g\beta}}_{\alpha'\fatsemi_f\beta'}$. Given the way that $<_{\alpha'\fatsemi_f\beta'}$ is constructed, there are three cases to consider:

\begin{enumerate}

\item $y \in \succ_{\alpha'}(x)$. Since $\alpha \lelpo\alpha'$, we know that $\succ_\alpha(x) = \succ_{\alpha'}(x)\setminus \succplus{\Bot_\alpha}_{\alpha'}$. However, $y \notin \succplus{\Bot_{\alpha}}_{\alpha'}$ (since $\succplus{\Bot_{\alpha}}_{\alpha'} \subseteq \succplus{\Bot_{\alpha\fatsemi_g\beta}}_{\alpha'\fatsemi_f\beta'}$); therefore, it must be that $y \in \succ_\alpha(x)$, and so $y \in \succ_{\alpha\fatsemi_g\beta}(x)$.

\medskip
\item $y \in \succ_{\beta'_\psi}(x)$ for some $\psi \in \br_\alpha$. This means that $x,y \in N_{\beta'_\psi}$. Let $x' = f_\psi^{-1}(x)$ and $y' = f_\psi^{-1}(y)$, so $x', y' \in N_{\beta'}$.
Since $\beta' \equiv \beta'_\psi$, we know that $y' \in \succ_{\beta'}(x')$. Since $\beta \lelpo \beta'$, then $\succ_\beta(x') = \succ_{\beta'}(x') \setminus \succplus{\Bot_{\beta}}_{\beta'}$. However, since $y \notin \succplus{\Bot_{\alpha\fatsemi\beta}}_{\alpha'\fatsemi\beta'}$, then it must be that $y' \notin \succplus{\Bot_\beta}_{\beta'}$; so we must have that $y' \in \succ_\beta(x')$. Since $x',y' \in N_\beta$, then $g_\psi(x') = f_\psi(x') = x$ and $g_\psi(y') = f_\psi(y') = y$; so, $y \in \succ_{\beta_\psi}(x)$ and therefore $y \in \succ_{\alpha\fatsemi_g\beta}(x)$.

\medskip
\item $x \in N_{\alpha}$ and $y \in N_{\beta'_\psi}$ for some $\psi \in \br_{\alpha'}$.
By \Cref{lem:br-mono}, either $\psi \in \br_\alpha$ or $\psi \not\Rightarrow \lnot\stuck_\alpha$.

For the first case, $x <_{\alpha'\fatsemi_f\beta'} y$, $x \in N_{\alpha'}$ and $y \in N_{\beta'_\psi}$ imply that $\psi \Rightarrow \varphi_{\alpha'}(x) = \varphi_\alpha(x)$. Furthermore, $y$ must be the root of $\beta'_\psi$, otherwise it could not be a successor of $x \in N_{\alpha'}$. Since $\beta'_\psi$ is single-rooted and $\beta_\psi \lelpo \beta'_\psi$, then $y$ must also be the root of $\beta_\psi$; so, $x <_{\alpha \fatsemi_g \beta} y$ and, by \Cref{lem:succBot}(1), $y \in \succ_{\alpha\fatsemi_g\beta}(x)$.

To conclude, we show that the second case is impossible. Suppose that $\psi \not\Rightarrow\lnot \stuck_\alpha$.
From $\psi \in \br_{\alpha'}$ we also know that $\psi \Rightarrow \lnot\stuck_{\alpha'}$. Therefore, there is some $v$ such that $v \vDash \psi$, $v \vDash \stuck_\alpha$ and $v \not\vDash \stuck_{\alpha'}$. From $v\vDash\stuck_\alpha$, we get that $v \vDash \varphi_\alpha(z)$, for some $z\in\Bot_\alpha$; so $v\vDash \varphi_{\alpha'}(z)$ too (since $\varphi_{\alpha}(z) = \varphi_{\alpha'}(z)$).
Since $v\vDash \varphi_{\alpha'}(z)$ and $v\not\vDash\stuck_{\alpha'}$, then $\varphi_{\alpha'}(z) \not\Rightarrow\stuck_{\alpha'}$; therefore $z \in \extens_{\alpha'}$.
Moreover, since $v \vDash \psi$ and $v \vDash \varphi_{\alpha'}(z)$, we know that $\psi \land \varphi_{\alpha'}(z)$ is satisfiable; since $S$ is maximal, then $z\in S$. Therefore, $\psi \Rightarrow \varphi_{\alpha'}(z)$, and so $z <_{\alpha'\fatsemi_f\beta'} y$ (by the definition of $<_{\alpha'\fatsemi_f\beta'}$), in contradiction with $y \notin \succplus{\Bot_{\alpha\fatsemi_g\beta}}_{\alpha'\fatsemi_f\beta'}$. 
\qedhere
\end{enumerate}
\end{proof}

\begin{restatable}[Monotonicity of $\fatsemi$]{lemma}{monSeq}
\label{lem:monSeq}
If $\Alpha \lepom \Alpha'$ and $\Beta \lepom \Beta'$, then $\Alpha \fatsemi \Beta \lepom \Alpha' \fatsemi \Beta'$
\end{restatable}
\begin{proof}
Take any $\gamma \in \Alpha'\fatsemi\Beta'$; by construction, $\gamma = \alpha' \fatsemi_f \beta'$, for some $\alpha' \in \Alpha'$, $\beta'\in\Beta'$, and $f \in \copylpo_{\alpha',\beta'}$. By \Cref{lem:altLePom}(3), there exist $\alpha \in \Alpha$ and $\beta\in\Beta$ such that $\alpha\lelpo\alpha'$ and $\beta\lelpo\beta'$. So, by \Cref{lem:mon-seq-lpo} there exists a $g \in \copylpo_{\alpha,\beta}$ such that $\alpha\fatsemi_g \beta \lelpo \alpha'\fatsemi_f \beta'$. So, by \Cref{lem:altLePom}(3) $\Alpha\fatsemi\Beta \lepom \Alpha'\fatsemi\Beta'$.
\end{proof}

\begin{corollary}
\label{cor:seqCont}
$\fatsemi$ is Scott continuous.
\end{corollary}
\begin{proof}
Immediate from \Cref{lem:extension,lem:monSeq}.
\end{proof}

\subsection{Guarded Branching}

\begin{lemma}[Monotonicity of $\guard$ for \lpof s]
\label{lem:monGuardLPO}
For every $\ell \in L$, $x \in \Nodes$ and $\alpha, \alpha', \beta, \beta' \in \lpo_\fin(L)$ such that $\alpha \lelpo \alpha'$, $\beta \lelpo \beta'$, $N_{\alpha' }\cap N_{\beta'} = \emptyset$ and $x\not\in N_{\alpha' }\cup N_{\beta'}$, it holds that $\guard(x,\ell,\alpha,\beta) \lelpo \guard(x,\ell,\alpha',\beta')$.
\end{lemma}
\begin{proof}
To lighten notations, let us fix $\ell$, $\alpha, \alpha', \beta, \beta'$ and $x$, 
and denote $\gamma \triangleq \guard(x,\ell,\alpha,\beta)$ and $\gamma' \triangleq \guard(x,\ell,\alpha',\beta')$.

First, being $N_\alpha \subseteq N_\beta$ and $N_{\alpha'} \subseteq N_{\beta'}$, we also have that
$N_{\gamma} = \{x\} \cup N_\alpha \cup N_\beta \subseteq \{x\} \cup N_{\alpha'} \cup N_{\beta'} = N_{\gamma'}$;
we have to prove that $N_{\gamma}$ is downward closed.
Let $y \in N_{\gamma}$ and $z <_{\gamma'} y$; clearly, $y \neq x$, because $\pred_{\gamma'}(x) = \emptyset$.
Assume that  $y \in N_\alpha$ ($\subseteq N_{\alpha'}$); the case for $y \in N_\beta$ is identical.
If $z = x$, we trivially conclude. If $z \in N_{\alpha'}$, we have $z <_{\alpha'} y$, which leads to $z \in N_\alpha$ (because $N_\alpha \dwclosed N_{\alpha'}$); hence, $z \in N_\gamma$.

Second, 
\[
\begin{array}{lcl}
\multicolumn{3}{l}{<_{\gamma'} \cap\ (N_{\gamma} \times N_{\gamma})}
\vspace*{.15cm}
\\
\quad & = &
(<_{\alpha'} \cup <_{\beta'} \cup\ (\{x\} \times (N_{\alpha'} \cup N_{\beta'})))\ \cap\ (N_{\gamma} \times N_{\gamma})
\vspace*{.15cm}
\\
\quad & = &
(<_{\alpha'} \cap\ ( (\{x\} \cup N_\alpha \cup N_\beta) \times (\{x\} \cup N_\alpha \cup N_\beta)))
\\
&& \cup\ (<_{\beta'} \cap\ ( (\{x\} \cup N_\alpha \cup N_\beta) \times (\{x\} \cup N_\alpha \cup N_\beta)))
\\
&& \cup\ 
((\{x\} \times (N_{\alpha'} \cup N_{\beta'})) \cap ( (\{x\} \cup N_\alpha \cup N_\beta) \times (\{x\} \cup N_\alpha \cup N_\beta)))
\vspace*{.15cm}
\\
\quad & = &
(<_{\alpha'} \cap\ ( N_\alpha \times N_\alpha))
\ \cup\ (<_{\beta'} \cap\ ( N_\beta \times N_\beta))
\ \cup\ (\{x\} \times (N_{\alpha} \cup N_{\beta}))
\vspace*{.15cm}
\\
&=&
<_\alpha \cup <_\beta \cup \ (\{x\} \times (N_\alpha \cup N_\beta))
\ \ =\ \ <_{\gamma}
\end{array}
\]

Fix $y \in N_{\gamma}$; we have three subcases to consider:
\begin{enumerate}
\item $y=x$: by construction, $\lambda_{\gamma}(y) = \lambda_{\gamma'}(y) = \ell$ and $\varphi_\gamma(y) = \varphi_{\gamma'} = \tru$.
\item $y \in N_\alpha$: then, $\lambda_{\gamma}(y) = \lambda_\alpha(y) \leq \lambda_{\alpha'}(y) = \lambda_{\gamma'}(y)$ and  
$\varphi_{\gamma}(y) = \varphi_\alpha(y) \wedge x = \varphi_{\alpha'}(y) \wedge x = \lambda_{\gamma'}(y)$.
\item $y \in N_\beta$: this case is similar to previous one, with $\neg x$ in place of $x$ for the $\varphi$ case.
\end{enumerate}
To conclude, we need to prove that $\succ_\gamma(y) \supseteq \succ_{\gamma'}(y) \setminus \succplus{\Bot_\gamma}_{\gamma'}$ (since the other inclusion holds for \Cref{lem:succBot}). Thus, take $z \in \succ_{\gamma'}(y)$; we have three cases to consider:
\begin{enumerate}
\item $y=x$: by construction, $z \in \{\min_{\alpha'},\min_{\beta'}\}$ and so we immediately conclude $z \in \succ_{\gamma}(y)$, being $\min_\alpha = \min_{\alpha'}$ and $\min_\beta = \min_{\beta'}$.

\item $y \in N_{\alpha}$: by construction, $z \in \succ_{\alpha'}(y)$.
Since $\alpha \lelpo \alpha'$, we know that $\succ_\alpha(y) = \succ_{\alpha'}(y) \setminus \succplus{\Bot_\alpha}_{\alpha'}$.
So, either $z \in \succ_{\alpha}(y) = \succ_{\gamma}(y)$ or $z \in \succplus{\Bot_\alpha}_{\alpha'} \subseteq \succplus{\Bot_\gamma}_{\gamma'}$.

\item $y \in N_{\beta}$: this case is similar to previous one.
\qedhere
\end{enumerate}
\end{proof}

\begin{restatable}[Monotonicity of $\guard$]{lemma}{monGuard}
\label{lem:monGuard}
If $\Alpha \lepom \Alpha'$ and $\Beta \lepom \Beta'$, then 
$\guard(\ell, \Alpha, \Beta) \lepom \guard(\ell, \Alpha',\Beta')$.
\end{restatable}
\begin{proof}
Take any $\gamma \in \guard(\ell,\Alpha',\Beta')$; by construction, $\gamma = \guard(x,\ell,\alpha', \beta')$, for some $\alpha' \in \Alpha'$ and $\beta'\in\Beta'$ such that $N_{\alpha'} \cap N_{\beta'} = \emptyset$ and $x \notin N_{\alpha'}\cup N_{\beta'}$. By \Cref{lem:altLePom}(3), there exist $\alpha \in \Alpha$ and $\beta\in\Beta$ such that $\alpha\lelpo\alpha'$ and $\beta\lelpo\beta'$. So, by \Cref{lem:monGuardLPO} $\guard(x,\ell,\alpha, \beta) \lelpo \guard(x,\ell,\alpha', \beta')$. So, by \Cref{lem:altLePom}(3) $\guard(\ell,\Alpha,\Beta) \lepom \guard(\ell,\Alpha',\Beta')$.
\end{proof}

\begin{corollary}
\label{cor:guardCont}
$\guard$ is Scott continuous.
\end{corollary}
\begin{proof}
Immediate from \Cref{lem:monGuard,lem:extension}.
\end{proof}


\section{Denotational Semantics}

\subsection{Linearization}

\begin{lemma}\label{lem:next-mono}
If $\alpha \lelpo \beta$ and $S \cap \Bot_\alpha = \emptyset$, then $\mathsf{next}(\alpha, \psi, S) = \mathsf{next}(\beta, \psi, S)$.
\end{lemma}
\begin{proof}
We show the set inclusion in both directions. Take any $y \in \mathsf{next}(\alpha, \psi, S)$, so $y \in N_\alpha$, $y 
\notin S$, $\predplus{y}_\alpha \subseteq S$, and $\psi \Rightarrow \varphi_\alpha(y)$. Since $\alpha \lelpo \beta$, we get that $y \in N_\beta$, $\predplus{y}_\alpha = \predplus{y}_\beta$, and $\varphi_\alpha(y) = \varphi_\beta(y)$, so $y \in \mathsf{next}(\beta, \psi, S)$.

For the reverse inclusion, take any $y \in \mathsf{next}(\beta, \psi, S)$, so $y \in N_\beta$, $y 
\notin S$, $\predplus{y}_\beta \subseteq S$, and $\psi \Rightarrow \varphi_\beta(y)$. If $y \in N_\alpha$, then the proof follows from a similar argument as above. If not, we will show that there is a contradiction. Suppose that $y \notin N_\alpha$, so by \Cref{lem:lelpo-missing}, we get that $y \in \succplus{\Bot_\alpha}_\beta$. So there is a $z\in \Bot_\alpha$ such that $z <_\beta y$. Since $\predplus{y}_\beta \subseteq S$, then $z \in S$. However, this is a contradiction since we assumed that $S \cap \Bot_\alpha = \emptyset$.
\end{proof}

For any set $X$, the pointwise extension $\tuple{X \to D(\st), \sqsubseteq^\bullet}$ is also a pointed \dcpo, where $f \sqsubseteq^\bullet g$ iff $f(x) \sqsubseteq g(x)$ for all $x\in X$. We use this fact below.

\begin{lemma}[Monotonicity of $\linlpo$]
\label{lem:linlpo-mono}
For any $\alpha,\beta \in \lpo_\fin(L)$, $\psi \in \Form$, $S \subseteq N_\alpha$, and $s \in \mathcal S$ such that $\alpha \lelpo \beta$ and $S \cap \Bot_\alpha = \emptyset$:
\[
  \linlpo(\alpha, \psi, S) \sqsubseteq^\bullet \linlpo(\beta, \psi, S)
\]
\end{lemma}
\begin{proof}
  Let $T_{S,\psi} = \{ x \in N_\alpha \mid x \notin S, \sat(\psi \land \varphi_\alpha(x)) \}$ be the set of nodes still to schedule. The proof is by induction on the size of the set $T_{S,\psi}$.
  
  If $T_{S,\psi}$ is empty, then $x \in S$ or $\psi \land \varphi_\alpha(x)$ is unsatisfiable for all $x \in N_\alpha$. In either case, clearly $x \notin \mathsf{next}(\alpha,\psi,S)$, therefore $\mathsf{next}(\alpha,\psi,S) = \emptyset$. Moreover, by \Cref{lem:next-mono}, $\mathsf{next}(\beta,\psi,S) =\emptyset$ too and so $\linlpo(\alpha,\psi,S)(s) = \linlpo(\beta,\psi,S)(s) = \eta(s)$.
  
  Now suppose the claim holds for all sets smaller than $T_{S,\psi}$, and let $f$ be defined as follows:
  \[
    f(\gamma,\psi,S,x)(s) = \left\{
      \begin{array}{ll}
        \linlpo(\gamma,\psi, S \cup\{ x \})^\dagger(\de{\lambda_\gamma(x)}_\act(s)) & \text{if} ~ \lambda_\gamma(x) \in \act
        \\
        \linlpo(\gamma,\psi \land \sem{x = \de{\lambda_\gamma(x)}(s)}, S \cup\{x\})(s) & \text{if}~ \lambda_\gamma(x) \in \test
        \\
        \bot & \text{if}~ \lambda_\gamma(x) = \bot
        \\
        \linlpo(\gamma,\psi, S\cup\{x\})(s) & \text{if} ~ \lambda_\gamma(x) = \fork
      \end{array}
    \right.
  \]
  So, $\linlpo(\gamma,\psi,S)(s) = \bignd_{x \in \mathsf{next}(\gamma,\psi,S)} f(\gamma,\psi,S,x)(s)$. We start by showing that $f(\alpha,\psi,S,x)(s) \sqsubseteq f(\beta,\psi,S,x)(s)$. First, we remark that $S\cap\Bot_\alpha = \emptyset$ and so, if $x \notin\Bot_\alpha$, then $(S\cup\{x\})\cap \Bot_\alpha = \emptyset$, which will allow us to use the induction hypothesis in cases (1), (2), and (4) below. We proceed by case analysis on $\lambda_\alpha(x)$.
  \begin{enumerate}
     \item If $\lambda_\alpha(x) \in \act$, then $\lambda_\beta(x) \in \act$ too, and $\de{\lambda_\alpha(x)}_\act(s) \sqsubseteq \de{\lambda_\beta(x)}_\act(s)$. Note that $T_{S\cup \{x\},\psi} \subset T_{S,\psi}$, so by the induction hypothesis, we can conclude that:
\[
  \linlpo(\alpha, \psi, S \cup \{x\}) \sqsubseteq^\bullet \linlpo(\beta, \psi, S\cup \{x\})
\]
Therefore, by monotonicity of Kleisli extension, we get that:
\begin{align*}
  f(\alpha, \psi,S,x)(s)
  &= \linlpo(\alpha, \psi, S\cup\{x\})^\dagger(\de{\lambda_\alpha(x)}_\act(s))
  \\
  &\sqsubseteq \linlpo(\beta, \psi, S\cup\{x\})^\dagger(\de{\lambda_\beta(x)}_\act(s))
  \\
  &= f(\beta, \psi,S,x)(s)
\end{align*}

\item If $\lambda_\alpha(x) \in \test$, then again $\lambda_\beta(x) \in \test$ as well and $\lambda_\alpha(x) = \lambda_\beta(x)$ since the order over tests is flat. Since $T_{S\cup\{x\},\psi\land \sem{x = \de{b}_\test(s)}} \subset T_{S,\psi}$, then we can use the induction hypothesis to conclude that:
\begin{align*}
  f(\alpha, \psi,S,x)(s)
  &= \linlpo(\alpha, \psi\land \sem{x = \de{\lambda_\alpha(x)}_\test(s)}, S\cup\{x\})(s)
  \\
  &\sqsubseteq \linlpo(\beta, \psi\land \sem{x = \de{\lambda_\beta(x)}_\test(s)}, S\cup\{x\})(s)
  \\
  &= f(\beta, \psi,S,x)(s)
\end{align*}

\item If $\lambda_\alpha(x) = \bot$, then we have:
\[
  f(\alpha, \psi,S,x)(s)
  \quad=\quad \bot
  \quad\sqsubseteq\quad f(\beta, \psi,S,x)(s)
\]

\item If $\lambda_\alpha(x) = \fork$, then $\lambda_\beta(x) = \fork$, and so since  $T_{S\cup \{x\},\psi} \subset T_{S,\psi}$, the induction hypothesis gives us the following:
\[
  f(\alpha, \psi,S,x)(x)
  \;\;=\;\; \linlpo(\alpha, \psi, S\cup\{x\})(x)
  \;\;\sqsubseteq\;\; \linlpo(\beta, \psi, S\cup\{x\})(x)
  \;\;=\;\; f(\beta, \psi,S,x)(x)
\]
\end{enumerate}

We now complete the proof of the main claim as follows:
\begin{align*}
  \linlpo(\alpha,\psi, S)(s)
  &= \bignd_{x \in \mathsf{next}(\alpha,\psi,S)} f(\alpha,\psi,S,x)(s)
  \intertext{Since $S\cap\Bot_\alpha = \emptyset$, by \Cref{lem:next-mono} we know that $\mathsf{next}(\alpha, \psi,S) = \mathsf{next}(\beta,\psi,S)$.}
  &= \bignd_{x \in \mathsf{next}(\beta,\psi,S)} f(\alpha,\psi,S,x)(s)
  \intertext{Since $\nd$ is monotone, and given the claim we just proved about $f$:}
  &\sqsubseteq \bignd_{x \in \mathsf{next}(\beta,\psi,S)} f(\beta,\psi,S,x)(s)
  \\
  &= \linlpo(\beta,\psi, S)(s)
\qedhere
\end{align*}
\end{proof}

\begin{restatable}[Monotonicity of $\linfin$]{lemma}{monLin}
\label{lem:monLin}
If $\Alpha \lelpo \Beta$, then $\lin_\fin(\Alpha) \sqsubseteq^\bullet \lin_\fin(\Beta)$.
\end{restatable}
\begin{proof}
Fix any $\alpha\in \Alpha$, then we know there exists a $\beta\in\Beta$ such that $\alpha\lelpo\beta$:
\begin{align*}
  \linfin([\alpha])
  &= \linlpo(\alpha, \tru, \emptyset)
  \\
  & \sqsubseteq^\bullet \linlpo(\beta, \tru, \emptyset)
  &&\text{By \Cref{lem:linlpo-mono}}
  \\
  &= \linfin([\beta])
&& \qedhere
\end{align*}
\end{proof}

\begin{corollary}
$\lin \colon \pom \to (\st \to D(\st))$ is Scott continuous.
\end{corollary}
\begin{proof}
Immediate from \Cref{lem:extension,lem:monLin}.
\end{proof}

\subsection{Properties of Linearization: Skip and Actions}

\begin{restatable}{lemma}{linSkip}\label{lem:lin-skip}
\textnormal{$\lin(\de{\skp}) = \eta$} and $\lin(\de{a}) = \de{a}_\act$.
\end{restatable}
\begin{proof}
We show only the case for $\skp$, as the case for actions is nearly identical.
\begin{align*}
  \lin(\de{\skp})(s)
  &= \lin(\singleton{\fork})(s)
  \intertext{Since a singleton pomset is already finite, $\lin = \linfin$. Fixing an arbitrary $x\in\Nodes$, we have:}
  &= \linfin([\singleton{\fork}_x])(s)
  \\
  &= \linlpo(\singleton{\fork}_x, \tru, \emptyset)(s)
  \intertext{Obviously, $\next(\singleton{\fork}_x, \tru, \emptyset) = \{x\}$, so we get:}
  &= \linlpo(\singleton{\fork}_x, \tru, \{x\})(s)
  \\
  &= \eta(s)
\qedhere
\end{align*}
\end{proof}

\subsection{Properties of Linearization: Sequential Composition}

We first introduce a few definitions:
\begin{align*}
  \next^\ast(\alpha, \psi, S) &= \{ x \in N_\alpha \setminus S \mid \sat(\psi \land \varphi_\alpha(x)) \}
  \\
  \mathsf{tests}_\alpha(S) &= \{ x \in S \mid \lambda_\alpha(x) \in \test \}
  \\
  \forms(X) &= \set{ \bigwedge_{x \in X} \sem{x = b_x} ~\middle|~ \forall x\in X.\ b_x \in \mathbb B }
  \\
  \valid_\alpha(S) &= \{ \psi \mid \psi \in \forms(\tests_\alpha(S)), \forall x\in S.\ \psi \Rightarrow \varphi_\alpha(x) \}
\end{align*}
The set $\next^*(\alpha, \psi, S)$ contains all of the nodes in $\alpha$ that are still able to be scheduled, given that $S$ has already been processed and $\psi$ is the path condition. The set $\tests_\alpha(S)$ contains all the test nodes from $S$. The set $\forms(X)$ contains all the formulae over the variables $X$, which are conjunctions of literals, where a literal is either a variable $x$ or its negation $\lnot x$. Finally, $\valid_\alpha(S)$ contains all the formulae over $\tests_\alpha(S)$ that imply the formulae of all $x\in S$.
As an abuse of notation, we will write $\psi \in S$ to mean that there exists a $\psi' \in S$ such that $\psi\Leftrightarrow \psi'$.

\begin{lemma}\label{lem:nextstar-empty}
  If $\alpha$ is binary branching, $S \subseteq \nBot_\alpha$, $\next^*(\alpha,\psi,S) = \emptyset$ and $\psi \in \valid_\alpha(S)$, then $\psi \in \br_\alpha$ and $\next(\alpha\fatsemi_f\beta) = \min_{f(\psi)}$.
\end{lemma}
\begin{proof}
Since $\next^*(\alpha,\psi,S) = \emptyset$, we know that for all $x \in N_\alpha$, either $x\in S$ or $\psi \land\varphi_\alpha(x)$ is unsatisfiable. In particular, since $S \subseteq \nBot_\alpha$, then $\psi \land \varphi_\alpha(x)$ is unsatisfiable for all $x\in\Bot_\alpha$. We will now establish that $\psi\in\br_\alpha$ by proving the following claims:
\begin{enumerate}

\item $\psi \Rightarrow \lnot\stuck_\alpha \Leftrightarrow \lnot \bigvee_{x\in \Bot_\alpha} \varphi_\alpha(x) \Leftrightarrow \bigwedge_{x\in\Bot_\alpha} \lnot \varphi_\alpha(x)$. It will suffice to show that $\psi \Rightarrow \lnot \varphi_\alpha(x)$ for all $x \in\Bot_\alpha$. Take any $v\vDash\psi$. We know that $\psi\land\varphi_\alpha(x)$ is unsatisfiable since $x\in\Bot_\alpha$, therefore $v\vDash \lnot\varphi_\alpha(x)$, so we are done.

\item $\psi \Leftrightarrow \bigwedge_{x\in S} \varphi_\alpha(x)$. From $\psi \in \valid_\alpha(S)$, we immediately get that $\psi \Rightarrow \bigwedge_{x\in S} \varphi_\alpha(x)$. We now also show the reverse. Since $\psi \in \valid_\alpha(S)$, then $\psi \in \forms(\tests_\alpha(S))$. We now show that for every $y \in \tests_\alpha(S)$, either $\bigwedge_{x\in S} \varphi_\alpha(x) \Rightarrow y$ or $\bigwedge_{x\in S} \varphi_\alpha(x) \Rightarrow \lnot y$. Take any $y\in\tests_\alpha(S)$. By the binary branching property, $\succ_\alpha(y) = \{ z_1, z_2\}$ such that $\varphi_\alpha(z_1) \Leftrightarrow \varphi_\alpha(y) \land y$ and $\varphi_\alpha(z_2) \Leftrightarrow \varphi_\alpha(y) \land \lnot y$. We also know that since $z_1,z_2\in N_\alpha$, then for each $i\in\{1,2\}$ either $z_i \in S$ or $\psi\land\varphi_\alpha(z_i)$ is unsatisfiable. We already know that $\psi \Rightarrow \varphi_\alpha(y)$ since $y\in S$. We also know that $\psi \Rightarrow y$ or $\psi \Rightarrow\lnot y$ since $y\in\tests_\alpha(S)$ and $\psi \in\valid_\alpha(S)$. Without loss of generality, suppose that $\psi \Rightarrow y$, then clearly $\psi\land \varphi_\alpha(z_1) \Leftrightarrow \psi \land \varphi_\alpha(y) \land y$ is satisfiable, therefore $z_1 \in S$. This clearly means that $\bigwedge_{x\in S} \varphi_\alpha(x) \Rightarrow \varphi_\alpha(z_1) \Rightarrow y$. If instead we had that $\psi \Rightarrow \lnot y$, then we would have $\bigwedge_{x\in S} \varphi_\alpha(x) \Rightarrow \lnot y$.

Therefore we have shown that $\bigwedge_{x\in S} \varphi_\alpha(x)$ assigns a truth value to each $y\in \free(\psi)$, and so $\bigwedge_{x\in S} \varphi_\alpha(x) \Rightarrow \psi$.

\item We now show that $S\subseteq\extens_\alpha$, that is $\varphi_\alpha(x) \not\Rightarrow \stuck_\alpha$ for all $x\in S$.
For the sake of contradiction, suppose that $\varphi_\alpha(x) \Rightarrow \stuck_\alpha$ for some $x\in S$. Since $\psi\in\valid_\alpha(S)$, we also have that $\psi \Rightarrow \varphi_\alpha(x)$, which gives us:
\[
  \psi
  \Rightarrow \varphi_\alpha(x)
  \Rightarrow \stuck_\alpha
  \Leftrightarrow \bigvee_{y\in\Bot_\alpha} \varphi_\alpha(y)
\]
But this contradicts the fact that $\psi\land\varphi_\alpha(y)$ is unsatisfiable for all $y\in\Bot_\alpha$, therefore it cannot be true that $\varphi_\alpha(x) \Rightarrow \stuck_\alpha$, and instead we have $\varphi_\alpha(x) \not\Rightarrow \stuck_\alpha$.

\item $\psi \land \varphi_\alpha(x)$ is unsatisfiable for any $x \in \extens_\alpha \setminus S$. This is immediate, since $\psi \land \varphi_\alpha(x)$ is unsatisfiable for any $x \notin S$.

\end{enumerate}
So, we have now shown that $\psi \in \br_\alpha$. Next, we show that $\next(\alpha\fatsemi_f\beta, \psi, S) = \min_{f(\psi)}$ by showing the set inclusion in both directions.

Take any $x \in \next(\alpha\fatsemi_f\beta, \psi, S)$, so $x\in N_{\alpha\fatsemi_f\beta}$, $x \notin S$, $\predplus{x}_{\alpha\fatsemi_f\beta} \subseteq S$, and $\psi\Rightarrow \varphi_{\alpha\fatsemi_f\beta}(x)$. Clearly $x \notin N_\alpha$, otherwise one of $x\notin S$ or $\psi\Rightarrow \varphi_{\alpha\fatsemi_f\beta}(x) = \varphi_\alpha(x)$ would be violated. Now, take any $y <_{\alpha\fatsemi_f\beta} x$, we know from $\predplus{x}_{\alpha\fatsemi_f\beta} \subseteq S$ that $y \in S \subseteq N_\alpha$. So, since $\psi \Rightarrow \varphi_{\alpha\fatsemi_f\beta}(x)$, then $x\in N_{f(\psi)}$. Furthermore, $\{x\} = \min_{f(\psi)}$  since we already established that all predecessors of $x$ are in $\alpha$, and so $x$ has no predecessors in $\beta_{f(\psi)}$.

We now show the reverse direction. Take any $\{x\} = \min_{f(\psi)}$. Clearly $x \in N_{\alpha\fatsemi_f\beta}$ and $x\notin S$. By construction, $\varphi_{\alpha\fatsemi_f\beta}(x) = \varphi_{f(\psi)}(x) \land \psi \Leftrightarrow \tru \land\psi$ ($\varphi_{f(\psi)}(x) \Leftrightarrow \tru$ since $x$ is minimal and therefore no branches have occurred), so clearly $\psi \Rightarrow \varphi_{\alpha\fatsemi_f\beta}(x)$. Now take any $y <_{\alpha\fatsemi_f\beta} x$, since $\{x\} = \min_{f(\psi)}$, then $y \in N_\alpha$, so $\psi \Rightarrow \varphi_\alpha(y)$. We already know that for all $y\in N_\alpha$ either $y \in S$, or $\psi\land \varphi_\alpha(y)$ is unsatisfiable, but clearly the latter case is impossible, so $y\in S$. Therefore $\predplus{x}_{\alpha\fatsemi_f\beta} \subseteq S$, and we have shown all the conditions to guarantee that $x \in \next(\alpha\fatsemi_f\beta, \psi, S)$. 
\end{proof}

\begin{lemma}\label{lem:nextstar-nonempty}
  If $\alpha$ is binary branching, $S\subseteq \nBot_\alpha$, $\next^\ast(\alpha, \psi, S) \neq\emptyset$, and $\psi \in \valid_\alpha(S)$, then $\next(\alpha\fatsemi_f\beta,\psi, S) = \next(\alpha, \psi, S) \neq\emptyset$
\end{lemma}
\begin{proof}
  Since $\next^\ast(\alpha, \psi, S) \neq\emptyset$, then there exists $x\in \next^\ast(\alpha, \psi, S)$, meaning that $x \in N_\alpha \setminus S$ and $\psi\land\varphi_\alpha(x)$ is satisfiable. Take a minimal such $x$, so that there is no $y \in \next^\ast(\alpha, \psi, S)$ such that $y <_\alpha x$.
  
  We now show that $\predplus{x}_\alpha \subseteq S$. Take any $y <_\alpha x$. We know that $y \notin \next^*(\alpha,\psi,S)$, so either $y \in S$ or $\psi \land\varphi_\alpha(y)$ is unsatisfiable. However, since $y <_\alpha x$, then $\varphi_\alpha(x) \Rightarrow \varphi_\alpha(y)$, so $\psi\land\varphi_\alpha(y)$ must be satisfiable, meaning that the only option is that $y\in S$.
  
  Now, since $\alpha$ is binary branching, there are two options. The first option is that $x \in \succ_\alpha(y)$ where $y$ is a test node and $\varphi_\alpha(x) \Leftrightarrow \varphi_\alpha(y) \land y$ or $\varphi_\alpha(x) \Leftrightarrow \varphi_\alpha(y) \land \lnot y$. Since $\psi \in \valid_\alpha(S)$ and $\psi\land\varphi_\alpha(x)$ is satisfiable, then it must be that $\psi \Rightarrow \varphi_\alpha(x)$. If instead $x$ is not the successor of any test nodes, then $\varphi_\alpha(x) \Leftrightarrow \bigwedge_{y\in\pred_\alpha(x)} \varphi_\alpha(y)$, so again since $\psi\land\varphi_\alpha(x)$ is satisfiable, then it must be that $\psi \Rightarrow \varphi_\alpha(x)$. In both cases, we have that $x \in \next(\alpha, \psi, S)$.
  
  Finally, we show that $\next(\alpha, \psi, S) = \next(\alpha\fatsemi_f\beta,\psi, S)$ by showing the set inclusion in both directions. Take any $y \in \next(\alpha, \psi, S)$, so $y \in N_\alpha \setminus S$ such that $\predplus{y}_\alpha \subseteq S$ and $\psi \Rightarrow \varphi_\alpha(y)$. Since $N_\alpha \dwclosed N_{\alpha\fatsemi_f\beta}$, and $\varphi_\alpha(y) = \varphi_{\alpha\fatsemi_f\beta}(y)$, this immediately gives us that $y \in N_{\alpha\fatsemi_f\beta} \setminus S$, $\predplus{y}_{\alpha\fatsemi_f\beta} \subseteq S$, and $\psi \Rightarrow \varphi_{\alpha\fatsemi_f\beta}(x)$ therefore $y \in \next(\alpha\fatsemi_f\beta,\psi, S)$.
  
   Now, for the reverse inclusion, take any $y \in \next(\alpha\fatsemi_f\beta,\psi, S)$, so $y \in N_{\alpha\fatsemi_f\beta} \setminus S$ such that $\predplus{y}_{\alpha\fatsemi_f\beta} \subseteq S$ and $\psi \Rightarrow \varphi_{\alpha\fatsemi_f\beta}(x)$. From previously, we know that $x \in N_\alpha \setminus S$ and $\psi \land\varphi_\alpha(x)$ is satisfiable. Now, suppose for the sake of contradiction that $y \notin N_\alpha$, so it must be that $y \in N_{f(\psi')}$ for some $\psi' \in \br_\alpha$. This implies that $x <_{\alpha\fatsemi_f\beta} y$, as we will now show. First, $\psi \Rightarrow \varphi_{\alpha\fatsemi_f\beta}(y) = \varphi_{f(\psi')}(y) \land \psi '\Rightarrow \psi'$, so $\psi' \land \varphi_\alpha(x)$ is satisfiable. If $x \notin \extens_\alpha$, then $\varphi_\alpha(x) \Rightarrow\stuck_\alpha$, but we know that $\psi' \Rightarrow \lnot\stuck_\alpha$, which contradicts that $\psi'\land\varphi_\alpha(x)$ is satisfiable. Therefore, $x \in \extens_\alpha$, and since $\psi'$ is maximal, then it must be the case that $\psi' \Rightarrow \varphi_\alpha(x)$ already, which implies that $x <_{\alpha\fatsemi_f\beta} y$. However, since $x \notin S$, then $\predplus{y}_{\alpha\fatsemi_f\beta} \not\subseteq S$, which is a contradiction. Therefore, it must be the case that $y\in N_\alpha$. We therefore have that $y \in N_\alpha\setminus S$, $\predplus{y}_\alpha = \predplus{y}_{\alpha\fatsemi_f\beta} \subseteq S$, and $\psi\Rightarrow \varphi_{\alpha\fatsemi_f\beta}(y) = \varphi_\alpha(y)$; therefore, $y \in \next(\alpha, \psi, S)$.
\end{proof}

\begin{lemma}\label{lem:lin-fin-seq}
For any binary branching $\alpha,\beta\in\lpo_\fin(L)$, $S \subseteq \nBot_\alpha$, and $\psi\in \valid_\alpha(S)$:
\[
  \linlpo(\alpha \fatsemi \beta, \psi, S) = \linlpo(\beta, \tru, \emptyset)^\dagger \circ \linlpo(\alpha, \psi, S)
\]
\end{lemma}
\begin{proof}
The proof is by induction on the size of the set $\next^\ast(\alpha, \psi, S)$.
\medskip

In the base case, $\next^\ast(\alpha, \psi, S)$ is empty, so by \Cref{lem:nextstar-empty}, we know that $\psi \in \br_\alpha$ and $\next(\alpha\fatsemi_f\beta, \psi, S) = \min_{f(\psi)}$. Since each lpo is single-rooted, then $\min_{f(\psi)} = \{x\}$ for some $x \in N_{f(\psi)}$. This gives us:
\begin{align*}
  &\linlpo(\alpha\fatsemi\beta, \psi,S)(s)
  \\
  &= \bignd_{x\in\mathsf{next}(\alpha\fatsemi\beta, \psi,S)} \left\{
      \begin{array}{ll}
        \linlpo(\alpha\fatsemi\beta,\psi, S \cup\{ x \})^\dagger(\de{\lambda_{\alpha\fatsemi\beta}(x)}_\act(s)) & \text{if} ~ \lambda_{\alpha\fatsemi\beta}(x) \in \act
        \\
        \linlpo(\alpha\fatsemi\beta,\psi \land \sem{x = \de{\lambda_{\alpha\fatsemi\beta}(x)}(s)}, S \cup\{x\})(s) & \text{if}~ \lambda_{\alpha\fatsemi\beta}(x) \in \test
        \\
        \bot & \text{if}~ \lambda_{\alpha\fatsemi\beta}(x) = \bot
        \\
        \linlpo(\alpha\fatsemi\beta,\psi, S\cup\{x\})(s) & \text{if} ~ \lambda_{\alpha\fatsemi\beta}(x) = \fork
      \end{array}
    \right.
    \\
  &= \left\{
      \begin{array}{ll}
        \linlpo(\alpha\fatsemi\beta,\psi, S \cup\{ x \})^\dagger(\de{\lambda_{\alpha\fatsemi\beta}(x)}_\act(s)) & \text{if} ~ \lambda_{\alpha\fatsemi\beta}(x) \in \act
        \\
        \linlpo(\alpha\fatsemi\beta,\psi \land \sem{x = \de{\lambda_{\alpha\fatsemi\beta}(x)}(s)}, S \cup\{x\})(s) & \text{if}~ \lambda_{\alpha\fatsemi\beta}(x) \in \test
        \\
        \bot & \text{if}~ \lambda_{\alpha\fatsemi\beta}(x) = \bot
        \\
        \linlpo(\alpha\fatsemi\beta,\psi, S\cup\{x\})(s) & \text{if} ~ \lambda_{\alpha\fatsemi\beta}(x) = \fork
      \end{array}
    \right.
    \\
  &= \left\{
      \begin{array}{ll}
        \linlpo(\beta_\psi,\tru, \{ x \})^\dagger(\de{\lambda_{\beta_\psi}(x)}_\act(s)) & \text{if} ~ \lambda_{\beta_\psi}(x) \in \act
        \\
        \linlpo(\beta_\psi,\sem{x = \de{\lambda_{\beta_\psi}(x)}(s)}, \{x\})(s) & \text{if}~ \lambda_{\beta_\psi}(x) \in \test
        \\
        \bot & \text{if}~ \lambda_{\beta_\psi}(x) = \bot
        \\
        \linlpo(\beta_\psi,\tru, \{x\})(s) & \text{if} ~ \lambda_{\beta_\psi}(x) = \fork
      \end{array}
    \right.
    \\
  &= \bignd_{x\in\next(\beta, \tru,\emptyset)} \left\{
      \begin{array}{ll}
        \linlpo(\beta,\tru, \{ x \})^\dagger(\de{\lambda_\beta(x)}_\act(s)) & \text{if} ~ \lambda_{\beta}(x) \in \act
        \\
        \linlpo(\beta,\sem{x = \de{\lambda_\beta(x)}(s)}, \{x\})(s) & \text{if}~ \lambda_{\beta}(x) \in \test
        \\
        \bot & \text{if}~ \lambda_{\beta}(x) = \bot
        \\
        \linlpo(\beta,\tru, \{x\})(s) & \text{if} ~ \lambda_{\beta}(x) = \fork
      \end{array}
    \right.
  \\
  &= \linlpo(\beta,\tru,\emptyset)(s)   
  \\
  &= \linlpo(\beta,\tru,\emptyset)^\dagger(\eta(s))
  \\
  &= \linlpo(\beta,\tru,\emptyset)^\dagger(\linlpo(\alpha, \psi, S)(s))
\end{align*}
\medskip

Now suppose $\next^\ast(\alpha, \psi, S) \neq \emptyset$. By \Cref{lem:nextstar-nonempty}, we know that $\next(\alpha\fatsemi_f\beta, \psi, S)=\next(\alpha, \psi, S)\neq\emptyset$. Given this, we have:
\begin{align*}
  &\linlpo(\alpha\fatsemi\beta, \psi,S)(s)
  \\
  &= \bignd_{x\in\mathsf{next}(\alpha\fatsemi\beta, \psi,S)} \left\{
      \begin{array}{ll}
        \linlpo(\alpha\fatsemi\beta,\psi, S \cup\{ x \})^\dagger(\de{\lambda_{\alpha\fatsemi\beta}(x)}_\act(s)) & \text{if} ~ \lambda_{\alpha\fatsemi\beta}(x) \in \act
        \\
        \linlpo(\alpha\fatsemi\beta,\psi \land \sem{x = \de{\lambda_{\alpha\fatsemi\beta}(x)}(s)}, S \cup\{x\})(s) & \text{if}~ \lambda_{\alpha\fatsemi\beta}(x) \in \test
        \\
        \bot & \text{if}~ \lambda_{\alpha\fatsemi\beta}(x) = \bot
        \\
        \linlpo(\alpha\fatsemi\beta,\psi, S\cup\{x\})(s) & \text{if} ~ \lambda_{\alpha\fatsemi\beta}(x) = \fork
      \end{array}
    \right.
  \\
  &= \bignd_{x\in\mathsf{next}(\alpha, \psi,S)} \left\{
      \begin{array}{ll}
        \linlpo(\alpha\fatsemi\beta,\psi, S \cup\{ x \})^\dagger(\de{\lambda_{\alpha}(x)}_\act(s)) & \text{if} ~ \lambda_{\alpha}(x) \in \act
        \\
        \linlpo(\alpha\fatsemi\beta,\psi \land \sem{x = \de{\lambda_{\alpha}(x)}(s)}, S \cup\{x\})(s) & \text{if}~ \lambda_{\alpha}(x) \in \test
        \\
        \bot & \text{if}~ \lambda_{\alpha}(x) = \bot
        \\
        \linlpo(\alpha\fatsemi\beta,\psi, S\cup\{x\})(s) & \text{if} ~ \lambda_{\alpha}(x) = \fork
      \end{array}
    \right.
    \intertext{Clearly $\next^\ast(\alpha, \psi, S\cup\{x\}) \subset \next^\ast(\alpha, \psi, S)$ and $\psi \in\valid_\alpha(S\cup\{x\})$, if $x$ is not a test node and therefore $\valid_\alpha(S\cup\{x\}) = \valid_\alpha(S)$. So, by the induction hypothesis and monad laws, we get:}
  &= \bignd_{x\in\mathsf{next}(\alpha, \psi,S)} \left\{
      \begin{array}{ll}
        \linlpo(\beta, \tru, \emptyset)^\dagger(\linlpo(\alpha,\psi, S \cup\{ x \})^\dagger(\de{\lambda_{\alpha}(x)}_\act(s))) & \text{if} ~ \lambda_{\alpha}(x) \in \act
        \\
        \linlpo(\alpha\fatsemi\beta,\psi \land \sem{x = \de{\lambda_{\alpha}(x)}(s)}, S \cup\{x\})(s) & \text{if}~ \lambda_{\alpha}(x) \in \test 
        \\
        \bot & \text{if}~ \lambda_{\alpha}(x) = \bot
        \\
        \linlpo(\beta, \tru, \emptyset)^\dagger(\linlpo(\alpha,\psi, S\cup\{x\})(s)) & \text{if} ~ \lambda_{\alpha}(x) = \fork
      \end{array}
    \right.
    \intertext{We also have that $\next^\ast(\alpha, S\cup\{x\}, \psi\land \sem{x = \de{\lambda_{\alpha}(x)}_\test(s)}) \subset \next^\ast(\alpha, S,\psi)$. In addition, $\psi\land \sem{x = \de{\lambda_{\alpha}(x)}_\test(s)} \in \valid_\alpha(S\cup\{x\})$. So, by the induction hypothesis:
}
  &= \bignd_{x\in\mathsf{next}(\alpha, \psi,S)} \left\{
      \begin{array}{ll}
        \linlpo(\beta, \tru, \emptyset)^\dagger(\linlpo(\alpha,\psi, S \cup\{ x \})^\dagger(\de{\lambda_{\alpha}(x)}_\act(s))) & \text{if} ~ \lambda_{\alpha}(x) \in \act
        \\
        \linlpo(\beta, \tru, \emptyset)^\dagger(\linlpo(\alpha,\psi \land \sem{x = \de{\lambda_{\alpha}(x)}(s)}, S \cup\{x\})(s)) & \text{if}~ \lambda_{\alpha}(x) \in \test
        \\
        \bot & \text{if}~ \lambda_{\alpha}(x) = \bot
        \\
        \linlpo(\beta, \tru, \emptyset)^\dagger(\linlpo(\alpha,\psi, S\cup\{x\})(s)) & \text{if} ~ \lambda_{\alpha}(x) = \fork
      \end{array}
    \right.
    \intertext{Finally since $f^\dagger(\bot) = \bot$ for any $f$, we get:}
  &= \bignd_{x\in\mathsf{next}(\alpha, \psi,S)} \left\{
      \begin{array}{ll}
        \linlpo(\beta, \tru, \emptyset)^\dagger(\linlpo(\alpha,\psi, S \cup\{ x \})^\dagger(\de{\lambda_{\alpha}(x)}_\act(s))) & \text{if} ~ \lambda_{\alpha}(x) \in \act
        \\
        \linlpo(\beta, \tru, \emptyset)^\dagger(\linlpo(\alpha,\psi \land \sem{x = \de{\lambda_{\alpha}(x)}(s)}, S \cup\{x\})(s)) & \text{if}~ \lambda_{\alpha}(x) \in \test
        \\
        \linlpo(\beta, \tru, \emptyset)^\dagger(\bot) & \text{if}~ \lambda_{\alpha}(x) = \bot
        \\
        \linlpo(\beta, \tru, \emptyset)^\dagger(\linlpo(\alpha,\psi, S\cup\{x\})(s)) & \text{if} ~ \lambda_{\alpha}(x) = \fork
      \end{array}
    \right.
    \intertext{Since $f^\dagger(X \nd Y) = f^\dagger(X) \nd f^\dagger(Y)$:}
  &= \linlpo(\beta, \tru, \emptyset)^\dagger\left(\bignd_{x\in\mathsf{next}(\alpha, \psi,S)} \left\{
      \begin{array}{ll}
        \linlpo(\alpha,\psi, S \cup\{ x \})^\dagger(\de{\lambda_{\alpha}(x)}_\act(s)) & \text{if} ~ \lambda_{\alpha}(x) \in \act
        \\
        \linlpo(\alpha,\psi \land \sem{x = \de{\lambda_{\alpha}(x)}(s)}, S \cup\{x\})(s) & \text{if}~ \lambda_{\alpha}(x) \in \test
        \\
        \bot & \text{if}~ \lambda_{\alpha}(x) = \bot
        \\
        \linlpo(\alpha,\psi, S\cup\{x\})(s) & \text{if} ~ \lambda_{\alpha}(x) = \fork
      \end{array}
    \right.\right)
  \\
  &= \linlpo(\beta, \tru, \emptyset)^\dagger\left( \linlpo(\alpha, \psi, S)(s) \right)
\qedhere
\end{align*}
\end{proof}

\begin{restatable}
{lemma}{linSeq}
\label{lem:lin-seq}
$
  \lin(\de{C_1; C_2}) = \lin(\de{C_2})^\dagger \circ \lin(\de{C_1})
$.
\end{restatable}
\begin{proof}

\begin{align*}
  \lin(\de{C_1 ; C_2})
  &= \lin(\de{C_1} \fatsemi \de{C_2})
  \intertext{By the definition of $\fatsemi$ for infinite pomsets:}
  &= \lin\left(\sup_{\Alpha \ll \de{C_1}} \sup_{\Beta \ll\de{C_2}} \Alpha\fatsemi\Beta \right)
  \intertext{Since $\lin$ is Scott continuous:}
  &= \sup_{\Alpha \ll \de{C_1}} \sup_{\Beta \ll\de{C_2}} \lin(\Alpha\fatsemi\Beta) 
  \\
  &= \sup_{\Alpha \ll \de{C_1}} \sup_{\Beta \ll\de{C_2}} \sup_{\GGamma \ll \Alpha\fatsemi\Beta} \linfin(\GGamma)
  \intertext{Since $\Alpha\fatsemi\Beta$ is finite, then $\Alpha \fatsemi\Beta \ll \Alpha\fatsemi\Beta$ and therefore since $\linfin$ is monotone, then $\sup_{\GGamma \ll \Alpha\fatsemi\Beta} \linfin(\GGamma) = \linfin(\Alpha\fatsemi\Beta)$.}
  &= \sup_{\Alpha \ll \de{C_1}} \sup_{\Beta \ll\de{C_2}} \linfin(\Alpha\fatsemi\Beta)
  \\
  &= \sup_{\Alpha \ll \de{C_1}} \sup_{\Beta \ll\de{C_2}} \linfin( \{ \alpha \fatsemi_f \beta \mid \alpha \in \Alpha, \beta\in\Beta, f\in\copylpo_{\alpha,\beta} \})
  \intertext{The set above is an equivalence class, so fixing an arbitrary $\alpha \in \Alpha$, $\beta\in\Beta$, and $f\in\copylpo_{\alpha,\beta}$, we can rewrite the expression as follows:}
  &= \sup_{\Alpha \ll \de{C_1}} \sup_{\Beta \ll\de{C_2}} \linfin([ \alpha \fatsemi_f \beta])
  \\
  &= \sup_{\Alpha \ll \de{C_1}} \sup_{\Beta \ll\de{C_2}} \linlpo(\alpha \fatsemi_f \beta, \tru, \emptyset)
  \intertext{By \Cref{lem:lin-fin-seq}:}
  &= \sup_{\Alpha \ll \de{C_1}} \sup_{\Beta \ll\de{C_2}} \linlpo(\beta, \tru, \emptyset)^\dagger \circ \linlpo(\alpha, \tru, \emptyset)
  \\
  &= \sup_{\Alpha \ll \de{C_1}} \sup_{\Beta \ll\de{C_2}} \linfin([\beta])^\dagger \circ \linfin([\alpha])
  \intertext{By continuity of Kleisli extension, and since $\Alpha = [\alpha]$ and $\Beta = [\beta]$:}
  &=  \left(\sup_{\Beta \ll\de{C_2}} \linfin(\Beta)\right)^\dagger \circ \sup_{\Alpha \ll \de{C_1}}\linfin(\Alpha)
  \\
  &= \lin(\de{C_2})^\dagger \circ \lin(\de{C_1})
\qedhere
\end{align*}
\end{proof}

\subsection{Properties of Linearization: Guarded Branching}

\begin{lemma}\label{lem:lin-lpo-guard}
For any $\alpha,\beta\in\lpo_\fin(\act)$, $b\in \test$, and $x \in\Nodes \setminus (N_\alpha \cup N_\beta)$:
\[
  \linlpo(\guard(x, b, \alpha, \beta), \tru, \emptyset)(s) = \left\{
    \begin{array}{ll}
      \linlpo(\alpha, \tru, \emptyset)(s) & \text{if}\ \de{b}_\test(s) = 1
      \\
      \linlpo(\beta, \tru, \emptyset)(s) & \text{if}\ \de{b}_\test(s) = 0
    \end{array}
  \right.
\]
\end{lemma}
\begin{proof}
\begin{align*}
  &\linlpo(\guard(x, b, \alpha, \beta), \tru, \emptyset)(s)
  \\
  &= {\bignd_{x\in\mathsf{next}(\alpha, \psi,S)}} \left\{
      \begin{array}{ll}
        \linlpo(\guard(x, b, \alpha, \beta),\tru, \{ x \})^\dagger(\de{\lambda_{\guard(x, b, \alpha, \beta)}(x)}_\act(s)) & \text{if} ~ \lambda_{\guard(x, b, \alpha, \beta)}(x) \in \act
        \\
        \linlpo(\guard(x, b, \alpha, \beta),\sem{x = \de{\lambda_{\guard(x, b, \alpha, \beta)}(x)}(s)}, \{x\})(s) & \text{if}~ \lambda_{\guard(x, b, \alpha, \beta)}(x) \in \test
        \\
        \bot & \text{if}~ \lambda_{\guard(x, b, \alpha, \beta)}(x) = \bot
        \\
        \linlpo(\guard(x, b, \alpha, \beta),\tru, \{x\})(s) & \text{if} ~ \lambda_{\guard(x, b, \alpha, \beta)}(x) = \fork
      \end{array}
    \right.
  \intertext{Since $x$ is the root of $\guard(x, b, \alpha, \beta)$, then $\mathsf{next}(\alpha, \psi,S) = \{x\}$ and $\lambda_{\guard(x, b, \alpha, \beta)}(x) \in \test$.}
  &= \linlpo(\guard(x, b, \alpha, \beta),\sem{x = \de{\lambda_{\guard(x, b, \alpha, \beta)}(x)}(s)}, \{x\})(s)
  \\
  &= \left\{
    \begin{array}{ll}
     \linlpo(\guard(x, b, \alpha, \beta),x, \{x\})(s) & \text{if}\ \de{b}_\test(s) = 1
      \\
      \linlpo(\guard(x, b, \alpha, \beta),\lnot x, \{x\})(s) & \text{if}\ \de{b}_\test(s) = 0
    \end{array}
  \right.
  \intertext{The formula $x$ prevents all nodes in $\beta$ from being scheduled and similarly $\lnot x$ prevents the $\alpha$ nodes from being scheduled, so we can remove the $\guard$ call to obtain the following:}
  &= \left\{
    \begin{array}{ll}
      \linlpo(\alpha, \tru, \emptyset)(s) & \text{if}\ \de{b}_\test(s) = \tru
      \\
      \linlpo(\beta, \tru, \emptyset)(s) & \text{if}\ \de{b}_\test(s) = \fls
    \end{array}
  \right.
\qedhere
\end{align*}
\end{proof}

\begin{restatable}{lemma}{linIf}\label{lem:lin-if}
\textnormal{
\begin{align*}
  \lin(\de{\iftf b{C_1}{C_2}})(s)
  &= \left\{
    \begin{array}{ll}
      \lin(\de{C_1})(s) & \text{if}~ \de{b}_\test(s) = 1
      \\
      \lin(\de{C_2})(s) & \text{if}~ \de{b}_\test(s) = 0
    \end{array}
  \right.
\end{align*}}
\end{restatable}
\begin{proof}
\begin{align*}
  &\lin{\de{\iftf b{C_1}{C_2}}}(s)
  \\
  &= \lin(\guard(b, \de{C_1}, \de{C_2}))(s)
  \\
  &= \sup_{\Alpha\ll\de{C_1}} \sup_{\Beta\ll\de{C_2}} \linfin(\guard(b, \Alpha, \Beta))(s)
  \\
  &= \sup_{\Alpha\ll\de{C_1}} \sup_{\Beta\ll\de{C_2}} \linfin(\{ \guard(x, b, \alpha, \beta) \mid x\in\Nodes, \alpha\in\Alpha, \beta\in\Beta \})(s)
  \\
  &= \sup_{\Alpha\ll\de{C_1}} \sup_{\Beta\ll\de{C_2}} \linfin([\guard(x, b, \alpha, \beta)])(s)
  \\
  &= \sup_{\Alpha\ll\de{C_1}} \sup_{\Beta\ll\de{C_2}} \linlpo(\guard(x, b, \alpha, \beta), \tru, \emptyset)(s)
  \\
  &= \sup_{\Alpha\ll\de{C_1}} \sup_{\Beta\ll\de{C_2}} \left\{
    \begin{array}{ll}
      \linlpo(\alpha, \tru, \emptyset)(s) & \text{if}\ \de{b}_\test(s) = 1
      \\
      \linlpo(\beta, \tru, \emptyset)(s) & \text{if}\ \de{b}_\test(s) = 0
    \end{array}
  \right.
  \qquad   \text{(By \Cref{lem:lin-lpo-guard})}
  \\
  &= \left\{
    \begin{array}{ll}
      \sup_{\Alpha\ll\de{C_1}}\linfin(\Alpha)(s) & \text{if}\ \de{b}_\test(s) = 1
      \\
      \sup_{\Beta\ll\de{C_2}}\linfin(\Beta)(s) & \text{if}\ \de{b}_\test(s) = 0
    \end{array}
  \right.
  \\
  &= \left\{
    \begin{array}{ll}
      \lin(\de{C_1})(s) & \text{if}\ \de{b}_\test(s) = 1
      \\
      \lin(\de{C_2})(s) & \text{if}\ \de{b}_\test(s) = 0
    \end{array}
  \right.
\qedhere
\end{align*}
\end{proof}

\subsection{Properties of Linearization: While Loops}

\begin{lemma}\label{lem:lin-while-iter}
Let $D(S)$ be the dcpo such that $\lin \colon \pom \to S \to D(S)$, and let $\bot_\pom \in \pom$ be the bottom of the pomset order, $\bot_D \in D(S)$ be the bottom of $D(S)$, and $\bot_D^\bullet \colon S\to D(S)$ be the bottom of the pointwise order on $D(S)$. Then
\[
    \lin\left(\Phi^n_{\tuple{C,b}}(\bot_\pom)\right) = \Psi^n_{\tuple{\lin(\de{C}), b}}(\bot_D^\bullet)
  \]
\end{lemma}
\begin{proof}
  The proof is by induction on $n$. If $n=0$, then we have:
  \[
    \lin\left(\Phi^0_{\tuple{C,b}}(\bot_\pom)\right)(s)
    = \lin( \bot_\pom )(s)
    = \bot_D
    = \bot_D^\bullet(s)
    = \Psi^0_{\tuple{\lin(\de{C}), b}}(\bot_D^\bullet)(s)
  \]
Now, suppose the claim holds for $n$, then we have:
\begin{align*}
  \lin\left(\Phi^{n+1}_{\tuple{C,b}}(\bot_\pom)\right)(s)
  &= \lin\left(\Phi_{\tuple{C,b}}\left(\Phi^{n}_{\tuple{C,b}}(\bot_\pom) \right)\right)(s)
  \\
  &= \lin\left( \guard\left( b, \de{C} \fatsemi \Phi^{n}_{\tuple{C,b}}(\bot_\pom), \de{\skp} \right) \right)(s)
  \intertext{By \Cref{lem:lin-lpo-guard}.}
  &= \left\{
    \begin{array}{ll}
      \lin\left( \de{C} \fatsemi \Phi^{n}_{\tuple{C,b}}(\bot_\pom) \right)(s) & \text{if}~ \de{b}_\test(s) = 1
      \\
      \lin(\de{\skp})(s) & \text{if}~ \de{b}_\test(s) = 0
    \end{array}
  \right.
  \intertext{By \Cref{lem:lin-skip,lem:lin-seq}.}
  &= \left\{
    \begin{array}{ll}
      \lin\left(\Phi^{n}_{\tuple{C,b}}(\bot_\pom)\right)^\dagger(\lin(\de{C})(s)) & \text{if}~ \de{b}_\test(s) = 1
      \\
      \eta(s) & \text{if}~ \de{b}_\test(s) = 0
    \end{array}
  \right.
  \intertext{By the induction hypothesis.}
  &= \left\{
    \begin{array}{ll}
      \left(\Psi^{n}_{\tuple{\lin{\de{C}},b}}(\bot_D^\bullet)\right)^\dagger(\lin(\de{C})(s)) & \text{if}~ \de{b}_\test(s) = 1
      \\
      \eta(s) & \text{if}~ \de{b}_\test(s) = 0
    \end{array}
  \right.
  \\
  &= \Psi_{\tuple{\lin{\de{C}},b}} \left( \Psi^{n}_{\tuple{\lin{\de{C}},b}}(\bot_D^\bullet)\right)(s)
  \\
  &= \Psi^{n+1}_{\tuple{\lin{\de{C}},b}}(\bot_D^\bullet)(s)
\qedhere
\end{align*}
\end{proof}

\begin{restatable}{lemma}{linWhile}\label{lem:lin-while}
\textnormal{
$\lin(\de{\whl bC}) = \mathsf{lfp}\left( \Psi_{\tuple{\lin(\de{C}),b}} \right)$}
where
\[
  \Psi_{\tuple{f, b}}(g)(s) \triangleq \left\{
    \begin{array}{ll}
      g^\dagger(f(s)) & \text{if}~ \de{b}_\test(\sigma) = 1
      \\
      \eta(s) & \text{if}~ \de{b}_\test(\sigma) = 0
    \end{array}
  \right.
\]
\end{restatable}
\begin{proof}
\begin{align*}
  \lin(\de{\whl bC})
  &= \lin\left( \mathsf{lfp} \left( \Phi_{\tuple{C, b}} \right) \right)
  \\
  &= \lin\left( \sup_{n\in\mathbb N}  \Phi^n_{\tuple{C, b}}(\bot_\pom) \right) && \text{(By the Kleene fixed point theorem)}
  \\
  &= \sup_{n\in\mathbb N}  \lin\left(\Phi^n_{\tuple{C, b}}(\bot_\pom) \right) && \text{(By Scott continuity of $\lin$)}
  \\
  &= \sup_{n\in\mathbb N}  \Psi^n_{\tuple{\lin(\de{C}), b}}(\bot_D^\bullet) && \text{(By \Cref{lem:lin-while-iter})}
  \\
  &= \mathsf{lfp}\left(\Psi_{\tuple{\lin(\de{C}), b}}\right) && \text{(By the Kleene fixed point theorem)}
\qedhere
\end{align*}
\end{proof}

\subsection{Main Lemma}

\linProps*
\begin{proof}
Follows immediately from \Cref{lem:lin-skip,lem:lin-seq,lem:lin-if,lem:lin-while}
\end{proof}

\section{Powerdomains}
\label{app:powdom}

In this section, we fix our domain of computation to be the Hoare powerdomain $\tuple{\mathcal P(S), \subseteq}$. In particular, this means that suprema are given by set union: $\sup_{i\in I} S_i \triangleq \bigcup_{i\in I} S_i$.

We begin by defining standard pomset language semantics in \Cref{fig:pomlang-semantics}.
Although the semantics is defined using pomsets with tests, $\test_\powdom = \emptyset$, so the formula of every node in $\Alpha$ is $\tru$. In addition, all pomsets are finite and $\bot$ nodes are never introduced, meaning that $\Alpha$ consists only of actions and fork nodes. We therefore get that the definition sequential composition is the standard one, \eg \cite{gischer1988equational,pratt1986modeling,meyer1989pomset,bakker1990metric}.

\begin{restatable}{lemma}{seqPowdom}\label{lem:seq-powdom}
If $\Bot_\alpha = \emptyset$, $\varphi_\alpha(x) = \tru$ for all $x\in N_\alpha$, and $f = [\tru \mapsto \beta]$, then:
\[
  N_{\alpha\fatsemi_f\beta} = N_\alpha \cup N_\beta
  \quad\text{and}\quad
  \mathord{<}_{\alpha\fatsemi_f\beta} = \mathord{<_\alpha} \cup \mathord{<_\beta} \cup (N_\alpha \times N_\beta)
\]
\end{restatable}
\begin{proof}
First note that since $\Bot_\alpha = \emptyset$, then $\stuck_\alpha = \fls$ and so $\extens_\alpha = N_\alpha$. Therefore, $\br_\alpha = \{ \bigwedge_{x \in N_\alpha} \varphi_\alpha(x) \} = \{ \tru \}$. We now show the condition on nodes.
\[
  N_{\alpha\fatsemi_f\beta}
  = N_\alpha \cup \bigcup_{\psi\in\br_\alpha} N_{f(\psi)}
  = N_\alpha \cup N_{f(\tru)}
  = N_\alpha \cup N_\beta
\]
We now also show the condition on orders:
\begin{align*}
  \mathord{<_{\alpha\fatsemi_f\beta}}
  &= \mathord{<_{\alpha}} \cup \bigcup_{\psi\in\br_\alpha} \left( \mathord{<_{f(\psi)}} \cup \{ x \in N_\alpha \mid \psi \Rightarrow \varphi_\alpha(x) \} \times N_{f(\psi)} \right)
  \\
  &= \mathord{<_{\alpha}} \cup \mathord{<_{f(\tru)}} \cup \{ x \in N_\alpha \mid \tru \Rightarrow \tru \} \times N_{f(\tru)}
  \\
  &= \mathord{<_{\alpha}} \cup \mathord{<_{\beta}} \cup (N_\alpha \times N_{\beta})
\qedhere
\end{align*}
\end{proof}

Loops are interpreted as the least fixed point of $\Xi_{\tuple{C,b}}$, which is  Scott continuous in the subset $\subseteq$ order, since suprema (given by $\cup$) are associative and commutative. The $\skp$ appears in the false branch to correspond to $\Phi_{\tuple{C,b}}$.

\subsection{Definition of Translation}

To define $\tr$, we start with a translation on finite \lpof s, which additionally takes a formula $\psi$ to encode the resolution of tests in the current trace.
\[
    \tr_\lpo(\alpha,\psi) \triangleq \tuple{N_\psi, <_\psi, \lambda_\psi, [-\mapsto \tru]}
\]
where $[-\mapsto \tru]$ is the constant $\tru$ function and
\begin{align*}
  N_\psi &\triangleq \{ x \in N_\alpha \mid \psi \Rightarrow \varphi_\alpha(x) \}
  \\
  \mathord{<_\psi} &\triangleq \mathord{<_\alpha} \cap (N_\psi\times N_\psi)
\\
  \lambda_\psi(x) &\triangleq \left\{
    \begin{array}{ll}
      \lambda_\alpha(x) & \text{if}~ \lambda_\alpha(x) \notin \test
      \\
      \assume \lambda_\alpha(x) & \text{if}~ \lambda_\alpha(x) \in \test ~\text{and}~ \psi\Rightarrow x
      \\
      \assume{\lnot \lambda_\alpha(x)} & \text{if}~ \lambda_\alpha(x) \in \test ~\text{and}~ \psi\Rightarrow \lnot x
    \end{array}
  \right.
\end{align*}
The nodes of $\tr_\lpo(\alpha,\psi)$ are those nodes $x\in N_\alpha$ such that $\psi \Rightarrow \varphi_\alpha(x)$, indicating that $x$ agrees with the outcomes of the tests described by $\psi$. Test nodes with label $b$ are converted into either $\assume b$ or $\assume{\lnot b}$, depending on whether that node is true or false in $\psi$.
We define translation on finite pomsets by taking a union over all branches $\psi\in\br_\alpha$.
\[
  \tr_\fin([\alpha]) \triangleq \{ [\tr_\lpo(\alpha, \psi)] \mid \psi\in\br_\alpha \}
\]
Finite translation is monotonic (\Cref{lem:tr-fin-mono})---$\tr_\fin(\Alpha) \subseteq \tr_\fin(\Beta)$ if $\Alpha \lepom \Beta$---therefore we can extend to infinite pomsets with the extension lemma.
\[
  \tr(\Alpha) \triangleq \tr_\fin^*(\Alpha)
\]

\subsection{Properties of the Translation Function}

\begin{lemma}[Monotonicity of $\tr_\lpo$]\label{lem:tr-lpo-mono}
  If $\alpha \lelpo \alpha'$ and $\psi\in\br_\alpha$, then $\tr_\lpo(\alpha, \psi) = \tr_\lpo(\alpha', \psi)$
\end{lemma}
\begin{proof}
We first show that $N_{\tr_\lpo(\alpha, \psi)} = N_{\tr_\lpo(\alpha', \psi)}$, or in other words, that $\{ x \in N_\alpha \mid \psi \Rightarrow \varphi_\alpha(x) \} = \{ x \in N_{\alpha'} \mid \psi \Rightarrow \varphi_{\alpha'}(x) \}$ The forward inclusion is immediate since $N_\alpha \subseteq N_\alpha'$ and $\varphi_\alpha(x) = \varphi_{\alpha'}(x)$ for all $x \in N_\alpha$. We now show the reverse inclusion. Take any $x \in N_{\alpha'}$ such that $\psi\Rightarrow \varphi_{\alpha'}(x)$. If $x \notin N_\alpha$, then by \Cref{lem:lelpo-missing}, $x \in \succplus{\Bot_\alpha}_{\alpha'}$. This means that there is some $y\in \Bot_\alpha$ such that $y <_{\alpha'} x$, and so $\varphi_{\alpha'}(x) \Rightarrow \varphi_{\alpha'}(y) \Rightarrow \stuck_\alpha$. But this contradicts the fact that $\psi \Rightarrow \varphi_{\alpha'}(x)$ and $\psi \in \br_\alpha$. Therefore, it must be that $x \in N_\alpha$, and so the claim holds.

Now, for the order, we have:
\begin{align*}
  \mathord{<_{\tr_\lpo(\alpha, \psi)}}
  &= \mathord{<_\alpha} \cap (N_{\tr_\lpo(\alpha, \psi)} \times N_{\tr_\lpo(\alpha, \psi)})
  \\
  &= (\mathord{<_{\alpha'}} \cap (N_\alpha \times N_\alpha)) \cap (N_{\tr_\lpo(\alpha, \psi)} \times N_{\tr_\lpo(\alpha, \psi)})
  \intertext{Since $N_{\tr_\lpo(\alpha, \psi)} \subseteq N_\alpha$:}
  &= \mathord{<_{\alpha'}} \cap (N_{\tr_\lpo(\alpha, \psi)} \times N_{\tr_\lpo(\alpha, \psi)})
  \\
  &= \mathord{<_{\alpha'}} \cap (N_{\tr_\lpo(\alpha', \psi)} \times N_{\tr_\lpo(\alpha', \psi)})
  \\
  &= \mathord{<_{\tr_\lpo(\alpha', \psi)}}
\end{align*}
Finally, for any $x \in N_{\tr_\lpo(\alpha, \psi)}$, we have:
\begin{align*}
  \lambda_{\tr_\lpo(\alpha, \psi)}(x)
  &= \left\{
    \begin{array}{ll}
      \lambda_\alpha(x) & \text{if}~ \lambda_\alpha(x) \notin \test
      \\
      \assume \lambda_\alpha(x) & \text{if}~ \lambda_\alpha(x) \in \test ~\text{and}~ \psi\Rightarrow x 
      \\
      \assume{\lnot \lambda_\alpha(x)} & \text{if}~ \lambda_\alpha(x) \in \test ~\text{and}~ \psi\Rightarrow \lnot x
    \end{array}
  \right.
  \intertext{By virtue of the fact that $x \in N_{\tr_\lpo(\alpha, \psi)}$, we know that $\psi \Rightarrow \varphi_\alpha(x)$, and since $\psi\in\br_\alpha$, then $\psi \Rightarrow \lnot \stuck_\alpha$, therefore $\lambda_\alpha(x)$ cannot be $\bot$, and since $\lab$ is a flat cpo, then it must be that $\lambda_\alpha(x) = \lambda_{\alpha'}(x)$.}
  &= \left\{
    \begin{array}{ll}
      \lambda_{\alpha'}(x) & \text{if}~ \lambda_{\alpha'}(x) \notin \test
      \\
      \assume \lambda_{\alpha'}(x) & \text{if}~ \lambda_{\alpha'}(x) \in \test ~\text{and}~ \psi\Rightarrow x 
      \\
      \assume{\lnot \lambda_{\alpha'}(x)} & \text{if}~ \lambda_{\alpha'}(x) \in \test ~\text{and}~ \psi\Rightarrow \lnot x 
    \end{array}
  \right.
  \\
  &= \lambda_{\tr_\lpo(\alpha',\psi)}(x)
\end{align*}
And finally, clearly we have $\varphi_{\tr_\lpo(\alpha,\psi)}(x) = \varphi_{\tr_\lpo(\alpha',\psi)}(x) = \tru$.
\end{proof}

\begin{lemma}[Monotonicity of $\tr_\fin$]\label{lem:tr-fin-mono}
If $\Alpha \lepom \Alpha'$, then $\tr_\fin(\Alpha) \subseteq \tr_\fin(\Alpha')$.
\end{lemma}
\begin{proof}
Since $\Alpha \lepom \Alpha'$, then by \Cref{lem:altLePom} there exists an $\alpha\in\Alpha$ and $\alpha'\in\Alpha'$ such that $\alpha\lelpo\alpha'$. We now prove the claim as follows:
\begin{align*}
  \tr_\fin(\Alpha)
  &= \tr_\fin([\alpha])
  \\
  &= \{ [\beta] \mid \psi \in \br_\alpha, \beta\in\tr_\lpo(\alpha) \}
  \\
  &\subseteq \{ [\beta] \mid \psi \in \br_{\alpha'}, \beta\in\tr_\lpo(\alpha) \} && \text{By \Cref{lem:br-mono}.}
  \\
  &= \{ [\beta] \mid \psi \in \br_{\alpha'}, \beta\in\tr_\lpo(\alpha') \} && \text{By \Cref{lem:tr-lpo-mono}.}
  \\
  &= \tr_\fin([\alpha'])
  \\
  &= \tr_\fin(\Alpha')
&&\qedhere
\end{align*}
\end{proof}


\subsection{Translation and Pomset Operations}

\begin{lemma}\label{lem:tr-singleton}
For any $\ell \in \act \cup \{\fork\}$:
\[
   \tr(\singleton{\ell}) = \{ \singleton{\ell} \}
\]
\end{lemma}
\begin{proof}
Since $\singleton{\ell}$ is a finite pomset:
\begin{align*}
  \tr(\singleton{\ell})
  &= \tr_\fin(\singleton{\ell})
  \intertext{Fixing any arbitrary $x\in \Nodes$, we have that $\singleton{\ell} = [\singleton{\ell}_x]$.}
  &= \tr_\fin([\singleton{\ell}_x])
  \\
  &= \{ [\tr_\lpo( \singleton{\ell}_x, \psi ) ] \mid \psi \in \br_{\singleton{\ell}_x} \}
  \intertext{Since $\br_{\singleton{\ell}_x} = \{\tru\}$ and $\tr_\lpo( \singleton{\ell}_x, \tru ) = \singleton{\ell}_x$:}
  &= \{ [ \singleton{\ell}_x ] \}
  \\
  &= \{ \singleton{\ell} \}
\qedhere
\end{align*}
\end{proof}

For sequential composition of \lpof s, where $\br_\alpha= \{\tru\}$, we will write $\alpha\fatsemi\beta$ (omitting the copy function $f$), to mean $\alpha\fatsemi_f\beta$ where $f \triangleq [\tru \mapsto \beta]$, \ie $f$ does not perform any renaming, since only a single isomorphic copy of $\beta$ is needed.

\begin{lemma}\label{lem:tr-lpo-seq}
  For any $\alpha,\beta\in\lpo_\fin(\lab)$, $f\in\copylpo_{\alpha,\beta}$, $\psi \in \br_\alpha$, and $\psi' \in \br_{f(\psi)}$,
  \[
     \tr_\lpo(\alpha\fatsemi_f\beta, \psi\land\psi')
     =
    \tr_\lpo(\alpha, \psi) \fatsemi \tr_\lpo(f(\psi), \psi')
  \]
\end{lemma}
\begin{proof}
We show the equality component-wise. Some of the steps use \Cref{lem:seq-powdom}. First, for the nodes, we have:
\begin{align*}
  N_{\tr_\lpo(\alpha\fatsemi_f\beta, \psi\land\psi')}
  &= \{ x \in N_{\alpha\fatsemi_f\beta} \mid \psi\land\psi' \Rightarrow \varphi_{\alpha\fatsemi_f\beta}(x) \}
  \intertext{Since $\psi\in\br_\alpha$, then we will only take nodes from $f(\psi)$ and no other $\psi''\in \br_\alpha$}
  &= \{ x \in N_{\alpha} \mid \psi\land\psi' \Rightarrow \varphi_{\alpha}(x) \}
    \cup \{ x \in N_{f(\psi)} \mid \psi\land\psi' \Rightarrow \varphi_{f(\psi)}(x) \land \psi \}
  \\
  &= \{ x \in N_{\alpha} \mid \psi \Rightarrow \varphi_{\alpha}(x) \}
    \cup \{ x \in N_{f(\psi)} \mid \psi' \Rightarrow \varphi_{f(\psi)}(x) \}
  \\
  &= N_{\tr_\lpo(\alpha, \psi)} \cup N_{\tr_\lpo(f(\psi), \psi')}
  \\
  &= N_{\tr_\lpo(\alpha, \psi) \fatsemi \tr_\lpo(f(\psi), \psi')}
\end{align*}
Now, we show that the orders are the same:
\begin{align*}
  &\mathord{<_{\tr_\lpo(\alpha\fatsemi_f\beta, \psi\land\psi')}}
  \\
  &= \mathord{<_{\alpha\fatsemi_f\beta}} \cap ( N_{\tr_\lpo(\alpha\fatsemi_f\beta, \psi\land\psi')} \times N_{\tr_\lpo(\alpha\fatsemi_f\beta, \psi\land\psi')} )
  \\
  &= (\mathord{<_{\alpha}} \cap ( N_{\tr_\lpo(\alpha, \psi)} \times N_{\tr_\lpo(\alpha, \psi)} ) \cup
      (\mathord{<_{f(\psi)}} \cap (N_{\tr_\lpo(f(\psi), \psi')} \times N_{\tr_\lpo(f(\psi), \psi')}) \cup
      (N_{\tr_\lpo(\alpha, \psi)} \times N_{\tr_\lpo(f(\psi), \psi')})
  \\
  &= \mathord{<_{\tr_\lpo(\alpha, \psi) \fatsemi \tr_\lpo(f(\psi), \psi')}}
\end{align*}
Now, for any $x$, clearly we have:
\[
  \varphi_{\tr_\lpo(\alpha\fatsemi_f\beta, \psi\land\psi')}(x)
  = \varphi_{\tr_\lpo(\alpha, \psi) \fatsemi \tr_\lpo(f(\psi), \psi')}(x)
  = \tru
\]
If $x \in N_\alpha$, then:
\begin{align*}
  \lambda_{\tr_\lpo(\alpha\fatsemi_f\beta, \psi\land\psi')}(x)
  &= \left\{
    \begin{array}{ll}
      \lambda_{\alpha\fatsemi_f\beta}(x) & \text{if}~ \lambda_{\alpha\fatsemi_f\beta}(x) \notin \test
      \\
      \assume \lambda_{\alpha\fatsemi_f\beta}(x) & \text{if}~ \lambda_{\alpha\fatsemi_f\beta}(x) \in \test ~\text{and}~ \psi\land\psi'\Rightarrow x
      \\
      \assume{\lnot \lambda_{\alpha\fatsemi_f\beta}(x)} & \text{if}~ \lambda_{\alpha\fatsemi_f\beta}(x) \in \test ~\text{and}~ \psi\land\psi'\Rightarrow \lnot x
    \end{array}
  \right.
  \\
  &= \left\{
    \begin{array}{ll}
      \lambda_{\alpha}(x) & \text{if}~ \lambda_{\alpha}(x) \notin \test
      \\
      \assume \lambda_{\alpha}(x) & \text{if}~ \lambda_{\alpha}(x) \in \test ~\text{and}~ \psi\Rightarrow x
      \\
      \assume{\lnot \lambda_{\alpha}(x)} & \text{if}~ \lambda_{\alpha}(x) \in \test ~\text{and}~ \psi\Rightarrow \lnot x
    \end{array}
  \right.
  \\
  &= \lambda_{\tr_\lpo(\alpha, \psi)}(x)
  \\
  &= \lambda_{\tr_\lpo(\alpha, \psi) \fatsemi \tr_\lpo(\beta,\psi')}(x)
\end{align*}
The case for $x \in f(\psi)$ is nearly identical.
\end{proof}

\begin{lemma}\label{lem:tr-seq}
$
  \tr(\Alpha \fatsemi \Beta) = \{ \Alpha' \fatsemi \Beta' \mid \Alpha' \in \tr(\Alpha), \Beta'\in \tr(\Beta) \}
$.
\end{lemma}
\begin{proof}
\begin{align*}
  \tr(\Alpha \fatsemi \Beta)
  &= \tr( \sup_{\Alpha' \ll \Alpha} \sup_{\Beta' \ll \Beta} \Alpha'\fatsemi \Beta')
  \\
  &= \bigcup_{\Alpha' \ll \Alpha} \bigcup_{\Beta' \ll \Beta}\tr_\fin( \Alpha'\fatsemi \Beta')
  \\
  &= \bigcup_{\Alpha' \ll \Alpha} \bigcup_{\Beta' \ll \Beta}\tr_\fin(\{ \alpha \fatsemi_f \beta \mid \alpha \in \Alpha', \beta\in \Beta', f\in\copylpo_{\alpha,\beta}\})
  \intertext{Since the above set is an equivalence class, we can fix any $\alpha\in\Alpha'$, $\beta \in \Beta'$, and $f\in\copylpo_{\alpha,\beta}$.}
  &= \bigcup_{\Alpha' \ll \Alpha} \bigcup_{\Beta' \ll \Beta}\tr_\fin([\alpha \fatsemi_f \beta])
  \\
  &= \bigcup_{\Alpha' \ll \Alpha} \bigcup_{\Beta' \ll \Beta} \{ [\tr_\lpo(\alpha \fatsemi_f \beta, \psi)] \mid \psi\in \br_{\alpha \fatsemi_f \beta} \}
  \\
  &= \bigcup_{\Alpha' \ll \Alpha} \bigcup_{\Beta' \ll \Beta} \{ [\tr_\lpo(\alpha \fatsemi_f \beta, \psi\land\psi')] \mid \psi\in \br_{\alpha}, \psi'\in\br_{f(\psi)} \}
\\
  &= \bigcup_{\Alpha' \ll \Alpha} \bigcup_{\Beta' \ll \Beta} \{ [\tr_\lpo(\alpha, \psi) \fatsemi \tr_\lpo(f(\psi), \psi')] \mid \psi\in \br_{\alpha}, \psi'\in\br_{f(\psi)} \}
  & \text{(By \Cref{lem:tr-lpo-seq})}
  \\
  &= \bigcup_{\Alpha' \ll \Alpha} \bigcup_{\Beta'\ll\Beta} \tr_\fin(\Alpha') \fatsemi \tr_\fin(\Beta') 
  \\
  &= \left( \bigcup_{\Alpha' \ll \Alpha} \tr_\fin(\Alpha')\right) \fatsemi \left( \bigcup_{\Beta'\ll\Beta} \tr_\fin(\Beta') \right)
  \\
  &= \tr(\Alpha) \fatsemi \tr(\Beta)
&&\qedhere
\end{align*}
\end{proof}

\begin{lemma}\label{lem:tr-guard-lpo}
$
  \tr_\lpo(\guard(x, b, \alpha, \beta), \psi \land x) = \singleton{\assume b}_x \fatsemi \tr_\lpo(\alpha, \psi)
$
and
$
  \tr_\lpo(\guard(x, b, \alpha, \beta), \psi \land \lnot x) = \singleton{\assume{\lnot b}}_x \fatsemi \tr_\lpo(\beta, \psi)
$.
\end{lemma}
\begin{proof}
We prove the equality for the first case. The second is entirely symmetrical. We first show equality of the node sets:
\begin{align*}
  N_{\tr_\lpo(\guard(x, b, \alpha, \beta), \psi \land x)}
  &= \{ y \in N_{\guard(x, b, \alpha, \beta)} \mid \psi \land x \Rightarrow \varphi_{\guard(x, b, \alpha, \beta)}(y) \}
  \intertext{By the construction of guard, $x$ is true in all nodes in the $\alpha$ branch, and the root $x$ has formula $\tru$.}
  &= \{x\} \cup \{ y \in N_\alpha \mid \psi \Rightarrow \varphi_\alpha(y) \}
  \\
  &= N_{\singleton{\assume b}_x} \cup N_{\tr(\alpha, \psi)}
  \\
  &= N_{\singleton{\assume b}_x \fatsemi \tr(\alpha, \psi)}
\end{align*}
For the order, we have:
\begin{align*}
  &\mathord{<_{\tr_\lpo(\guard(x, b, \alpha, \beta), \psi \land x)}}
  \\
  &= \mathord{<_{\guard(x, b, \alpha, \beta)}} \cap (N_{\tr_\lpo(\guard(x, b, \alpha, \beta), \psi \land x)} \times N_{\tr_\lpo(\guard(x, b, \alpha, \beta), \psi \land x)})
  \\
  &= \left( \mathord{<_\alpha} \cup \mathord{<_\beta} \cup (\{x\} \times (N_\alpha \cup N_\beta) \right)  \cap (N_{\tr_\lpo(\guard(x, b, \alpha, \beta), \psi \land x)} \times N_{\tr_\lpo(\guard(x, b, \alpha, \beta), \psi \land x)})
  \intertext{Using what we just showed about the node sets:}
  &= \left( \mathord{<_\alpha} \cap (N_{\tr(\alpha, \psi)} \times N_{\tr(\alpha, \psi)} ) \right) \cup
      (\{x\} \times N_{\tr(\alpha, \psi)})
  \\
  &= \mathord{<_{\tr(\alpha, \psi)}}\cup (\{x\} \times N_{\tr(\alpha, \psi)})
  \\
  &= <_{\singleton{\assume b}_x \fatsemi \tr(\alpha, \psi)}
\end{align*}
For every $y$, clearly:
\[
  \varphi_{ \tr_\lpo(\guard(x, b, \alpha, \beta), \psi \land x)}(y)
  = \varphi_{\singleton{\assume b}_x \fatsemi \tr(\alpha, \psi)}(y)
  = \tru 
\]
If $y=x$, then we have:
\begin{align*}
  &\lambda_{\tr_\lpo(\guard(x, b, \alpha, \beta), \psi \land x)}(y)
  \\
  &= \left\{
    \begin{array}{ll}
      \lambda_{\guard(x, b, \alpha, \beta)}(y) & \text{if}~ \lambda_{\guard(x, b, \alpha, \beta)}(y) \notin \test
      \\
      \assume \lambda_{\guard(x, b, \alpha, \beta)}(y) & \text{if}~ \lambda_{\guard(x, b, \alpha, \beta)}(y) \in \test ~\text{and}~ \psi\land x\Rightarrow y
      \\
      \assume{\lnot \lambda_{\guard(x, b, \alpha, \beta)}(y)} & \text{if}~ \lambda_{\guard(x, b, \alpha, \beta)}(y) \in \test ~\text{and}~ \psi\land x\Rightarrow \lnot y
    \end{array}
  \right.
  \\
  &= \assume \lambda_{\guard(x, b, \alpha, \beta)}(x)
  \\
  &= \assume b
  \\
  &= \lambda_{\singleton{\assume b}_x}(x)
  \\
  &= \lambda_{\singleton{\assume b}_x \fatsemi \tr(\alpha,\psi)}(x)
\end{align*}
If not, then $y \in \tr(\alpha,\psi) \subseteq N_\alpha$:
\begin{align*}
  &\lambda_{\tr_\lpo(\guard(x, b, \alpha, \beta), \psi \land x)}(y)
  \\
  &= \left\{
    \begin{array}{ll}
      \lambda_{\guard(x, b, \alpha, \beta)}(y) & \text{if}~ \lambda_{\guard(x, b, \alpha, \beta)}(y) \notin \test
      \\
      \assume \lambda_{\guard(x, b, \alpha, \beta)}(y) & \text{if}~ \lambda_{\guard(x, b, \alpha, \beta)}(y) \in \test ~\text{and}~ \psi\land x\Rightarrow y
      \\
      \assume{\lnot \lambda_{\guard(x, b, \alpha, \beta)}(y)} & \text{if}~ \lambda_{\guard(x, b, \alpha, \beta)}(y) \in \test ~\text{and}~ \psi\land x\Rightarrow \lnot y
    \end{array}
  \right.
  \\
  &= \left\{
    \begin{array}{ll}
      \lambda_{\alpha}(y) & \text{if}~ \lambda_{\alpha}(y) \notin \test
      \\
      \assume \lambda_{\alpha}(y) & \text{if}~ \lambda_{\alpha}(y) \in \test ~\text{and}~ \psi\Rightarrow y
      \\
      \assume{\lnot \lambda_{\alpha}(y)} & \text{if}~ \lambda_{\alpha}(y) \in \test ~\text{and}~ \psi\Rightarrow \lnot y
    \end{array}
  \right.
  \\
  &= \lambda_{\tr(\alpha, \psi)}(y)
  \\
  &= \lambda_{\singleton{\assume b}_x \fatsemi \tr(\alpha,\psi)}(y)
\qedhere
\end{align*}
\end{proof}

\begin{lemma}\label{lem:tr-guard}
\[\tr(\guard(b, \Alpha, \Beta))
  =
  \{ \singleton{\assume b} \fatsemi \Alpha' \mid \Alpha' \in \tr(\Alpha) \} \cup \{ \singleton{\assume{\lnot b}} \fatsemi \Beta' \mid \Beta'\in \tr(\Beta) \}
\]
\end{lemma}
\begin{proof}
\begin{align*}
  &\tr(\guard(b, \Alpha, \Beta))
  \\
  &= \tr\left(\sup_{\Alpha' \ll \Alpha} \sup_{\Beta' \ll \Beta} \guard(b, \Alpha', \Beta')\right)
  \\
  &= \bigcup_{\Alpha' \ll \Alpha} \bigcup_{\Beta' \ll \Beta} \tr_\fin(\guard(b, \Alpha', \Beta'))
  \\
  &= \bigcup_{\Alpha' \ll \Alpha} \bigcup_{\Beta' \ll \Beta} \tr_\fin(\{ \guard(x, b, \alpha, \beta) \mid \alpha \in \Alpha', \beta\in\Beta', x\in \Nodes \})
  \intertext{Since the set above is an equivalence class, we can fix any $x$, $\alpha$, and $\beta$ and write:}
  &= \bigcup_{\Alpha' \ll \Alpha} \bigcup_{\Beta' \ll \Beta} \tr_\fin([\guard(x, b, \alpha, \beta)])
  \\
  &= \bigcup_{\Alpha' \ll \Alpha} \bigcup_{\Beta' \ll \Beta} \{ [\tr_\lpo(\guard(x, b, \alpha, \beta), \psi)] \mid \psi \in \br_{\guard(x, b, \alpha, \beta)} \}
  \\
  &= \bigcup_{\Alpha' \ll \Alpha} \bigcup_{\Beta' \ll \Beta}
      \{ [\tr_\lpo(\guard(x, b, \alpha, \beta), \psi \land x)] \mid \psi \in \br_{\alpha} \} \cup
      \{ [\tr_\lpo(\guard(x, b, \alpha, \beta), \psi \land \lnot x)] \mid \psi \in \br_{\beta} \}
  \intertext{By \Cref{lem:tr-guard-lpo}:}
  &= \bigcup_{\Alpha' \ll \Alpha} \bigcup_{\Beta' \ll \Beta}
      \{ [\singleton{\assume b}_x \fatsemi \tr_\lpo(\alpha, \psi)] \mid \psi \in \br_{\alpha} \} \cup
      \{ [\singleton{\assume{\lnot b}}_x \fatsemi \tr_\lpo(\beta, \psi)] \mid \psi \in \br_{\beta} \}
  \\
  &= \bigcup_{\Alpha' \ll \Alpha} \bigcup_{\Beta' \ll \Beta}
      \{ \singleton{\assume b} \fatsemi \Alpha'' \mid \Alpha'' \in \tr_\fin(\Alpha') \}  \cup
      \{ \singleton{\assume{\lnot b}} \fatsemi \Beta'' \mid \Beta'' \in \tr_\fin(\Beta') \}
  \\ 
  &= \{ \singleton{\assume b} \fatsemi \Alpha'' \mid \Alpha'' \in \tr(\Alpha) \}  \cup
      \{ \singleton{\assume{\lnot b}} \fatsemi \Beta'' \mid \Beta'' \in \tr(\Beta) \} 
\qedhere
\end{align*}
\end{proof}

\begin{lemma}\label{lem:tr-par-lpo}
For any $\alpha,\beta\in\lpo_\fin$, $\psi \in \br_\alpha$, and $\psi'\in\br_\beta$:
\[
  \tr_\lpo(\alpha\parallel_x \beta, \psi\land\psi')
  = \tr_\lpo(\alpha, \psi) \parallel_x  \tr_\lpo(\beta, \psi')
\]
\end{lemma}
\begin{proof}
We show the equality for each component. First, for the nodes, we have:
\begin{align*}
  N_{\tr_\lpo(\alpha\parallel_x \beta, \psi\land\psi')}
  &= \{ y \in N_{\alpha\parallel_x \beta} \mid \psi\land\psi' \Rightarrow \varphi_{\alpha\parallel_x \beta}(y) \}
  \\
  &= \{ x \} \cup \{ y \in \nFork_\alpha \mid \psi \Rightarrow \varphi_\alpha(y) \} \cup \{ z \in \nFork_\beta \mid \psi' \Rightarrow \varphi_\beta(z) \}
  \\
  &= \{ x \} \cup \nFork_{\tr_\lpo(\alpha, \psi)} \cup \nFork_{\tr_\lpo(\beta, \psi') }
  \\
  &= N_{\tr_\lpo(\alpha, \psi) \parallel_x  \tr_\lpo(\beta, \psi')}
\end{align*}
Now, for the orders:
\begin{align*}
  &\mathord{<_{\tr_\lpo(\alpha\parallel_x \beta, \psi\land\psi')}}
  \\
  &= \mathord{<_{\alpha\parallel_x \beta}} \cap \left( N_{\tr_\lpo(\alpha\parallel_x \beta, \psi\land\psi')} \times N_{\tr_\lpo(\alpha\parallel_x \beta, \psi\land\psi')} \right)
  \\
   &= \left( (\{x\} \times (\nFork_\alpha \cup \nFork_\beta))
        \cup\ (\mathord{<_\alpha} \cap (\nFork_\alpha \times \nFork_\alpha))
        \cup\  (\mathord{<_\beta} \cap (\nFork_\beta \times \nFork_\beta))
        \right)
    \\ &\qquad \cap \left( N_{\tr_\lpo(\alpha\parallel_x \beta, \psi\land\psi')} \times N_{\tr_\lpo(\alpha\parallel_x \beta, \psi\land\psi')} \right)
  \\
   &= (\{x\} \times (\nFork_{\tr_\lpo(\alpha,\psi)} \cup \nFork_{\tr_\lpo(\beta,\psi')} )) 
        \\&\quad \cup\ (\mathord{<_\alpha} \cap (\nFork_{\tr_\lpo(\alpha,\psi)} \times \nFork_{\tr_\lpo(\alpha,\psi)} ))
        \\&\quad \cup\  (\mathord{<_\beta} \cap (\nFork_{\tr_\lpo(\beta,\psi')} \times \nFork_{\tr_\lpo(\beta,\psi')} ) 
  \\
  &= \mathord{<_{\tr_\lpo(\alpha, \psi) \parallel_x  \tr_\lpo(\beta, \psi')}}
\end{align*}
Now take any $y$, clearly we have:
\[
  \varphi_{\tr_\lpo(\alpha\parallel_x \beta, \psi\land\psi')}(y)
  = \varphi_{\tr_\lpo(\alpha, \psi) \parallel_x  \tr_\lpo(\beta, \psi')}(y)
  = \tru
\]
Now, if $y = x$, then $\lambda_{\alpha\parallel_x \beta}(y) = \fork$, so:
\begin{align*}
  \lambda_{\tr_\lpo(\alpha\parallel_x \beta, \psi\land\psi')}(y)
  &= \left\{
    \begin{array}{ll}
      \lambda_{\alpha\parallel_x \beta}(y) & \text{if}~ \lambda_{\alpha\parallel_x \beta}(y) \notin \test
      \\
      \assume \lambda_{\alpha\parallel_x \beta}(y) & \text{if}~ \lambda_{\alpha\parallel_x \beta}(y) \in \test ~\text{and}~ \psi\land \psi'\Rightarrow x
      \\
      \assume{\lnot \lambda_{\alpha\parallel_x \beta}(y)} & \text{if}~ \lambda_{\alpha\parallel_x \beta}(y) \in \test ~\text{and}~ \psi\land\psi'\Rightarrow \lnot x
    \end{array}
  \right.
  \\
  &= \lambda_{\alpha\parallel_x \beta}(y)
  \\
  &= \fork
  \\
  &= \lambda_{\tr_\lpo(\alpha, \psi) \parallel_x  \tr_\lpo(\beta, \psi')}(y)
\end{align*}
Now, if $y \in \nFork_\alpha$, then:
\begin{align*}
  \lambda_{\tr_\lpo(\alpha\parallel_x \beta, \psi\land\psi')}(y)
  &= \left\{
    \begin{array}{ll}
      \lambda_{\alpha\parallel_x \beta}(y) & \text{if}~ \lambda_{\alpha\parallel_x \beta}(y) \notin \test
      \\
      \assume \lambda_{\alpha\parallel_x \beta}(y) & \text{if}~ \lambda_{\alpha\parallel_x \beta}(y) \in \test ~\text{and}~ \psi\land \psi'\Rightarrow x
      \\
      \assume{\lnot \lambda_{\alpha\parallel_x \beta}(y)} & \text{if}~ \lambda_{\alpha\parallel_x \beta}(y) \in \test ~\text{and}~ \psi\land\psi'\Rightarrow \lnot x
    \end{array}
  \right.
  \\
  &= \left\{
    \begin{array}{ll}
      \lambda_{\alpha}(y) & \text{if}~ \lambda_{\alpha}(y) \notin \test
      \\
      \assume \lambda_{\alpha}(y) & \text{if}~ \lambda_{\alpha}(y) \in \test ~\text{and}~ \psi\Rightarrow x
      \\
      \assume{\lnot \lambda_{\alpha}(y)} & \text{if}~ \lambda_{\alpha}(y) \in \test ~\text{and}~ \psi\Rightarrow \lnot x
    \end{array}
  \right.
  \\
  &= \lambda_{\tr(\alpha, \psi)}(y)
  \\
  &= \lambda_{\tr_\lpo(\alpha, \psi) \parallel_x  \tr_\lpo(\beta, \psi')}(y)
\end{align*}
The case where $y \in \nBot(\beta)$ is nearly identical.
\end{proof}

Since $\parallel$ is not monotone (see \Cref{rem:parNotMon}), we need a different way to approximate infinite pomsets composed in parallel, which we now describe. Let $\mathsf{root}(\GGamma)$ be the label of the root node of the pomset $\GGamma$ and
\[
\GGamma \ll_1 \GGamma'
\qquad \text{iff} \qquad
\GGamma \ll \GGamma'
\quad\text{and}\quad
\mathsf{root}(\GGamma) = \mathsf{root}(\GGamma') = \fork 
~\text{or}~
\mathsf{root}(\GGamma') \neq \fork
\]
\ie $\GGamma \ll_1\GGamma'$ is similar to the standard approximation order $\ll$ on pomsets, but it prevents a fork node at the root of $\GGamma'$ from being converted to $\bot$, which would cause the root of a parallel composition involving $\GGamma'$ vs $\GGamma$ to have a different number of successors.
Coming back to the example of \Cref{rem:parNotMon}, we have that $[\alpha] \lepom [\beta]$ and also $[\alpha] \ll [\beta]$, whereas $[\alpha] \not\ll_1 [\beta]$, since the root of $\beta$ is $\fork$ but not so is the root of $\alpha$.

We need two properties of $\ll_1$. First:
\begin{equation}
\label{eq:llOneOne}
\sup_{\GGamma \ll_1 \GGamma'} \GGamma = \sup_{\GGamma \ll \GGamma'} \GGamma 
\end{equation}
By \Cref{lem:fin-approx} and \Cref{cor:lefin-approx}, we have that $\sup_{\GGamma \ll \GGamma'} \GGamma  = \GGamma'$. If $\mathsf{root}(\GGamma') = \fork$, then $\{\GGamma \mid \GGamma \ll_1 \GGamma'\} = \{\GGamma \mid \GGamma \ll \GGamma'\} \setminus\{\bot_\pom\}$, which clearly has the same supremum, as we have only removed the bottom element from the set. If $\mathsf{root}(\GGamma') \neq \fork$, then the orders are equivalent.
The second property that we need is 
\begin{equation}
\label{eq:llOneTwo}
\finapprox{\Alpha \parallel \Beta} \setminus \{\bot_\pom\} = \{ \Alpha' \parallel \Beta' \mid \Alpha' \ll_1 \Alpha, \Beta' \ll_1 \Beta \}
\end{equation}
This can be easily seen by observing that $\bot_\pom$ is not included in the right set, since all parallel compositions have fork at the root therein. Moreover, the successors nodes (but not their labels) of the root are also fixed in both cases: on the left they are fixed because the root is fixed, and on the right they are fixed due to the $\ll_1$ order, since any $\fork$ node at the root would be removed by the parallel composition. Above level 2, both sets allow for any approximation of $\Alpha$ and $\Beta$.

\begin{lemma}\label{lem:tr-par}
$
  \tr(\Alpha \parallel \Beta)
  =
  \{ \Alpha' \parallel \Beta' \mid \Alpha' \in \tr(\Alpha), \Beta'\in \tr(\Beta) \}
$.
\end{lemma}
\begin{proof}
\begin{align*}
  \tr(\Alpha \parallel \Beta)
  &= \tr\left( \sup_{\GGamma \in \finapprox{\Alpha\parallel\Beta}} \GGamma \right) 
  \\
  &= \tr(\singleton\bot) \cup \tr\left( \sup_{\GGamma \in \finapprox{\Alpha\parallel\Beta} \setminus \{\bot_\pom\}} \GGamma \right) 
  \intertext{By \Cref{eq:llOneTwo}.}
  &= \emptyset \cup \tr\left( \sup_{\Alpha' \ll_1 \Alpha} \sup_{\Beta'\ll_1 \Beta} \Alpha' \parallel \Beta'  \right)
  \\
  &= \bigcup_{\Alpha' \ll_1 \Alpha}\bigcup_{\Beta'\ll_1 \Beta} \tr_\fin(\Alpha' \parallel \Beta')
  \\
  &= \bigcup_{\Alpha' \ll_1 \Alpha}\bigcup_{\Beta'\ll_1\Beta} \tr_\fin(\{ \alpha \parallel_x \beta \mid \alpha \in \Alpha', \beta\in \Beta', x\in\Nodes \})
  \\
  &= \bigcup_{\Alpha' \ll_1 \Alpha}\bigcup_{\Beta'\ll_1\Beta} \tr_\fin( [ \alpha \parallel_x \beta ] )
  \\
  &= \bigcup_{\Alpha' \ll_1 \Alpha}\bigcup_{\Beta'\ll_1\Beta} \{ [\tr_\lpo(\alpha \parallel_x \beta, \psi)] \mid \psi \in \br_{\alpha\parallel_x\beta} \}
  \intertext{By \Cref{lem:tr-par-lpo}.}
  &= \bigcup_{\Alpha' \ll_1 \Alpha}\bigcup_{\Beta'\ll_1\Beta} \{ [\tr_\lpo(\alpha \parallel_x \beta, \psi \land \psi')] \mid \psi \in \br_{\alpha}, \psi'\in \br_\beta  \}
  \\
  &= \bigcup_{\Alpha' \ll_1 \Alpha}\bigcup_{\Beta'\ll_1\Beta} \{ [
    \{ \alpha' \parallel_x \beta' \mid \alpha' \in \tr_\lpo(\alpha, \psi), \beta' \in \tr_\lpo(\beta, \psi') \}
  ] \mid \psi \in \br_{\alpha}, \psi'\in \br_\beta  \}
  \\
  &= \bigcup_{\Alpha' \ll_1 \Alpha}\bigcup_{\Beta'\ll_1\Beta}
    \{ \Alpha'' \parallel \Beta'' \mid
      \Alpha'' \in \{ [\tr_\lpo(\alpha, \psi)] \mid \psi \in \br_{\alpha} \},\
      \Beta'' \in \{ [\tr_\lpo(\beta, \psi')] \mid \psi' \in \br_{\beta} \}
    \}
  \\
  &= \bigcup_{\Alpha' \ll_1 \Alpha}\bigcup_{\Beta'\ll_1\Beta}
    \{ \Alpha'' \parallel \Beta'' \mid
      \Alpha'' \in \tr_\fin(\Alpha'),\ 
      \Beta'' \in \tr_\fin(\Beta')
    \}
  \intertext{By \Cref{eq:llOneOne}.}
  &= \{ \Alpha'' \parallel \Beta'' \mid \Alpha'' \in \tr(\Alpha),\ \Beta'' \in \tr(\Beta) \}
&&\qedhere
\end{align*}
\end{proof}

\begin{lemma}\label{lem:tr-while}
If $\tr(\de{C}) = \de{C}_\powdom$, for all $n\in\mathbb N$:
  \[
    \tr( \Phi_{\tuple{C, b}}^n(\bot_\pom) ) = \Xi_{\tuple{C,b}}^n(\emptyset)
  \]
where above we use $\bot_\pom = \singleton{\bot}$ to disambiguate.
\end{lemma}
\begin{proof}
The proof is by induction on $n$. If $n=0$, then we have:
\[
  \tr \left( \Phi_{\tuple{C, b}}^0(\bot_\pom) \right)
  = \tr\left(\bot_\pom\right)
  = \{ [\tr_\lpo(\singleton{\bot}_x, \psi)] \mid \psi \in \br_{\singleton{\bot}_x} \}
  = \emptyset
  = \Xi_{\tuple{C,b}}^0(\emptyset)
\]
Where the second to last step follows from the fact that $\br_{\singleton{\bot}_x} = \emptyset$. Now, suppose the claim holds for $n$, we have:
\begin{align*}
  &\tr( \Phi_{\tuple{C, b}}^{n+1}(\bot_\pom) )
  \\
  &=  \tr\left( \Phi_{\tuple{C, b}} \left(\Phi_{\tuple{C, b}}^{n}(\bot_\pom) \right) \right)
  \\
  &= \tr \left( \guard\left (b, \de{C} \fatsemi \Phi_{\tuple{C,b}}^{n}(\emptyset), \de{\skp} \right) \right)
  \intertext{By \Cref{lem:tr-guard}.}
  &= \left\{ \singleton{\assume b} \fatsemi \Alpha \mid  \Alpha \in \tr\left( \de{C} \fatsemi \Phi_{\tuple{C,b}}^{n}(\emptyset)  \right) \right\}
      \cup \left\{ \singleton{\assume b} \fatsemi \Beta \mid \Beta \in \tr(\de{\skp}) \right\}
  \intertext{By \Cref{lem:tr-singleton,lem:tr-seq}.}
  &= \left\{ \singleton{\assume b} \fatsemi \Alpha \fatsemi \Beta \mid  \Alpha \in \tr(\de{C}), \Beta\in \tr\left(\Phi_{\tuple{C,b}}^{n}(\emptyset)  \right) \right\}
      \cup \left\{ \singleton{\assume b} \fatsemi\de{\skp} \right\}
  \intertext{By assumption, and the induction hypothesis:}
  &= \left\{ \singleton{\assume b} \fatsemi \Alpha \fatsemi \Beta \mid  \Alpha \in \de{C}_\powdom, \Beta\in \Xi_{\tuple{C,b}}^{n}(\emptyset) \right\}
      \cup \left\{ \singleton{\assume b} \fatsemi\de{\skp} \right\}
  \\
  &= \Xi_{\tuple{C, b}}\left( \Xi_{\tuple{C,b}}^{n}(\emptyset) \right)
  \\
  &= \Xi_{\tuple{C,b}}^{n+1}(\emptyset)
\qedhere
\end{align*}
\end{proof}

\subsection{Translations and Program Semantics}

\begin{lemma}\label{lem:tr-sem}
For any program $C\in\cmd$, $\de{C}_\powdom = \tr(\de{C})$.
\end{lemma}
\begin{proof}
The proof is by induction on the structure of the program.
\begin{itemize}
\item $C = \skp$ or $C \in \act$. Follows immediately from \Cref{lem:tr-singleton}.

\item $C = C_1 ; C_2$.
\begin{align*}
  \de{C_1; C_2}_\powdom
  &= \{ \Alpha \fatsemi \Beta \mid \Alpha \in \de{C_1}_\powdom, \Beta\in\de{C_2}_\powdom \}
  \\
  &= \{ \Alpha \fatsemi \Beta \mid \Alpha \in \tr(\de{C_1}), \Beta\in \tr(\de{C_2}) \} && \text{(By the induction hypothesis.)}
  \\
  &= \tr(\de{C_1} \fatsemi \de{C_2}) && \text{(By \Cref{lem:tr-seq}.)}
  \\
  &= \tr(\de{C_1 ; C_2})
\end{align*}

\item $C = C_1 \parop C_2$.
\begin{align*}
  \de{C_1 \parop C_2}_\powdom
  &= \{ \Alpha \parallel \Beta \mid \Alpha \in \de{C_1}_\powdom, \Beta\in\de{C_2}_\powdom \}
  \\
  &= \{ \Alpha \parallel \Beta \mid \Alpha \in \tr(\de{C_1}), \Beta\in \tr(\de{C_2}) \} && \text{(By the induction hypothesis.)}
  \\
  &= \tr(\de{C_1} \parallel \de{C_2}) && \text{(By \Cref{lem:tr-par}.)}
  \\
  &= \tr(\de{C_1 \parop C_2})
\end{align*}

\item $C = \iftf b{C_1}{C_2}$.
\begin{align*}
  &\de{\iftf b{C_1}{C_2}}_\powdom
  \\
  &= \{ \singleton{\assume b} \fatsemi \Alpha \mid \Alpha \in \de{C_1}_\powdom \}
      \cup \{ \singleton{\assume{\lnot b}} \fatsemi \Beta \mid \Beta \in\de{C_2}_\powdom \}
  \intertext{By the induction hypothesis:}
  &= \{ \singleton{\assume b} \fatsemi \Alpha \mid \Alpha \in \tr(\de{C_1}) \}
      \cup \{ \singleton{\assume{\lnot b}} \fatsemi \Beta \mid \Beta \in \tr(\de{C_2}) \}
  \intertext{By \Cref{lem:tr-guard}.}
  &= \tr(\guard(b, \de{C_1}, \de{C_2}))
  \\
  &= \tr(\de{\iftf b{C_1}{C_2}})
\end{align*}

\item $C = \whl b{C'}$.
\begin{align*}
  &\de{\whl b{C'}}_\powdom
  \\
  &= \mathsf{lfp}\left( \Xi_{\tuple{C', b}} \right)
  \\
  &= \bigcup_{n\in\mathbb N} \Xi_{\tuple{C', b}}^n(\emptyset)
  && \text{(By the Kleene fixed point theorem.)}  
  \\
  &= \bigcup_{n\in\mathbb N} \tr\left( \Phi_{\tuple{C', b}}^n(\bot_\pom) \right)
  && \text{(By \Cref{lem:tr-while}, and the induction hypothesis.)}
  \\
  &= \tr\left( \sup_{n\in\mathbb N} \Phi_{\tuple{C', b}}^n(\bot_\pom) \right)
  && \text{(By Scott continuity of $\tr$.)}  
  \\
  &= \tr\left( \mathsf{lfp}\left( \Phi_{\tuple{C', b}}^n \right) \right)
  && \text{(By the Kleene fixed point theorem.)}  
  \\
  &= \tr\left( \de{\whl b{C'}} \right)
&&\qedhere
\end{align*}
\end{itemize}
\end{proof}

\subsection{Translations and Linearization}

\begin{lemma}\label{lem:tr-next}
For any binary branching $\alpha\in\lpo_\fin(\lab)$, $S \subseteq \nBot_\alpha$, and $\psi\in\valid_\alpha(S)$, then:
\[
  \bigcup_{\psi'\in\br_\alpha : \psi'\Rightarrow \psi} \next(\tr_\lpo(\alpha, \psi'), \tru, S)
  = \next(\alpha, \psi, S) \setminus \{ x \in N_\alpha \mid \varphi_\alpha(x) \Rightarrow \stuck_\alpha \}
\]
\end{lemma}
\begin{proof}
We show the inclusion in both directions.

Take any $x \in \bigcup_{\psi'\in\br_\alpha : \psi'\Rightarrow \psi} \next(\tr_\lpo(\alpha, \psi'), \tru, S)$, so there is a $\psi'\in\br_\alpha$ such that $\psi'\Rightarrow \psi$ and $x \in N_{\tr_\lpo(\alpha, \psi')} \setminus S$ such that $\predplus{x}_{\tr_\lpo(\alpha, \psi')} \subseteq S$. From $x \in N_{\tr_\lpo(\alpha, \psi')} \setminus S$, we get that $x \in N_\alpha\setminus S$ and $\psi' \Rightarrow \varphi_\alpha(x)$. We now show that $\predplus{x}_\alpha \subseteq S$. Take any $y <_\alpha x$, so we know that $\psi' \Rightarrow \varphi_\alpha(x) \Rightarrow \varphi_\alpha(y)$, which means that $y\in N_{\tr_\lpo(\alpha, \psi')}$, and therefore $y <_{\tr_\lpo(\alpha, \psi')} x$, so it must be that $y\in S$.
From $\psi' \Rightarrow \psi$ and $\psi' \Rightarrow \varphi_\alpha(x)$, we know that $\varphi_\alpha(x)\land\psi$ is satisfiable. Since $\psi \in \valid_\alpha(S)$, then $\psi$ must assign a truth value to every test in $S$, and since we just showed that $\predplus{x}_\alpha \subseteq S$, then $\free(\varphi_\alpha(x)) \subseteq S$, and so we must have that $\psi'\Rightarrow \varphi_\alpha(x)$. Therefore, we have established all the conditions of $x\in \next(\alpha, \psi, S)$. Clearly also $\varphi_\alpha(x) \not\Rightarrow \stuck_\alpha$ since $\psi' \Rightarrow \varphi_\alpha(x)$ and $\psi' \Rightarrow \lnot\stuck_\alpha$.

Now take any $x\in \next(\alpha, \psi, S)\setminus \{ x \in N_\alpha \mid \varphi_\alpha(x) \Rightarrow \stuck_\alpha \}$, so $x\in N_\alpha\setminus S$, $\psi\Rightarrow \varphi_\alpha(x)$, $\predplus{x}_\alpha \subseteq S$, and $\varphi_\alpha(x) \not\Rightarrow \stuck_\alpha$. Since $\varphi_\alpha(x) \not\Rightarrow \stuck_\alpha$, then $x \in \extens_\alpha$, and so there is clearly a $\psi'\in\br_\alpha$ such that $\psi' \Rightarrow \varphi_\alpha(x)$. This means that $x \in N_{\tr_\lpo(\alpha, \psi')} \setminus S$. Since by definition $\varphi_{\tr_\lpo(\alpha, \psi')}(x) = \tru$, then clearly $\tru \Rightarrow \varphi_{\tr_\lpo(\alpha, \psi')}(x)$. Finally, we will show that $\predplus{x}_{\tr_\lpo(\alpha, \psi')} \subseteq S$. Take any $y <_{\tr_\lpo(\alpha, \psi')} x$, this means that $y <_\alpha x$, so clearly $y\in S$. Therefore, we have that $x \in \next(\tr_\lpo(\alpha, \psi'), \tru, S) \subseteq \bigcup_{\psi'\in\br_\alpha : \psi'\Rightarrow \psi} \next(\tr_\lpo(\alpha, \psi'), \tru, S)$. 
\end{proof}

\begin{lemma}\label{lem:lin-lpo-stuck}
If $S\subseteq \nBot_\alpha$, $\psi \in \valid_\alpha(S)$ and $\psi \Rightarrow \stuck_\alpha$, then $\lin_\lpo(\alpha, \psi, S) = \emptyset$.
\end{lemma}
\begin{proof}
The proof is by induction on the size of the set $\next^*(\alpha, \psi, S)$. If $\next^*(\alpha, \psi, S) = \emptyset$, then by \Cref{lem:nextstar-empty} we know that $\psi \in \br_\alpha$, but this is impossible since $\psi \Rightarrow \stuck_\alpha$. Therefore, it cannot be the case that $\next^*(\alpha, \psi, S) = \emptyset$.

Now suppose that $\next^*(\alpha, \psi, S) \neq \emptyset$. By \Cref{lem:nextstar-nonempty}, we know that $\next(\alpha, \psi, S) \neq \emptyset$. So, we have:
\begin{align*}
  \linlpo(\alpha, \psi, S)
  &= \bigcup_{x \in \next(\alpha, \psi, S)} \left\{
      \begin{array}{ll}
        \linlpo(\alpha,\psi, S \cup\{ x \})^\dagger(\de{\lambda_\alpha(x)}_\act(s)) & \text{if} ~ \lambda_\alpha(x) \in\act
        \\
        \linlpo(\alpha,\psi \land \sem{x = {\de{\lambda_\alpha(x)}}_\test(s)}, S \cup\{x\})(s) & \text{if}~ \lambda_\alpha(x) \in\test
        \\
        \emptyset & \text{if}~ \lambda_\alpha(x) = \bot
        \\
        \linlpo(\alpha,\psi, S\cup\{x\})(s) & \text{if} ~ \lambda_\alpha(x) = \fork
      \end{array}
    \right.
    \intertext{By the induction hypothesis in the first, second, and fourth cases, we get:}
  &= \bigcup_{x \in \next(\alpha, \psi, S)} \left\{
      \begin{array}{ll}
        \emptyset & \text{if} ~ \lambda_\alpha(x) \in\act
        \\
        \emptyset & \text{if}~ \lambda_\alpha(x) \in\test
        \\
        \emptyset & \text{if}~ \lambda_\alpha(x) = \bot
        \\
        \emptyset & \text{if} ~ \lambda_\alpha(x) = \fork
      \end{array}
    \right.
  \\
  &= \bigcup_{x \in \next(\alpha, \psi, S)} \emptyset
  = \emptyset
\qedhere
\end{align*}
\end{proof}

\begin{lemma}\label{lem:tr-lpo-lin} For any binary branching $\alpha \in \lpo_\fin(\lab)$, $S \subseteq \nBot_\alpha$, and $\psi \in \valid_\alpha(S)$:
\[
  \bigcup_{\psi' \in\br_\alpha : \psi' \Rightarrow \psi} \lin_\lpo(\tr_\lpo(\alpha, \psi'), \tru, S)(s)
  =
  \lin_\lpo(\alpha, \psi, S)(s)
\]
\end{lemma}
\begin{proof}
  The proof is by induction on the size of the set $\next^*(\alpha, \psi, S)$. If $\next^*(\alpha, \psi, S) = \emptyset$, then clearly $\next(\alpha, \psi, S) = \emptyset$ and by \Cref{lem:nextstar-empty}, we know that $\psi \in \br_\alpha$. That means that the union on the left side is only over a single term, generated from $\psi$. Now, $\tr_\lpo(\alpha, \psi)$ contains all those nodes $x$ such that $\psi \Rightarrow \varphi_\alpha(x)$, but we know from $\next^*(\alpha, \psi, S) = \emptyset$ that all such nodes are already in $S$, therefore by \Cref{lem:tr-next}, $\next(\tr_\lpo(\alpha, \psi), \tru, S) = \emptyset$ too. This gives us
\[
  \bigcup_{\psi' \in\br_\alpha : \psi' \Rightarrow \psi} \lin_\lpo(\tr_\lpo(\alpha, \psi'), \tru, S)(s)
  = \lin_\lpo(\tr_\lpo(\alpha, \psi), \tru, S)(s)
  = \{s\}
  = \lin_\lpo(\alpha, \psi, S)(s)
\]
Now, assume that $\next^*(\alpha, \psi, S) \neq \emptyset$, so by \Cref{lem:nextstar-nonempty}, then $\next(\alpha, \psi, S)\neq\emptyset$.

Suppose that $\varphi_\alpha(x) \Rightarrow \stuck_\alpha$ for some $x\in\next(\alpha, \psi, S)$. This means that $\psi \Rightarrow \varphi_\alpha(x) \Rightarrow  \stuck_\alpha$, so clearly there can be no $\psi \in \br_\alpha$ such that $\psi' \Rightarrow \psi$, so $\bigcup_{\psi' \in\br_\alpha : \psi' \Rightarrow \psi} \lin_\lpo(\tr_\lpo(\alpha, \psi'), \tru, S)(s) = \emptyset$. In addition, by \Cref{lem:lin-lpo-stuck}, $\linlpo(\alpha, \psi, S)(s) = \emptyset$ too.

So, we are left with the case where $\varphi_\alpha(x) \not\Rightarrow\stuck_\alpha$ for all $x\in \next(\alpha, \psi, S)$. Let $\alpha_{\psi'} = \tr_\lpo(\alpha, \psi')$. We now have:
\begin{align*}
  &\bigcup_{\psi'\in\br_\alpha : \psi'\Rightarrow \psi} \linlpo(\alpha_{\psi'}, \tru, S) 
  \\
  &= \bigcup_{\psi'\in\br_\alpha : \psi'\Rightarrow \psi} \bigcup_{x \in \next(\alpha_{\psi'}, \tru, S)} \left\{
  \footnotesize\arraycolsep=2pt
      \begin{array}{ll}
        \linlpo(\alpha_{\psi'},\tru, S \cup\{ x \})^\dagger(\de{\lambda_{\alpha_{\psi'}}(x)}_\act(s)) & \text{if} ~ \lambda_{\alpha_{\psi'}}(x) \in\act_\powdom
        \\
        \linlpo(\alpha_{\psi'},\tru \land \sem{x = {\de{\lambda_{\alpha_{\psi'}}(x)}}_\test(s)}, S \cup\{x\})(s) & \text{if}~ \lambda_{\alpha_{\psi'}}(x) \in\test_\powdom
        \\
        \emptyset & \text{if}~ \lambda_{\alpha_{\psi'}}(x) = \bot
        \\
        \linlpo(\alpha_{\psi'},\tru, S\cup\{x\})(s) & \text{if} ~ \lambda_{\alpha_{\psi'}}(x) = \fork
      \end{array}
    \right.
  \intertext{Since $\alpha_{\psi'}$ has no test or $\bot$ nodes, we can remove those two cases. In addition, we expand the case for actions into cases for whether or not the action is assume.}
  &= \bigcup_{\psi'\in\br_\alpha : \psi'\Rightarrow \psi} \bigcup_{x \in \next(\alpha_{\psi'}, \tru, S)} \left\{
  \footnotesize\arraycolsep=2pt
      \begin{array}{ll}
        \linlpo(\alpha_{\psi'},\tru, S \cup\{ x \})^\dagger(\de{\lambda_{\alpha}(x)}_\act(s)) & \text{if} ~ \lambda_{\alpha}(x) \in\act
        \\
        \linlpo(\alpha_{\psi'},\tru, S \cup\{ x \})^\dagger(\de{\assume{\lambda_\alpha(x)}}_\act(s)) & \text{if} ~ \lambda_{\alpha}(x) \in \test ~\text{and}~ \psi' \Rightarrow x
        \\
        \linlpo(\alpha_{\psi'},\tru, S \cup\{ x \})^\dagger(\de{\assume{\lnot \lambda_\alpha(x)}}_\act(s)) & \text{if} ~ \lambda_{\alpha}(x) \in \test ~\text{and}~ \psi' \Rightarrow \lnot x
        \\
        \linlpo(\alpha_{\psi'},\tru, S\cup\{x\})(s) & \text{if} ~ \lambda_{\alpha}(x) = \fork
      \end{array}
    \right.
  \intertext{By \Cref{lem:tr-next}, and the fact that $\psi \Rightarrow \varphi_\alpha(x)$ for all $x\in \next(\alpha, \psi, S)$, we can rearrange the unions and move the inner one inside of the conditional.}
  &= \bigcup_{x \in \next(\alpha, \psi, S)} \left\{
  \footnotesize\arraycolsep=2pt
      \begin{array}{ll}
        \bigcup_{\psi'\in\br_\alpha : \psi'\Rightarrow \psi}\linlpo(\alpha_{\psi'},\tru, S \cup\{ x \})^\dagger(\de{\lambda_{\alpha}(x)}_\act(s)) & \text{if} ~ \lambda_{\alpha}(x) \in\act
        \\
        \bigcup_{\psi'\in\br_\alpha : \psi'\Rightarrow \psi \land x}\linlpo(\alpha_{\psi'},\tru, S \cup\{ x \})^\dagger(\de{\lambda_\alpha(x)}_\test(s)) & \text{if} ~ \lambda_{\alpha}(x) \in \test
        \\
        \quad \cup \bigcup_{\psi'\in\br_\alpha : \psi'\Rightarrow \psi \land \lnot x}\linlpo(\alpha_{\psi'},\tru, S \cup\{ x \})^\dagger(\de{\lnot \lambda_\alpha(x)}_\test(s))
        \\
        \bigcup_{\psi'\in\br_\alpha : \psi'\Rightarrow \psi}\linlpo(\alpha_{\psi'},\tru, S\cup\{x\})(s) & \text{if} ~ \lambda_{\alpha}(x) = \fork
      \end{array}
    \right.
  \intertext{By the induction hypothesis.}
  &= \bigcup_{x \in \next(\alpha, \psi, S)} \left\{
  \footnotesize\arraycolsep=2pt
      \begin{array}{ll}
        \linlpo(\alpha, \psi, S \cup\{ x \})^\dagger(\de{\lambda_{\alpha}(x)}_\act(s)) & \text{if} ~ \lambda_{\alpha}(x) \in\act
        \\
        \linlpo(\alpha,\psi \land x, S \cup\{ x \})^\dagger(\de{\lambda_\alpha(x)}_\test(s)) & \text{if} ~ \lambda_{\alpha}(x) \in \test
        \\
        \quad \cup \linlpo(\alpha,\psi\land \lnot x, S \cup\{ x \})^\dagger(\de{\lnot \lambda_\alpha(x)}_\test(s))
        \\
        \linlpo(\alpha,\psi, S\cup\{x\})(s) & \text{if} ~ \lambda_{\alpha}(x) = \fork
      \end{array}
    \right.
  \intertext{Exactly one of $\de{\lambda_\alpha(x)}_\test(s)$ and $\de{\lnot\lambda_\alpha(x)}_\test(s)$ will evaluate to $\{s\}$ and the other will evaluate to $\emptyset$, so we can consolidate the terms.}
    &= \bigcup_{x \in \next(\alpha, \psi, S)} \left\{
  \footnotesize\arraycolsep=2pt
      \begin{array}{ll}
        \linlpo(\alpha, \psi, S \cup\{ x \})^\dagger(\de{\lambda_{\alpha}(x)}_\act(s)) & \text{if} ~ \lambda_{\alpha}(x) \in\act
        \\
        \linlpo(\alpha,\psi \land \sem{x = \de{\lambda_\alpha(x)}_\test(s)}, S \cup\{ x \})(s) & \text{if} ~ \lambda_{\alpha}(x) \in \test
        \\
        \linlpo(\alpha,\psi, S\cup\{x\})(s) & \text{if} ~ \lambda_{\alpha}(x) = \fork
      \end{array}
    \right.
  \\
  &= \linlpo(\alpha, \psi, S)
\qedhere
\end{align*}
\end{proof}

\begin{lemma}\label{lem:tr-lin}
For any $\Alpha\in\pom$ and $s\in S$, $\lin(\Alpha)(s) = \lin_\powdom(\tr(\Alpha))(s)$.
\end{lemma}
\begin{proof}
\begin{align*}
  \lin(\Alpha)(s)
  &= \bigcup_{\Alpha' \ll \Alpha} \lin_\fin(\Alpha')(s)
  \\
  &= \bigcup_{[\alpha] \ll \Alpha} \lin_\lpo(\alpha, \tru, \emptyset)(s)
  \\
  &= \bigcup_{[\alpha] \ll \Alpha} \bigcup_{\psi \in\br_\alpha} \lin_\lpo(\tr_\lpo(\alpha, \psi), \tru, \emptyset)(s)
  &   \text{By \Cref{lem:tr-lpo-lin}}
  \\
  &=  \bigcup_{\Alpha' \ll \Alpha} \bigcup_{\Beta \in\tr_\fin(\Alpha')} \lin_\fin(\Beta)(s)
  \\
  &=  \bigcup_{\Beta \in \bigcup_{\Alpha' \ll \Alpha} \tr_\fin(\Alpha')} \lin_\fin(\Beta)(s)
  \\
  &=  \bigcup_{\Beta \in \tr(\Alpha)} \lin_\fin(\Beta)(s)
  \\
  &= \lin_\powdom(\tr(\Alpha))
&&\qedhere
\end{align*}
\end{proof}

\subsection{Main Theorem}

\powDomChar*
\begin{proof}
  It suffices to show that $\de{-}_\powdom = \tr \circ \de{-}$ and $\lin_\powdom \circ \tr = \lin$, which follows immediately from \Cref{lem:tr-sem,lem:tr-lin}. We then have:
  \begin{align*}
    \lin_\powdom \circ \de{-}_\powdom
    &= \lin_\powdom \circ (\tr \circ \de{-})
    = (\lin_\powdom \circ \tr) \circ \de{-}
    = \lin \circ \de{-}
  \end{align*}
\end{proof}
\fi

\end{document}